\documentclass[pdflatex,sn-mathphys]{sn-jnl}
\usepackage{xspace}
\usepackage{url}
\usepackage{longtable}

\usepackage{amssymb}
\usepackage{complexity}
\usepackage{caption}
\usepackage{tikz}
\usepackage{program}
\usepackage{tkz-graph}
\usepackage{tabularx}

\jyear{2023}%

\theoremstyle{thmstyleone}%
\newtheorem{theorem}{Theorem}
\newtheorem{proposition}[theorem]{Proposition}%

\theoremstyle{thmstyletwo}%
\newtheorem{example}{Example}%

\theoremstyle{thmstylethree}%
\newtheorem{definition}{Definition}%

\begin{document}

\title[ ]{Deterministic Scheduling of Periodic Messages for Low Latency in Cloud RAN}


\author[1]{\fnm{Dominique} \sur{Barth}}\email{dominique.barth@uvsq.fr}

\author*[2]{\fnm{Maël} \sur{Guiraud}}\email{mguiraud@cesi.fr}

\author[1]{\fnm{Yann} \sur{Strozecki}}\email{yann.strozecki@uvsq.fr}

\affil[1]{\orgdiv{DAVID laboratory}, \orgname{Université de Versailles Saint-Quentin}, \orgaddress{\street{45 Av. des Etats Unis}, \city{Versailles}, \postcode{78000}, \country{France}}}

\affil[2]{\orgdiv{LINEACT}, \orgname{CESI}, \orgaddress{\street{93 Bd de la Seine}, \city{Nanterre}, \postcode{92000},  \country{France}}}

\newtheorem{lemma}[theorem]{Lemma}

\newcommand{\Fo}{\textsf{FO}} 

\newcommand{\Fmo}{\textsf{FMO}}
\newcommand\pma{\textsc{pma}\xspace}
\newcommand\firstfit{\texttt{First Fit}\xspace}
\newcommand\compactpair{\texttt{Compact Pairs}\xspace}
\newcommand\metaoffset{\texttt{Meta Offset}\xspace}
\newcommand\greedyuniform{\texttt{Greedy Uniform}\xspace}
\newcommand\swapandmove{\texttt{Swap and Move}\xspace}
\newcommand\compactfit{\texttt{Compact Fit}\xspace}
\newcommand\greedypotential{\texttt{Greedy Potential}\xspace}
\newcommand{\todo}[1]{{\color{red} TODO: {#1}}}

\newcommand\shortestlongest{\texttt{ShortestLongest}\xspace}

\newcommand\ESCA{\texttt{ESCA}\xspace}
\newcommand\greedydeadline{\texttt{GreedyDeadline}\xspace}
\newcommand\MLS{\texttt{MLS}\xspace}
\newcommand\PMLS{\texttt{PMLS}\xspace}
\newcommand\SPMLS{\texttt{SPMLS}\xspace}
\newcommand\ASPMLS{\texttt{ASPMLS}\xspace}

\newcommand\FIFO{\texttt{FIFO}\xspace}
\newcommand\framepre{\texttt{FramePreemption}\xspace}
\newcommand\critdead{\texttt{CriticalDeadline}\xspace}

\newcommand\pazl{\textsc{pazl}\xspace}
\newcommand\pall{\textsc{pall}\xspace}
\newcommand\wta{\textsc{wta}\xspace}
\newcommand\pra{\textsc{pra}\xspace}
\newcommand\minpazl{\textsc{minpazl}\xspace}
\newcommand\mintra{\textsc{mintra}\xspace}

\abstract{Cloud-RAN (C-RAN) is a cellular network architecture where processing units, previously attached to antennas, are centralized in data centers. The main challenge in meeting protocol time constraints is minimizing the latency of periodic messages exchanged between antennas and processing units. We demonstrate that statistical multiplexing introduces significant logical latency due to buffering at network nodes to prevent collisions. To address this, we propose a \emph{deterministic} scheme for periodic message transmission \emph{without collisions}, eliminating latency caused by buffering.

We develop several algorithms to compute such schemes for star-routed networks, a common topology where all antennas share a single link. First, we show that deterministic transmission is possible without buffering when routes are short or network load is low. Under high load, we allow buffering at processing units and introduce the \texttt{Periodic Minimal Latency Scheduling} (\PMLS) algorithm, adapted from classical scheduling methods. Experimental results indicate that even at full load, \PMLS finds deterministic transmission schemes with negligible logical latency, whereas statistical multiplexing incurs substantial delays. Moreover, \PMLS runs in polynomial time and scales efficiently to hundreds of antennas. Building on this approach, we also derive low-latency periodic transmission schemes that coexist with additional random network traffic. This article extends previous work presented at ICT~\cite{Guir1806:Deterministic}.}


\maketitle

\section{Introduction}

One of the key aspects of 5G+ and 6G is to reduce the End-to-End delay in telecom networks~\cite{dahlman20185g,saad2019vision}. This goal is driven by the need for Ultra-Reliable Low Latency~\cite{ali2021urllc} for several use cases such as automation of industry, telesurgery, intelligent transportation, augmented/virtual reality, and many others~\cite{chen2018ultra,nguyen20216g}. In this article, we focus on low latency for Cloud Radio Access Network or C-RAN. C-RAN has been proposed for 5G+ and 6G~\cite{niknam2020intelligent,larsen2022deployment} as a next-generation mobile network architecture to reduce energy consumption~\cite{gavrilovska2020cloud,mobile2011c,checko2014cloud} and more generally the total cost of ownership. C-RAN is a centralized architecture: each antenna has a Remote Radio Head (RRH) that sends the signal to a Baseband Unit (BBU) in a data center\footnote{Other terminologies exist in the literature. The results of this work are fully compatible with any variation of the C-RAN architecture.}. 
A primary challenge in C-RAN is achieving latency compatible with transport protocols~\cite{ieeep802}. Latency is measured from the moment an RRH sends a message to the reception of the response, computed by real-time virtualized network functions on a BBU. For example, LTE standards require one to process functions like HARQ (Hybrid Automatic Repeat reQuest) in $3$~ms~\cite{bouguen2012lte}. In 5G, some services need end-to-end latency as low as $1$~ms~\cite{dogra2020survey,3gpp5g,boccardi2014five}. Beyond latency constraints, C-RAN also imposes periodic data transfer in the \emph{fronthaul network} between RRHs and BBUs, requiring frames to be transmitted every millisecond in line with 5G+ and 6G specifications~\cite{bouguen2012lte,romano2019imt}. We aim to operate a C-RAN on a low-cost shared switched packet network.
This work investigates whether periodic messages can be scheduled to avoid collisions in a low-cost shared switched packet network. Eliminating queuing delays frees time for latency caused by physical transmission, enabling wider deployment areas.

 We assume that network components can collect information, transmit it to a centralized entity, classify traffic flows, and forward each flow at a precise time dictated by the controller. Time-Sensitive Networking (TSN) standards, such as IEEE 802.1Qat Stream Reservation Protocol~\cite{ieee802qat}, enhanced by IEEE 802.1Qcc~\cite{6755436}, provide technical solutions to support these assumptions, aligning with the Software-Defined Networking (SDN) paradigm~\cite{mohamed2021software,li2015software,7356556}.

We also assume that all RRHs share the same clock and period but are not synchronized for transmissions, meaning each RRH can emit its frame at a different moment within the period. While current cellular networks enforce RRH synchronization~\cite{omri2019synchronization}, this can be delegated to specialized mechanisms~\cite{khalili2016uplink,yemini2016multiple}. In particular, Timing Advance technology, already deployed to ensure synchronized reception despite varying path lengths, can be leveraged for this purpose~\cite{mahmood2019time}.

Our model suggests a design for future cellular networks where RRH transmissions are unsynchronized to improve latency. Furthermore, the results presented in this paper are directly applicable to remote control in Industry 4.0~\cite{peng2021latency,garcia2019latency}, where applications such as real-time industrial automation and extended reality require minimal jitter and near-instantaneous responsiveness~\cite{nikhileswar2022traffic,john2024industry}.

We model the network topology as a directed weighted multigraph composed of directed paths (routes). Each path represents the transmission of a message from an RRH to a BBU and the reception of the response by the RRH. Time is discretized into units called \emph{tics}, where one tic corresponds to the transmission time of a minimal data unit. To achieve optimal latency, we aim to eliminate buffering at internal network nodes (switches). This is feasible due to the deterministic nature of C-RAN messages, or datagrams, whose arrival times at RRHs are known in advance. Following LTE standards~\cite{bouguen2012lte}, we assume all datagram arrivals are periodic within the same period.  Our goal is to design a \emph{periodic} transmission process that avoids collisions. A process is periodic with period $P$ if the network state at times $t$ and $t+P$ is identical.

Each datagram follows a fixed route, and no buffering is allowed within the network. Thus, constructing a periodic sending process, or \emph{assignment}, involves selecting two key parameters: the \emph{offset}, i.e., the time within the period when each RRH sends its datagram, and the \emph{waiting time} at the BBU before the response is transmitted back to the RRH. The periodic nature of the process imposes constraints that make many traditional scheduling methods inapplicable: datagrams must not collide with others in the same or different periods. Minimizing latency—the time from datagram emission to response reception—requires minimizing the only permitted buffering: the waiting time at the BBU, especially for long routes.

We define the problem \pall: given a routed network, a period, message size, and a deadline for each datagram, compute an assignment that meets all deadlines. This article focuses on minimizing the worst-case datagram latency, effectively solving \pall with the tightest possible deadlines.

We also consider \pazl, a variant of \pall where waiting time at BBUs is strictly zero. A solution to \pazl introduces no logical latency and eliminates the need for buffering at BBUs. However, under high network load, \pazl often has no feasible solution. Beyond network load, the number of RRHs is a key parameter influencing the complexity of solving \pall and \pazl.

This article focuses on star-routed networks, a common C-RAN fronthaul topology where all Remote Radio Heads (RRHs) share a central link to a Baseband Unit (BBU)~\cite{electronics9122131,bhattacharjee2020time}. This link mutualization provides a cost-effective alternative to dedicated optical transport. Our algorithms scale to hundreds of RRHs, far exceeding typical star-routed C-RAN deployments, which aggregate only a few dozen RRHs per site. Thus, our approach comfortably meets the scalability demands of latency-critical applications.

To enable deterministic scheduling on star-routed networks, we developed a Hard-TSN switch~\cite{Marc2201:Experimental}, leveraging the absence of buffering to achieve zero jitter. Our prototype demonstrates the feasibility of offset-based scheduling with near-nanosecond gate precision on Ethernet hardware. Several related technologies have been patented~\cite{howe2005time,leclerc2016transmission}.

For more complex topologies and scenarios with synchronized RRHs (studied in~\cite{guiraud2021deterministic}), intermediate buffering becomes necessary, inevitably increasing latency. However, our techniques remain adaptable to ensure deterministic guarantees while also reducing best-effort traffic latency.

 \subsection*{Related works}

Statistical multiplexing is the most common mechanism used to manage packet-based networks. While there are tools~\cite{metricsietf} to ensure a latency lower than a given value for $95\%$ of the packets, such a guarantee is not sufficient when all packets must satisfy latency constraints. Mechanisms like Express Forwarding~\cite{exprforw} can be used to prioritize some packets over others, but they also fail to guarantee the delivery of a given packet in a given time delay when several packets compete for the same resource. 

The current solution to avoid statistical multiplexing in C-RAN consists in using a full optical approach, where each end-point (RRH on one side, BBU on the other side) is connected through direct fiber or full optical switches~\cite{kai2020amplify,tayq2017real}. This architecture is very expensive and not very scalable for a mobile network. As an illustration, a single (one-operator) mobile network in France is composed of about $10,000$ base stations. This number will increase by a factor of $2$ to $20$ with the emergence of “small cells” which increase base station density to reach higher throughput \cite{dahlman20185g,romano2019imt}. This underlines the need to find a low-cost solution to offer low latency over commoditized packet-based networks. 
An alternative approach consists in using optical rings for both fronthaul and best-effort flows~\cite{DBLP:conf/ondm/BarthGS19,rommel2020towards,luu2021dynamic}. This relies on over dimensioning the network and using optical multiplexing technologies to avoid contention, and it requires building new expensive optical networks.

Although 3GPP standards for 6G are not completely frozen yet, the core network is designed to use switched packet networks on Ethernet technology~\cite{ieee1914,gomes2015fronthaul}
contrary to solutions working over dedicated optical networks. The latency induced by contention buffers is one of the major manageable sources of delay in switched packet networks. In this article, we propose to deterministically manage datagrams to deal with contention. This is similar to the concept of Deterministic Networking: the DetNet working group from IETF~\cite{finn-detnet-architecture-08} collaborates with TSN (Time-Sensitive Networking)~\cite{ieee802}, a task group of IEEE, to develop technical solutions for deterministic networking~\cite{bhattacharjee2020time,8613095,durr2016no}. However, these works deal with stochastic flows of traffic while we manage \emph{deterministic and periodic} flows. Hence, the guarantee given by these works is an upper bound on the random latency, while we aim to guarantee the minimal latency by taking advantage of the deterministic nature of our flows. 

In this article, we try to solve the problems \pall and \pazl. They may look like wormhole problems~\cite{ni1993survey,cole1996benefit}, but they require minimizing the time lost in buffers and not only avoiding deadlocks. On a technical aspect, wormhole switches~\cite{cole1996benefit} are designed to read only the header of the messages before forwarding it instead of buffering the entire messages as in store-and-forward~\cite{tindell1992store}. This method has a huge impact on the latency, particularly on long messages, and we go further by trying to remove all buffering in the switches.

 Several graph colorings have been introduced to model problems similar to \pazl and \pall such as the allocation of frequencies~\cite{borndorfer1998frequency}, bandwidths~\cite{erlebach2001complexity}, or routes~\cite{cole1996benefit} in a network. Unfortunately, they do not take into account the periodicity of the scheduling we want to build, even though the associated problems are already $\NP$-complete. The only coloring taking periodicity into account is the circular coloring~\cite{ZHU2001371,zhou2013multiple} but it is not expressive enough to capture our problems. 

The problem \pall on a star routed network is very close to a two flow-shop scheduling problem~\cite{yu2004minimizing} with the additional constraint of periodicity. To our knowledge, all studied periodic scheduling problems are different from \pall. Either the aim is to minimize the number of processors on which the periodic tasks are scheduled~\cite{korst1991periodic,hanen1993cyclic} while our problem corresponds to a single processor and a constraint similar to makespan minimization. Or, in cyclic scheduling~\cite{levner2010complexity}, the aim is to minimize the period of a scheduling to maximize the throughput, while our period is fixed. 

The train timetabling problem~\cite{lusby2011railway} and its restriction, the periodic event scheduling problem~\cite{serafini1989mathematical} are generalizations of our problem. Indeed, they take the period as input and can express the fact that two trains (like two datagrams) should not cross. However, they are much more general: the trains can vary in size and speed, the network can be more complex than a single track and there are precedence constraints. Hence, the numerous variants of train scheduling problems are very hard to solve (and always $\NP$-hard). Thus, delays and speed variation are allowed to make the problem solvable and most of the research done~\cite{lusby2011railway} is devising practical algorithms using branch and bound, mixed integer programming, genetic algorithms\dots

Variations on the problem of scheduling periodic messages for time-sensitive star-shaped networks have been recently investigated in~\cite{9472838,nayak2017incremental,steiner2018traffic,silviu2017,naresh2016}. In these papers, the authors study the practical problem of scheduling a few numbers of datagram flows in a star-shaped network. To do so, linear integer programming is used to compute a schedule of the flows. In the same spirit, the use of an SMT solver rather than linear integer programming is proposed in~\cite{dos2019tsnsched}. The flows described in these papers are somewhat different from ours. While we consider that a single long datagram is sent by a source every period, they consider flows with several little datagrams. In this method, the scheduling is computed over several periods up to some time horizon, while we compute scheduling on a single period, which can be repeated periodically. Furthermore, the experiments are made on small topologies because integer linear programming does not scale with the number of routes~\cite{masoudi2020cost}, while we propose polynomial-time algorithms that give a satisfying solution. However, high complexity and exact algorithm, as explained in~\cite{steiner2018traffic}, can be used as a standardization tool to verify the viability of solutions computed by faster algorithms.

\subsection*{Contributions}

We propose in this paper a method to manage deterministic and periodic flows of C-RAN, which guarantees \textbf{minimal latency and zero jitters}. We introduce several algorithmic problems to capture the search for deterministic and periodic sending schemes with low latency, namely \pazl, \pall, and \wta and we prove that they are \NP-complete on simple topologies of fronthaul networks. We focus on star-routed networks, which represent a fronthaul network where all RRHs share a common link to the data center to lower the cost. 

We give exact algorithms to solve \pazl, \pall, and \wta, with exponential complexity in the number of RRHs in the networks
but not in the other larger parameters of the network. They solve our problems for fewer than twenty RRHs and allow us to validate the other proposed algorithms. 
We show that the load of the star-routed network is a fundamental parameter: when it is low, there is always a solution to the problem
\pazl, that is, without buffering. When it is higher, a buffering in the BBU is necessary and we propose many heuristics to solve the problem \pall. 
Through experiments on many star-routed networks, we show that \PMLS is the heuristic with the best quality, while its complexity is only quadratic.

We propose a method to adapt our periodic sending schemes to the presence of additional stochastic traffic. We show that using stochastic multiplexing for C-RAN does not 
guarantee a good latency, while our solution adds no logical latency to the C-RAN traffic and it even improves the latency of the stochastic traffic.

\subsection*{Outline}
 In Section~\ref{sec:def}, we propose a model of the fronthaul network and the periodic sending of datagrams along its routes. Then, we introduce the problem \pall to formalize the problem of finding a periodic sending of the messages in a network without collision at the contention points. We also present the problem \pazl, a restriction of \pall where no waiting time in the BBU is allowed. We present a simple but very common topology, the star-routed network, with a single shared duplex link, that is studied in the rest of the article.  In Section~\ref{sec:complexity}, we prove that both \pazl and \pall are $\NP$-hard for very restricted classes of graphs and that their optimization counterparts are hard to approximate. 
 In Section~\ref{sec:pazl}, we study the problem \pazl, and several algorithms are proposed: they are polynomial-time for small loads or small routes, or exponential time in the number of routes, based on a compact representation of optimal solutions. We use these algorithms to provide experimental evidence that \pazl can be solved positively when the network is mildly loaded. In Section~\ref{sec:PALL}, we propose polynomial-time heuristics and an exact FPT algorithm for the general \pall problem and experimentally show that they work well, even in extremely loaded networks. 
Finally, in Section~\ref{sec:comparison}, we compare our deterministic approach to stochastic multiplexing, with several buffering policies and with additional random traffic in the network.

\section{Modeling of the Problem}\label{sec:def}

Let $[n]$ denote the interval of $n$ integers $\{0,\dots,n-1\}$.

	\subsection{Routes and Contention Points}

  	We study a communication network comprised of pairs of vertices between which messages are sent periodically. The routing between each pair of such nodes is represented by a \textbf{route}, a sequence of vertices $r=(s, c_1, \ldots, c_l, t)$, such that no vertex appears twice in the sequence (there is no loop in a route). Each vertex $c_i$ represents a contention point, which is the beginning of a link of the communication network shared by several routes. Hence, all vertices of $r$ appear in several routes, except $s$, the first vertex of $r$, and $t$, its last vertex, which are exclusive to $r$ and represent the source and the target of the message. When modeling a C-RAN network, the first vertex represents the sending of the message by the RRH and the last vertex represents the same RRH that receives the answer, sent back by the BBU (see Figure~\ref{fig:routeexample}).

  	The set of routes is denoted by ${\cal R}$. A route is interpreted as a directed path in a directed multigraph constituted of all routes, where the sets of arcs of the routes are disjoint. The routes contain no loop nor cycle since all vertices of a route are different. Thus, the directed multigraph is acyclic. An arc in the multigraph may represent several physical links or nodes of the modeled network if they do not induce contention points.

  	Each arc $(u,v)$ of a route $r$ is labeled by an integer weight $\omega(r,u)$. It represents the time elapsed between the sending of the message of the route $r$ in $u$ and its reception in $v$.

\begin{figure}
\centering

\scalebox{0.5}{

\begin{tikzpicture}
  \SetGraphUnit{5}
    \tikzset{
  EdgeStyle/.append style = {->} }
   \tikzstyle{VertexStyle}=[shape = circle, draw, minimum size = 30pt]

  \node (s2) at (0,2) {\includegraphics[width = 1cm]{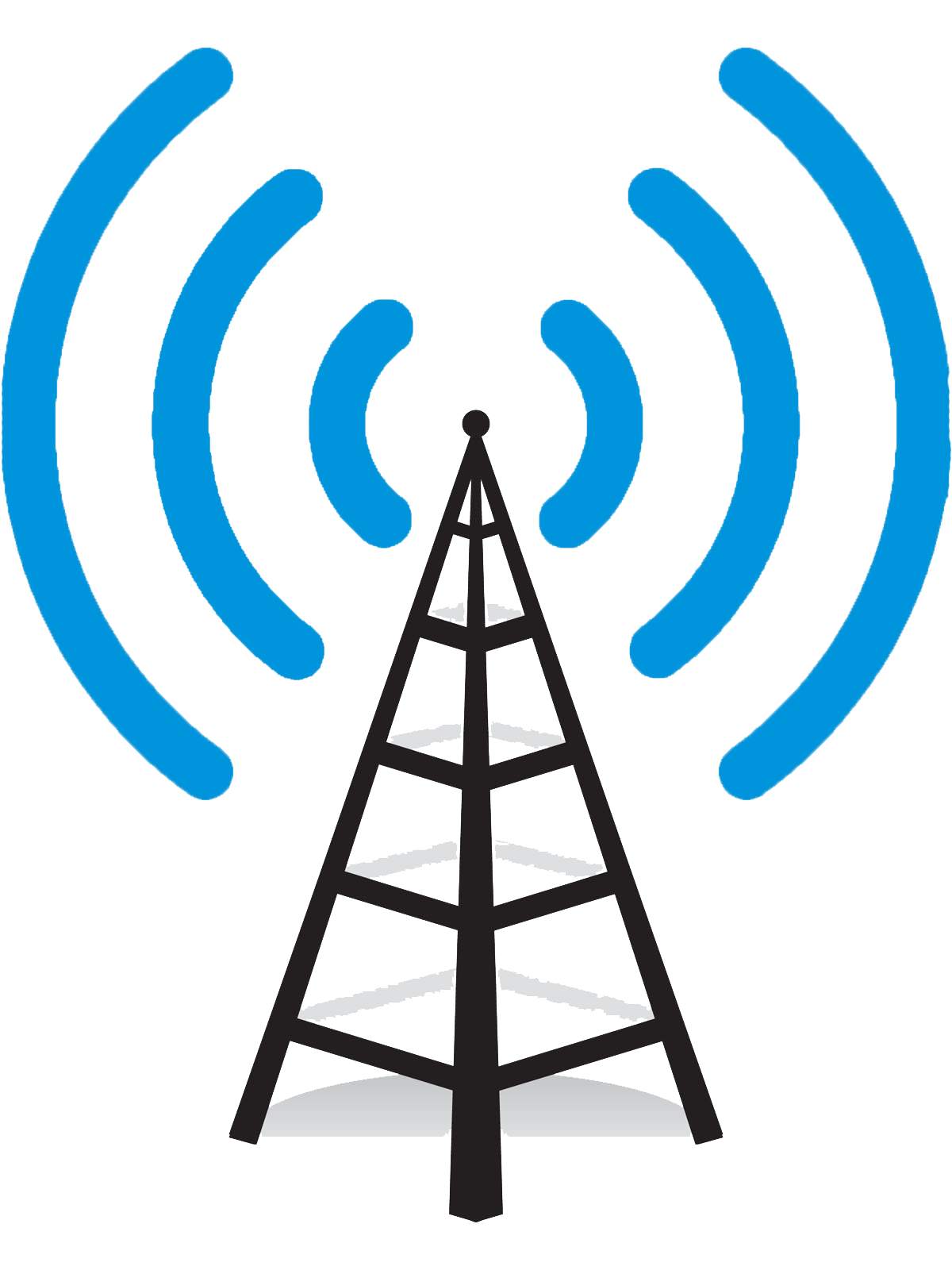}};
  \node[below] at (s2.south) {\huge $s$};

  \node (t2) at (12,2) {\includegraphics[width = 1cm]{rrh.png}};
  \node[below] at (t2.south) {\huge $t$};
    \node (c1) at (4,2) {\includegraphics[width = 1cm]{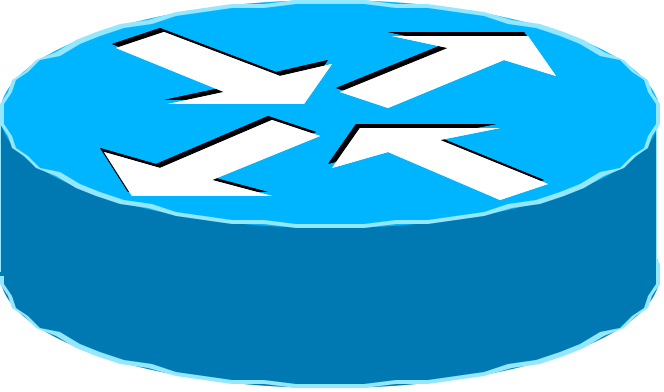}};
 \node[below] at (c1.south) {\huge $c_1$};
    \node (c2) at (8,2) {\includegraphics[width = 1cm]{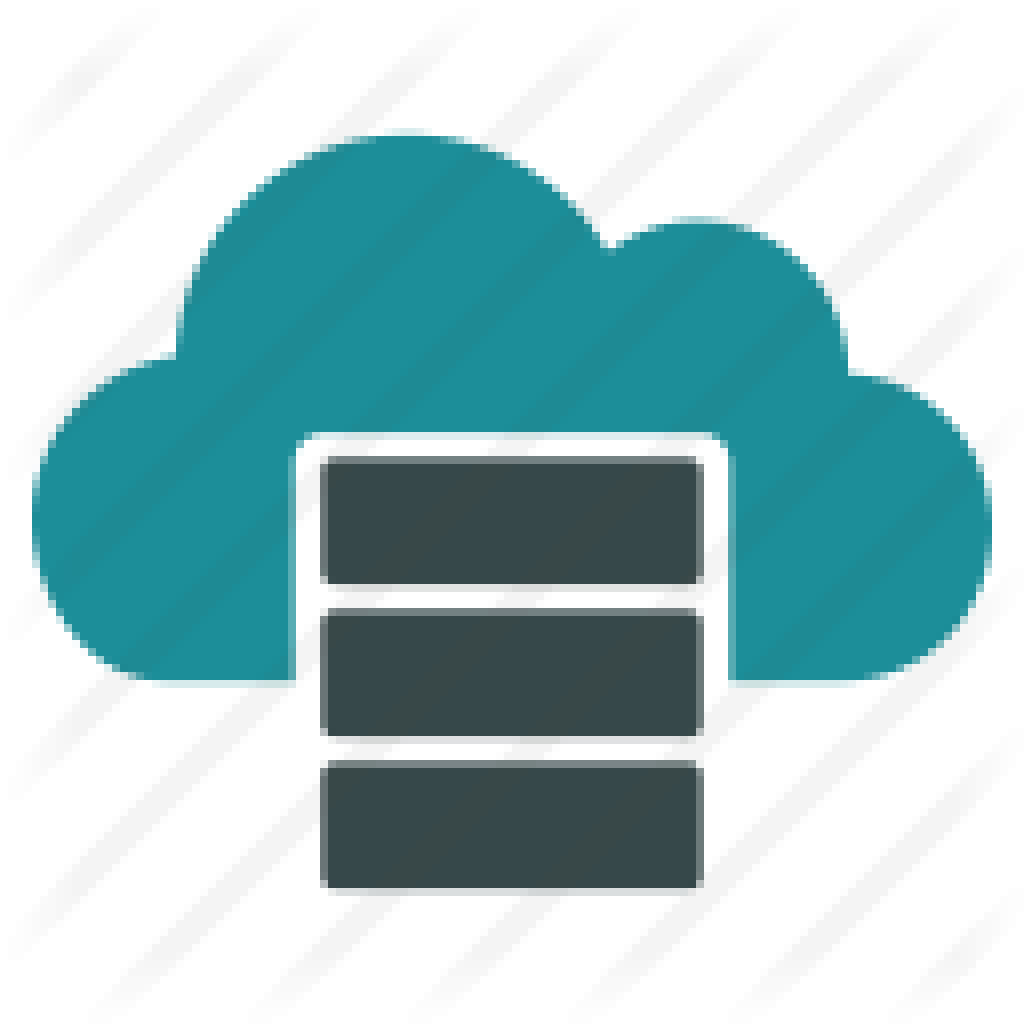}};
  \node[below] at (c2.south) {\huge $c_2$};

\path (s2) edge [->] node[anchor=south,inner sep = 0.2cm]{$3$} (c1);
\path (c1) edge [->] node[anchor=south,inner sep = 0.2cm]{$7$} (c2);
\path (c2) edge [->] node[anchor=south,inner sep = 0.2cm]{$4$} (t2);

\end{tikzpicture}
}

\caption{A route with two contention points and the weights of each arc. Vertex $c_2$ represents the BBU where messages may be buffered.}
\label{fig:routeexample}
\end{figure}

    The {\bf weight of a vertex} $u_i$ in a route $r=(u_0,\dots,u_l)$ is defined by $\lambda(r,u_i)= \sum\limits_{0 \leq j <i} \omega(r,u_j)$. It is the time needed by a message to go from the first vertex of the route to $u_i$. The \textbf{length} of the route $r$ is defined by $\lambda(r)= \lambda(r,u_l)$. 

  	On each route, we choose to buffer the message only in the BBU. Since the BBU does not correspond to a contention point, we identify the BBU with the next contention point in the route. The set of these contention points with possible buffering is denoted by ${\cal B}$. Hence, in this article, each route has exactly one vertex in $\cal{B}$. We can now define the \textbf{routed network}, which models the telecommunication network topology in this article; see Figure~\ref{fig:graphmodel} for an example.

  	\begin{definition}
    A \textbf{routed network} is a triple $(\cal{R},\,\cal{B},\,\omega)$ where $\cal{R}$ is a set of routes, $\cal{B}$ a set of vertices and $\omega$ a weight function on the routes of $\cal{R}$. 
    \end{definition}

\begin{figure}
\centering

	\includegraphics[scale=0.7]{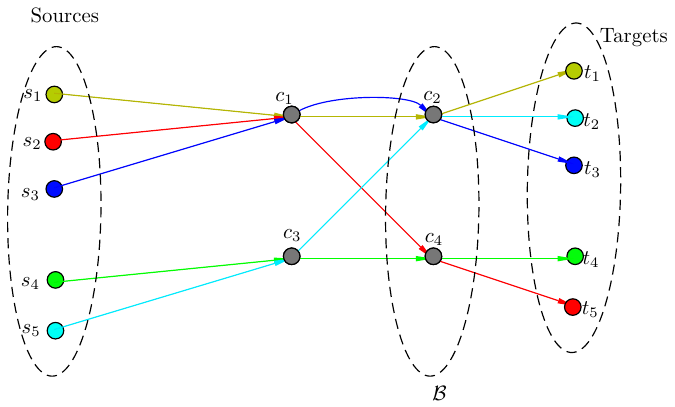}\\

\caption{A routed network, each route is represented by a colored path. The BBUs are located in $c_2$ and $c_4$. Weights on the arcs are omitted.}
\label{fig:graphmodel}
\end{figure}

 	\subsection{Dynamic of Datagrams Transmissions}
	    
 		In this article, we consider discretized time. The unit of time is called a {\bf tic}. This is the time needed to send atomic data in a link of the network. We assume that the link speed is the same throughout the network.  We are developing a prototype of this work based on Ethernet base-X~\cite{ieee_8023}, using standard values for the parameters of the network: the size of an atomic data is $64$ Bytes, the speed of the links is $10$ Gbps. Hence the duration of a tic is about $5.1$ nanoseconds. 

        In the process we study, a message called a {\bf datagram}, is sent on each route from each source node of a route. The \textbf{size} of a datagram is an integer, denoted by $\tau$: it is the number of tics needed by a node to emit the full datagram in a link. In this paper, we assume that $\tau$ is the same for all routes. It is justified by our application to C-RAN, where all source nodes are RRHs sending the same type of datagram. There is no fragmentation: Once a datagram has been emitted, it cannot be fragmented during its travel in the network. 

        To avoid contention, it is possible to choose the emission time of the datagrams and also to buffer datagrams in contention points in $\cal{B}$.
        These choices are represented by an \textbf{assignment}, defined as follows.

         \begin{definition}
         An assignment $A$ of a routed network $(\cal{R},\,\cal{B},\,\omega)$ is a function that associates to each route $r \in \cal{R}$, the pair of integers $A(r) = (o_r,w_r)$.
         \end{definition}
        The value $o_r$ is the \textbf{offset} of the route $r$ in the assignment $A$, the time at which the datagram is sent from the first vertex of $r$.
         The value $w_r$ is the \textbf{waiting time} of the route $r$ in the assignment $A$: the datagram is buffered for $w_r$ tics in $u_j \in {\cal B} \cap r$, the vertex representing the BBU.
 		The \textbf{arrival time} of a datagram in the vertex $u_i$ of $r$, is the first time at which the datagram sent on $r$ reaches $u_i$, and is defined by $t(r,u_i) = \lambda(r,u_i) + o_r $ if 
 		$i \leq j$ and $t(r,u_i) = \lambda(r,u_i) + o_r + w_r$ otherwise.

        \begin{example}
        Consider the route $r$ of Figure~\ref{fig:routeexample} and an assignment such that $o_r=2$ and $w_r = 6$. The arrival times of a datagram are the following: $t(r,c_1) =  \lambda(r,c_1) + o_r =  3 + 2 = 5$, $t(r,c_2) = \lambda(r,c_2) + o_r = 12$ and  $t(r,t) = \lambda(r,t) + o_r + w_r  = 22$.
        \end{example}

 		 Let $u_l$ be the last vertex of the route $r$; the \textbf{transmission time} of the datagram on 
  		$r$ is denoted by $TR(r,A)$ and is equal to $\lambda(r) + w_r$ (also $t(r,u_l) - o_r$). This is the total time taken by the process we study: the sending of the datagram from the RRH to the BBU and the return of the answer back to the RRH. We can decompose this time into $\lambda(r)$, the \emph{physical latency} of the process and $w_r$, the \emph{logical latency}. 
  		We define the \textbf{transmission time} of an assignment $A$ as the worst transmission time of a route: $TR(A) = \displaystyle \max\limits_{r \in {\cal R}} TR(r,A)$. 
        Figure~\ref{fig:datagramtimeline} represents the different events happening during the lifetime of a datagram sent on a route $r$.
  		\begin{figure}
  		 \begin{center}
      \includegraphics[width=\textwidth]{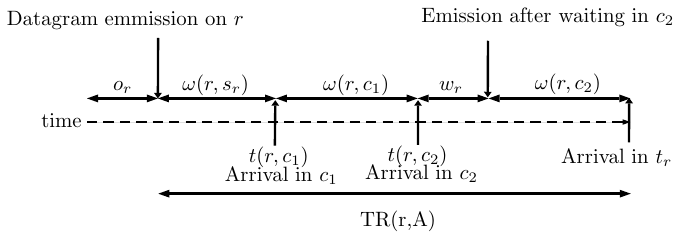}
      \end{center}
      \caption{Timeline of a datagram during its travel on a route $r = (s_r,c_1,c_2,t_r)$, with $c_2 \in \cal{B}$ and $A(r) = (o_r,w_r)$}
      \label{fig:datagramtimeline}
  		\end{figure}

 	\subsection{Periodic Emission of Datagrams}

	In the previous section, we have explained how \emph{one} datagram follows its route.
	However, the sources of datagrams we model in this article are \emph{periodic}: for each period of $P$ tics, a datagram is sent, from each source node in the network, at its offset. The process is assumed to be infinite since it must work for an arbitrary number of periods. In this article, for a given route, we use the same offset and waiting time in all periods. This choice makes our problem more tractable from a theoretical perspective and allows for a much simpler implementation in real networks. As a consequence, at the same time of two different periods, all datagrams are at the same position in the network: the assignments built are themselves periodic of period $P$. Thus, it is sufficient to consider the behavior of datagrams on each node of the network during a single period and to apply the same pattern to each subsequent period. 

 	Using a different offset for each route corresponds to sending their datagram at a different time in the period. 
This is consistent with our assumption that RRH emissions do not need to be synchronized, but they share a common global clock, useful for coordinating their emissions.

 	Let $A$ be an assignment of a routed network $(\cal{R},\,\cal{B},\,\omega)$.
    Let us denote by $[r,u]_{P,\tau}$, the set of tics used by a datagram on the route $r$ at vertex $u$ in a period $P$, that is $[r,u]_{P,\tau} = \{t(r,u) + i \mod P \mid 0 \leq i < \tau \}$. This set of tics depends on $A$, but $A$ is omitted in the notation since it is always clear from the context. We can now define the constraints that an assignment must respect to represent a \textbf{valid}
    sending process, for a given period $P$ and size $\tau$.

    \begin{definition}
    Let $r_1$ and $r_2$ be two routes of a routed network $N$ and let $A$ be an assignment of $N$. The routes $r_1$ and $r_2$ have a collision at the contention point $u$ for the assignment $A$ if and only if $[r_1,u]_{P,\tau} \cap [r_2,u]_{P,\tau} \neq \emptyset$.
    \end{definition}
    
\begin{example}
Consider $P=6$, $\tau = 3$, and two routes $r_1$ and $r_2$ with a common contention point $c$. Take $o_1 = o_2 = 0$, $\lambda(r_1,c) = 5$ and $\lambda(r_2,c) = 1$ (this example is illustrated by Figure~\ref{fig:cols}). Let the tics of a period be numbered from $0$ to $P-1$. Thus, we can compute $[r_1,c]_{P,\tau} =\{5,0,1\}$ and $[r_2,c]_{P,\tau} = \{1,2,3\}$. There is a collision because $[r_1,c]_{P,\tau}\cap[r_2,c]_{P,\tau}=\{1\} $. This collision can be avoided by taking $o_2=1$ in order to obtain $[r_2,c]_{P,\tau} = \{2,3,4\} $.
\end{example}
\begin{figure}
 
 \begin{center}
\scalebox{0.6}{

\begin{tikzpicture}
  \SetGraphUnit{5}
    \tikzset{
  EdgeStyle/.append style = {->} }
   \tikzstyle{VertexStyle}=[shape = circle, draw, minimum size = 30pt]
   \renewcommand{\VertexLightFillColor}{orange}
  \Vertex[x=0,y=3, L = {\huge $s_2$}]{a3};

  \Vertex[x=0,y=6, L = {\huge $s_1$}]{a1}

  \Vertex[x=4,y=3, L = {\huge $c$}]{c}

 \SetVertexNoLabel
\Vertex[x=7,y=3]{d}

      \Edge(c)(d)
  \tikzset{
  EdgeStyle/.append style = {blue} }
  \Edge[label = 5](a1)(c)

    \tikzset{
  EdgeStyle/.append style = {red} }
    \Edge[label = 1](a3)(c)

\end{tikzpicture}

}\\

\vspace{0.5cm}
\includegraphics[width=0.6\textwidth]{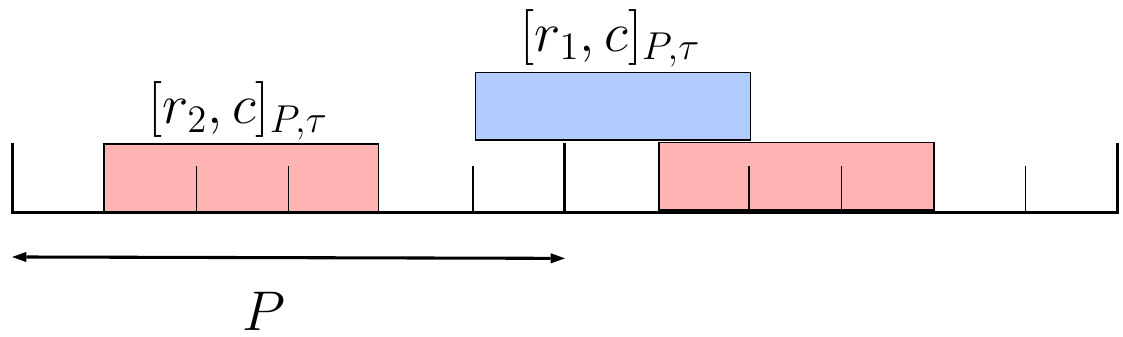}
\caption{A collision between two messages due to the periodicity when $P=6$,$\tau=3$ and $o_1 = o_2 = 0$.}
\label{fig:cols}
\end{center}

\end{figure}

    \begin{definition}
    An assignment $A$ is valid if, for all contention points $u$ and routes $r_1$ and $r_2$ containing $u$, $r_1$ and $r_2$ have no collision at $u$ for $A$. 
    \end{definition}

    The validity of an assignment depends on $P$ the period and $\tau$ the size of the datagrams; thus if we need to mention these parameters, we say that $A$ is a \textbf{valid $(P,\tau)$-assignment}. When $P$ and $\tau$ are clear from the context, we denote $[r,u]_{P,\tau}$ by $[r,u]$ and say that $A$ is a valid assignment. 

\begin{example}

    Figure~\ref{fig:example} illustrates two valid assignments for different values of $P$ and $\tau$, but the same network. The three routes are depicted in three different colors. If we let $P = 2$ and $\tau = 1$, then there is a $(2,1)$-periodic valid assignment with waiting times zero by taking $0$ as offset for each route. However, for the same routed network but $P=5$ and $\tau = 2$, there is no solution to the problem with all waiting times zero. If we allow $1$ tic of waiting time for one route, we can build the valid assignment $A'(r_1) = (0,0)$, $A'(r_2) = (2,1)$, $A'(r_3)= (0,0)$.

\end{example}
  
\begin{figure}[ht]
    \begin{center}
        \scalebox{0.7}{
		\begin{tikzpicture}
\tikzset{
  LabelStyle/.style = { rectangle, rounded corners, draw,
                       font = \bfseries },
  EdgeStyle/.append style = {->} }
  \SetGraphUnit{5}
  \node[draw,circle] (s3) at (4, 2) {\Large $s_2$}; 
  \node[draw,circle] (s2) at (0, 4) {\Large $s_1$}; 
  \node[draw,circle] (s1) at (0, 6) {\Large $s_0$}; 

  \node[draw,circle] (t3) at (14, 7) {\Large $t_2$}; 
  \node[draw,circle] (t2) at (14, 4) {\Large $t_1$}; 
  \node[draw,circle] (t1) at (14, 2) {\Large $t_0$}; 

   \node[circle] (buf) at (10, 6) {\Large $\cal{B}$};
  \SetVertexNoLabel
  \Vertex[x=2,y=5]{A}
    \Vertex[x=10,y=2]{B}
  \Vertex[x=10,y=5]{C}

  \Vertex[x=6,y=3]{E}
     \draw[dashed] (10,3.5) ellipse  (1cm and 2cm);
  \tikzset{
  EdgeStyle/.append style = {green} }
  \Edge[label = 2](s2)(A)

  \Edge[label = 1](A)(C)
 
  \Edge[label = 2](C)(t2)

   \tikzset{
  EdgeStyle/.append style = {red} }
  \Edge[label = 2](s3)(E)
  \Edge[label = 2](E)(C)
  \Edge[label = 3](C)(t3) 

     \tikzset{
  EdgeStyle/.append style = {blue} }
  \Edge[label = 1](s1)(A)
  \Edge[label = 2](A)(E)
  \Edge[label = 2](E)(B)
 \Edge[label = 1](B)(t1)

\end{tikzpicture}

}
  	\end{center}
    \caption{A routed network with $A(\textcolor{blue}{r_1})= \textcolor{blue}{(0,0)}$,
    $A(\textcolor{green}{r_2}) = \textcolor{green}{(0,0)}$, $A(\textcolor{red}{r_3}) = \textcolor{red}{(0,0)}$ as a $(2,1)$-periodic valid assignment and $A'(\textcolor{blue}{r_1})= \textcolor{blue}{(0,0)}$,
    $A'(\textcolor{green}{r_2}) = \textcolor{green}{(2,1)}$, $A'(\textcolor{red}{r_3}) = \textcolor{red}{(0,0)}$ as a $(5,2)$-periodic valid assignment}
    \label{fig:example}
\end{figure}

	\subsection{Valid Assignment with Low Latency}

      The period $P$, as well as the size of a datagram $\tau$, are fixed in our C-RAN settings, but not the buffering policy. Hence, this article aims to find a valid assignment that minimizes the worst latency of the transmissions over the network, that is $TR(A)$. We formally introduce the problem we want to solve in this article. 

      \bigskip

      \noindent \textbf{Minimal Transmission Time }(\mintra)

      \noindent {\bf Input:} A routed network $N$, integers $P$ and $\tau$.
      
      \noindent {\bf Question:} Find the minimum value $TR(A)$ for all $A$ valid $(P,\tau)$ assignments of $N$. 
      \bigskip

      For simpler hardness proofs and easier reductions, we rather study the decision version of \mintra, which we call \pall for \textbf{P}eriodic \textbf{A}ssignment for \textbf{L}ow \textbf{L}atency. Each route must respect a time limit called a \emph{deadline}. These limits are encoded in a deadline function $d$, which maps to each route $r$ an integer such that $TR(r,A)$ must be at most $d(r)$.

      \bigskip

      \noindent {\bf Periodic Assignment with Low Latency} (\pall)

      \noindent {\bf Input:}  A routed network $N$, integers $P$ and $\tau$ and a deadline function $d$.
      
      \noindent {\bf Question:} Does there exist a valid $(P,\tau)$ assignment $A$ of $N$ such that for all $r \in {\cal R}$, $TR(r,A) \leq d(r)$?
      \bigskip

In the next subsection, this problem is proved to be $\NP$-hard. In Section~\ref{sec:PALL}, we propose heuristics solving the search version of \pall (computing a valid assignment), also denoted by \pall for simplicity. In the definition of \pall, we have chosen to bound the transmission time of each route, in particular, we control the worst-case latency. It is justified by our C-RAN application with hard constraints on the latency.

We say that an assignment is \textbf{bufferless} when the waiting time of all routes is zero. The assignment can then be seen as a function from the routes to the integers, which produces the value of the offset, the waiting time being omitted. We consider a restricted version of \pall, requiring finding a bufferless assignment. This is equivalent to using the deadline function $d(r) = \lambda(r)$ (the transmission time must be equal to the size of the route), which implies $w_r = 0$ for all $r \in \cal{R}$. This problem is called \textbf{P}eriodic \textbf{A}ssignment for \textbf{Z}ero \textbf{L}atency and is denoted by \pazl.

     \bigskip

      \noindent {\bf Periodic Assignment with Zero Latency }(\pazl)

      \noindent {\bf Input:}  A routed network $N$, integers $P$ and $\tau$.
      
      \noindent {\bf Question:} Does there exist a valid bufferless $(P,\tau)$ assignment of $N$?
      \bigskip

     Studying \pazl is simpler: in the instance representing a routed network, $\cal{B}$ and the deadline function may be omitted and a solution is a function that associates an offset to each route.  Moreover, a solution to \pazl is more efficient when implemented in real telecommunication networks, since we do not need a contention buffer at all. A switch taking full advantage of the absence of a buffer is presented in~\cite{Marc2201:Experimental}. 
      An unusual property of assignments is that given a routed network and a deadline, we may have a $(P,\tau)$ assignment but no $(P',\tau)$ assignment with $P' > P$: the existence of an assignment is not monotone with regard to $P$.

	\begin{proposition} \label{prop:monotonic}
	 For any odd $P$, there is a routed network with a $(2,1)$-periodic bufferless assignment but no $(P,1)$-periodic bufferless assignment.
	\end{proposition}

	\begin{proof}
      Let us build $N$, a generalization of the routed network given in Figure~\ref{fig:example}. 
      Let $n$ be an integer, the vertices of the routes are $v_{i,j}$, $v_i^1$ and $v_i^2$, with $0 \leq i < j <n$. 
      There are $n$ routes denoted by $r_i$, for $i \in [n]$. The route $r_i$ is equal to $(v_i^1,v_{i,1},\dots,v_{i,n-1},v_i^2)$. The weights of the arcs are set so that $\lambda(r_i, v_{i,j}) - \lambda(r_j,v_{i,j})= P$, where $P$ is an odd number smaller than $n$. It is always possible by choosing appropriate values for $\omega(r_i,v_{i,j-1})$ and $\omega(r_j,v_{i-1,j})$. In such a graph, there is no bufferless $(P,\tau)$ assignment, since the problem reduces to finding a $P$-coloring in a complete graph with $n > P$ vertices, the colors being the offsets of the routes.

      If we consider a period of $2$, for all $i \neq j$, $\lambda(r_i, v_{i,j}) - \lambda(r_j, v_{i,j}) \mod 2 = 1$, hence two datagrams of same offset and size $1$ do not have a collision at $v_{i,j}$. Therefore, the bufferless assignment defined by $A(r_i) = 0$ for all $i \in [n]$ is a valid $(2,1)$ assignment of $N$.      
\end{proof}

      The table of Figure~\ref{tab:summary} summarizes the main notations used in the paper.
    \begin{figure}
      \begin{center}
    \begin{tabular}{|c|l|}
    \hline
     $N = (\cal{R},\,\cal{B},\,\omega)$ & Routed network \\
     \hline
     $n =  \lvert\cal{R} \rvert$ & Number of routes\\
     \hline
     $P$ & Period\\
     \hline
     $\tau$ & Size of a datagram\\
     \hline
     $\omega(r,u)$ & Weight of the arc $(u,v)$ of $r$ \\
     \hline
     $\lambda(r,u)$ & Length of the route $r$ up to vertex $u$\\
     \hline
     $\lambda(r)$ & Length of the route $r$\\
     \hline 
     $A$ & Assignment\\
     \hline 
     $A(r) = (o_r,w_r)$ & Offset and waiting time of the route $r$ given by $A$ \\
     \hline 
     $TR(A,r)$& Transmission time of the route $r$ for the assignment $A$\\
     \hline 
     $TR(A)$& Transmission time of the assignment $A$\\
     \hline
     $d(r)$ & Deadline of the route $r$\\
     \hline
	 $t(r,u)$ & Transmission time on the route $r$, up to the vertex $u$\\
     \hline
     $ [r,u]$ & Tics used in the period by the route $r$ at vertex $u$\\
     \hline
      \end{tabular}
      \end{center}
      \caption{Summary of the notations of the article.}\label{tab:summary}
    \end{figure}
  	
  Let us introduce a few parameters quantifying the complexity of a routed network.
The \textbf{contention depth} of a routed network is the size of the longest route (number of arcs) of the network, minus one. It is the number of contention points on the route since the first and the last vertex are private to the route. The \textbf{width} of a vertex is the number of routes that contain it --- equivalently its indegree or its outdegree. By definition,
the first and last vertices of a route are of width one, while all other vertices are of width at least two (otherwise they are removed).
The \textbf{contention width} of a routed network is the maximal width of its vertices.
A valid $(P,\tau)$ assignment of a routed network must satisfy that $P/\tau$ is larger or equal to the contention width. Now, let us fix $P$ and $\tau$. For a given vertex of contention width $c$, we define its \textbf{load} as $c\tau/P$. It represents the proportion of the period used by datagrams at this contention point. The load of the routed network is the maximum load of its vertices. A routed network must have a load less or equal to one to admit a valid assignment.

\subsection{The Star Routed Network} \label{sec:star_routed_network}

In this section, we define a family of simple routed networks modeling a Multipoint-to-Multipoint fronthaul (see Figure~\ref{fig:star}), which has been designed for C-RAN \cite{tayq2017real}. Let $N = (\cal{R},\,\cal{B},\,\omega)$ be a routed network; we say it is a \textbf{star routed network} if and only if ${\cal R} =\{r_0,\dots,r_{n-1}\}$ where for $0 \leq i<n$, $r_i = (s_i,c_1,c_2,t_i)$ and ${\cal B} = \{ c_2 \}$ (datagrams can wait in $c_2$). Star-routed networks have a contention depth of two but a maximal contention width of $n$. The load on each of the two contention points is thus $n\tau / P$.

The fronthaul network we model with a star-routed network has a single shared link, which connects all RRHs at one end and all BBUs at the other end. The links are all \emph{full-duplex}, meaning that the datagrams going from RRHs to BBUs do not interact with those going in the other direction.
This property does not need to be enforced in our theoretical modeling, but it matches a real fronthaul network, and we will use such examples for our experiments.

The two contention points $c_1$ and $c_2$ model the beginning of the shared link (used to go from the RRHs to the BBUs) and the other end of the shared link (used in the other direction).
The computation in the BBU of an answer to a datagram on the route $r$ takes some time.
In the star-routed network, this time is encoded in the weight of the arc between $c_1$ and $c_2$ in $r$. The weight $\omega(r,c_1)$ is the time needed to go through the shared link, then to arrive at the BBU, plus the computation time and the time to return to the shared link, as illustrated in Figure~\ref{fig:star}.

Star-routed networks are simple, but every network in which all routes share an arc and satisfy a coherent routing condition can be modeled by a star-routed network.
Fronthaul networks, where all the BBUs are located in the same data center, are star-routed networks. In such a situation, we can consider the weights of the arcs $(c_1,c_2)$ to be either all equal (in that case \pazl is trivial, see Section~\ref{sec:PALL}) or different due to the structure of the network inside the data center and the various hardware used for the BBUs, computing with different speeds.

\begin{figure}
\begin{center}
\scalebox{0.4}{

\begin{tikzpicture}
  \SetGraphUnit{5}
    \tikzset{
  EdgeStyle/.append style = {->} }
   \tikzstyle{VertexStyle}=[shape = circle, draw, minimum size = 30pt]

  \node (s1) at (0,4) {\includegraphics[width = 1cm]{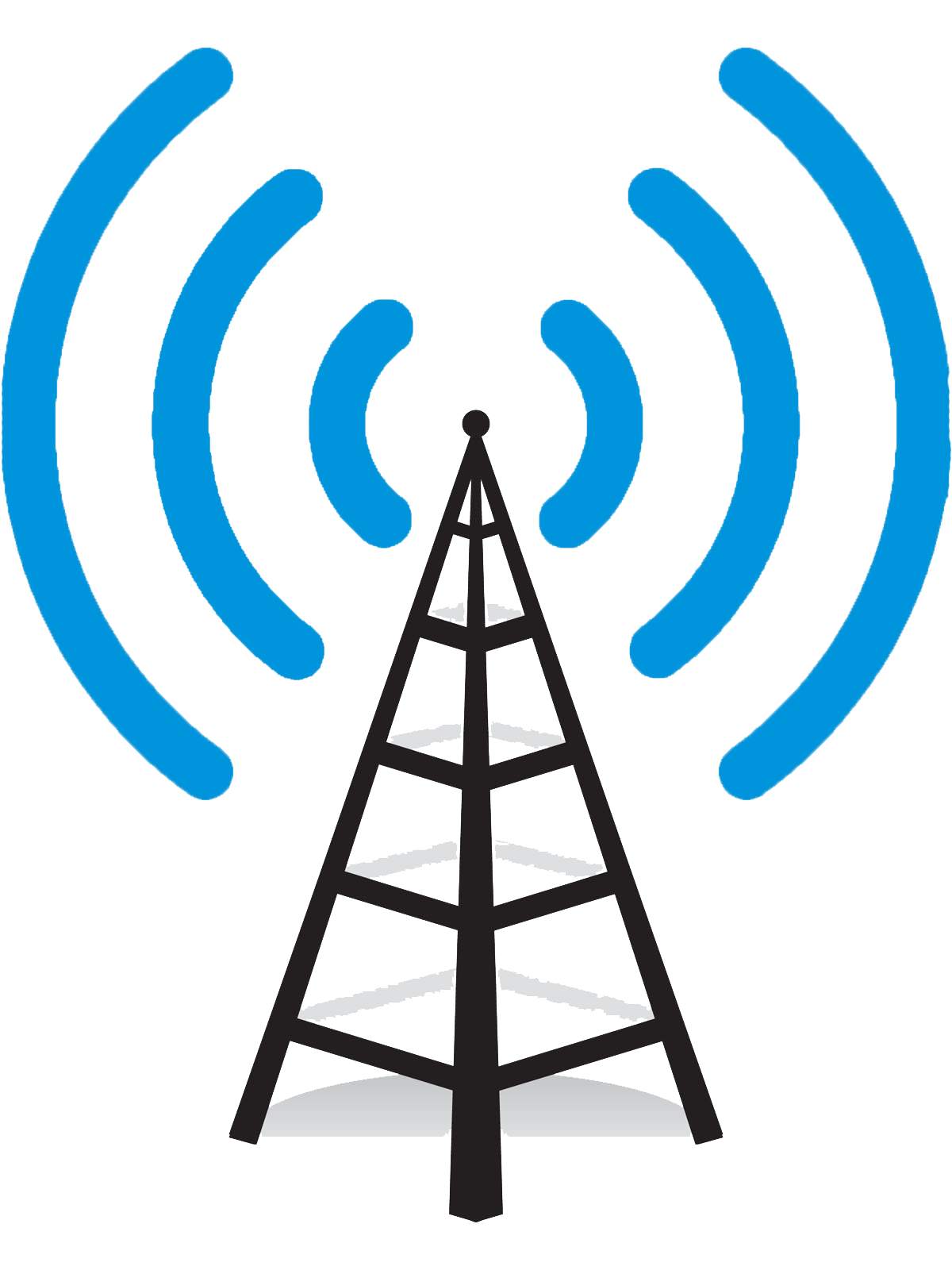}};
  \node[below] at (s1.south) {\huge $r_1$};
  \node (s2) at (0,2) {\includegraphics[width = 1cm]{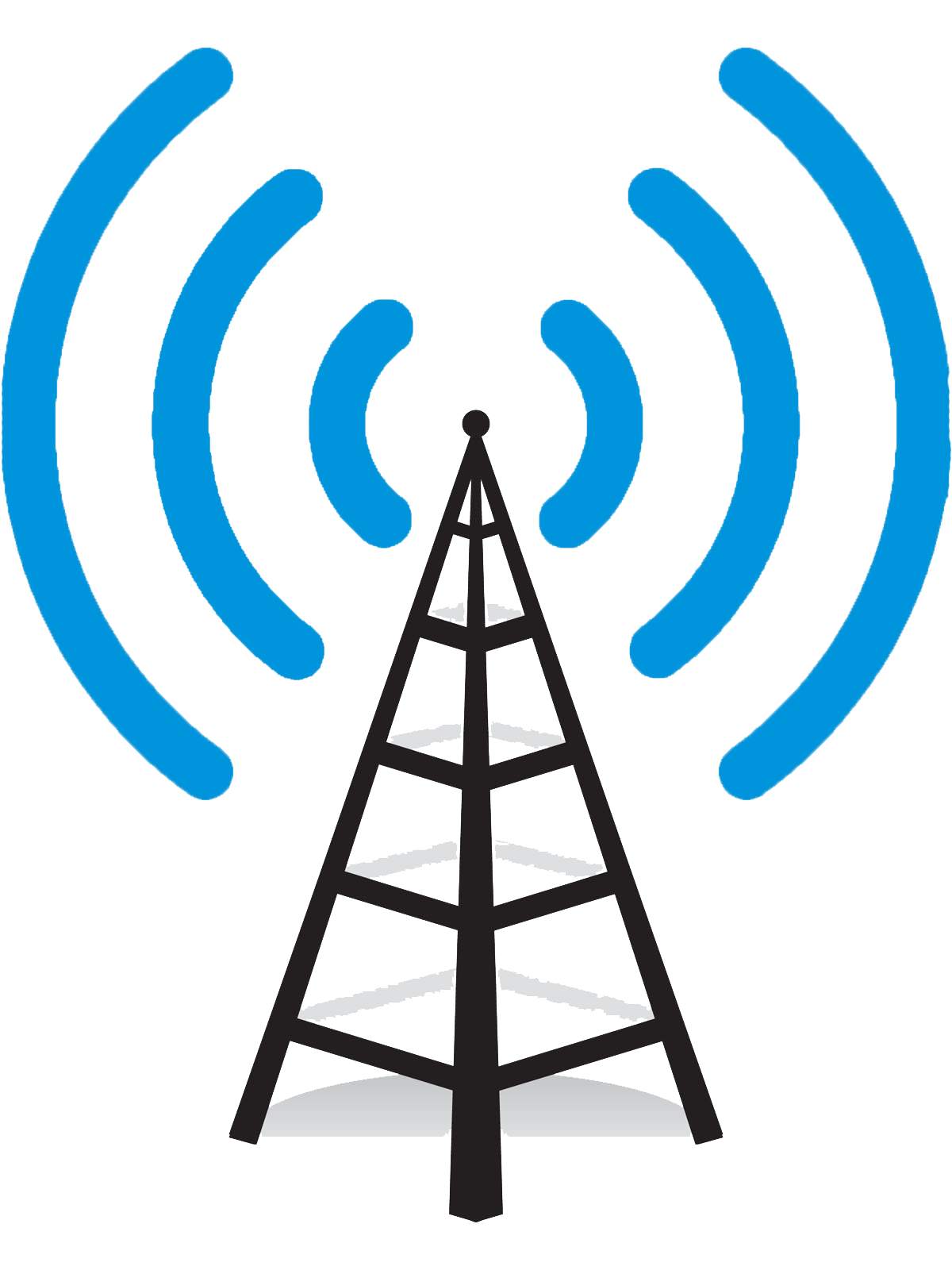}};
  \node[below] at (s2.south) {\huge $r_2$};
  \node (s3) at (0,0) {\includegraphics[width = 1cm]{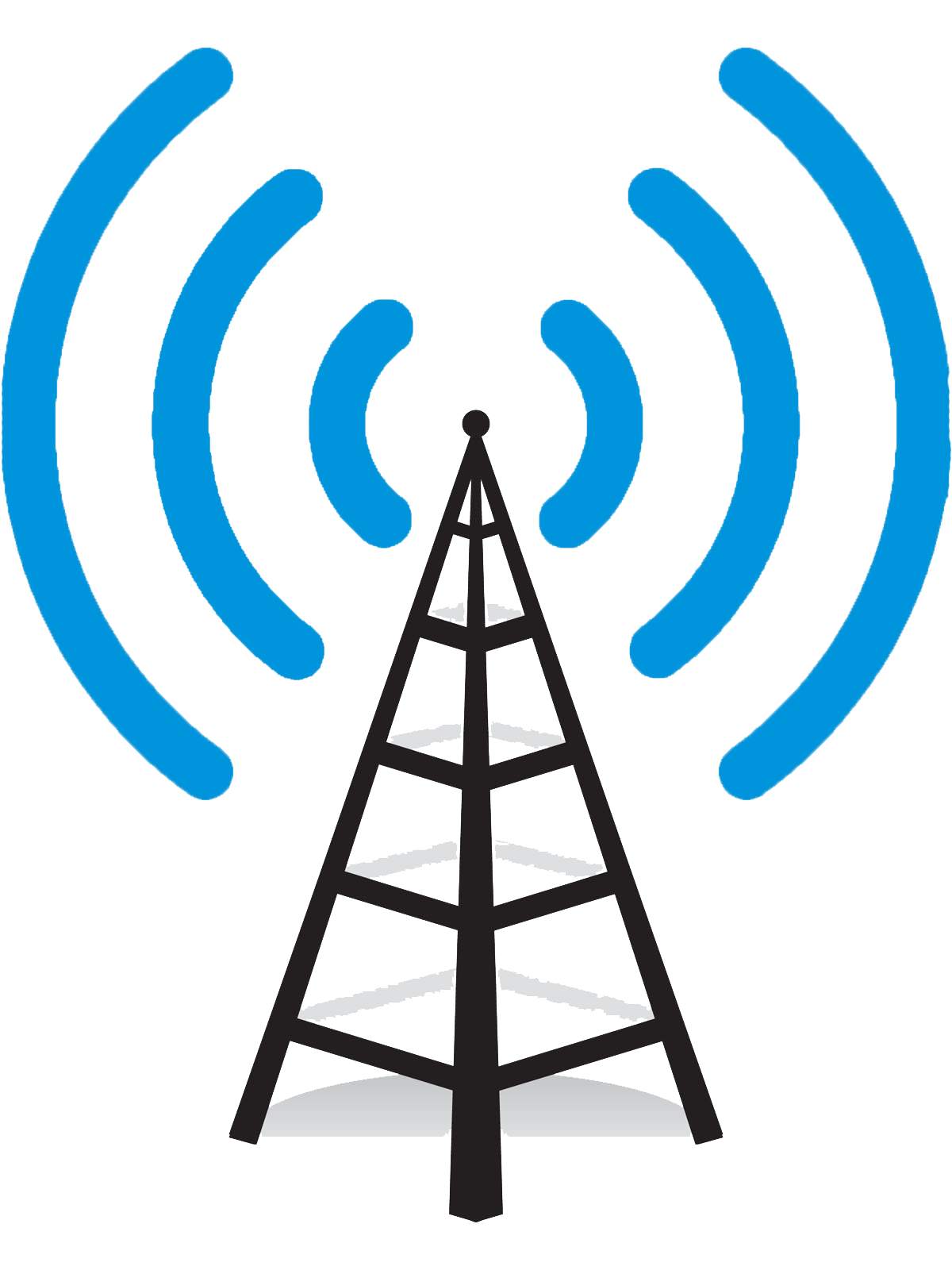}};
  \node[below] at (s3.south) {\huge $r_3$};
  
   \node (t1) at (12,4) {\includegraphics[width = 1cm]{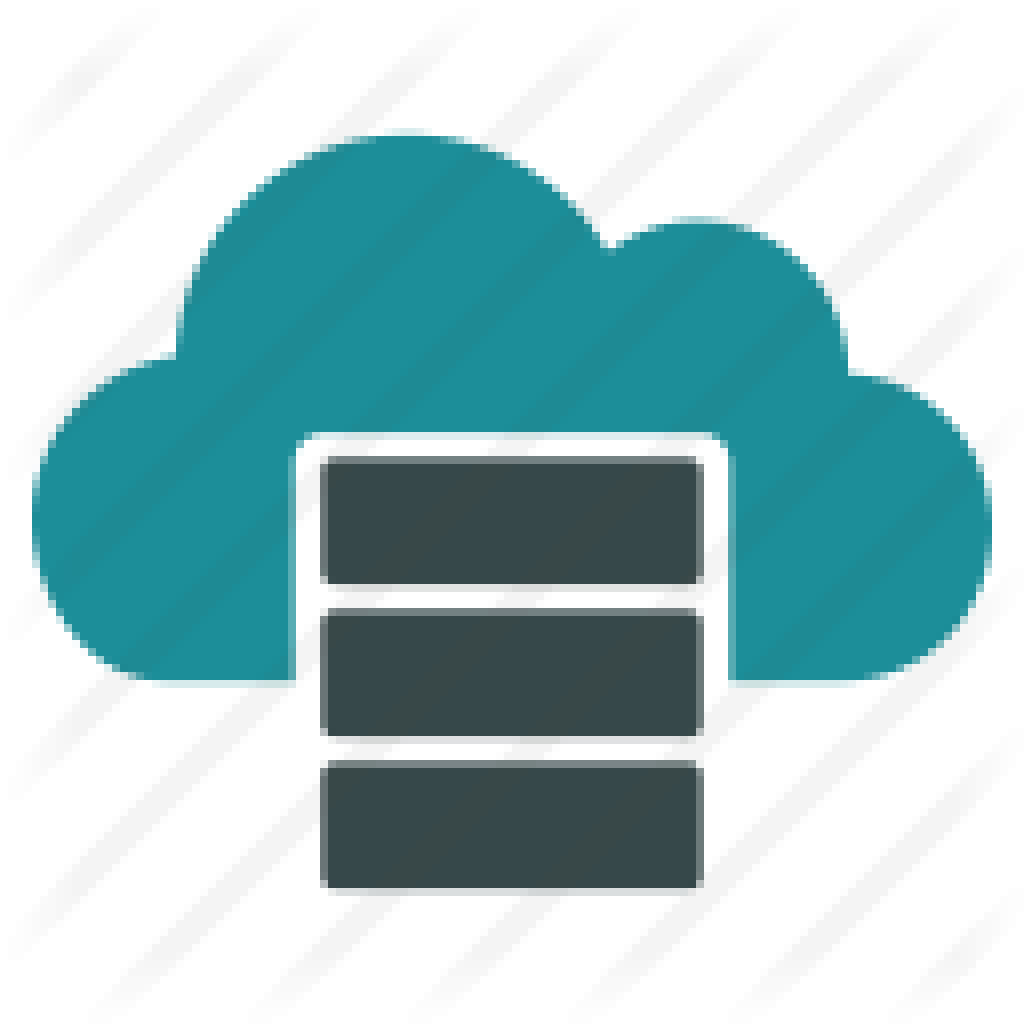}};
  \node[below] at (t1.south) {\huge $b_1$};
  \node at (t1.north west) {\textcolor{red}{3}};
  \node (t2) at (12,2) {\includegraphics[width = 1cm]{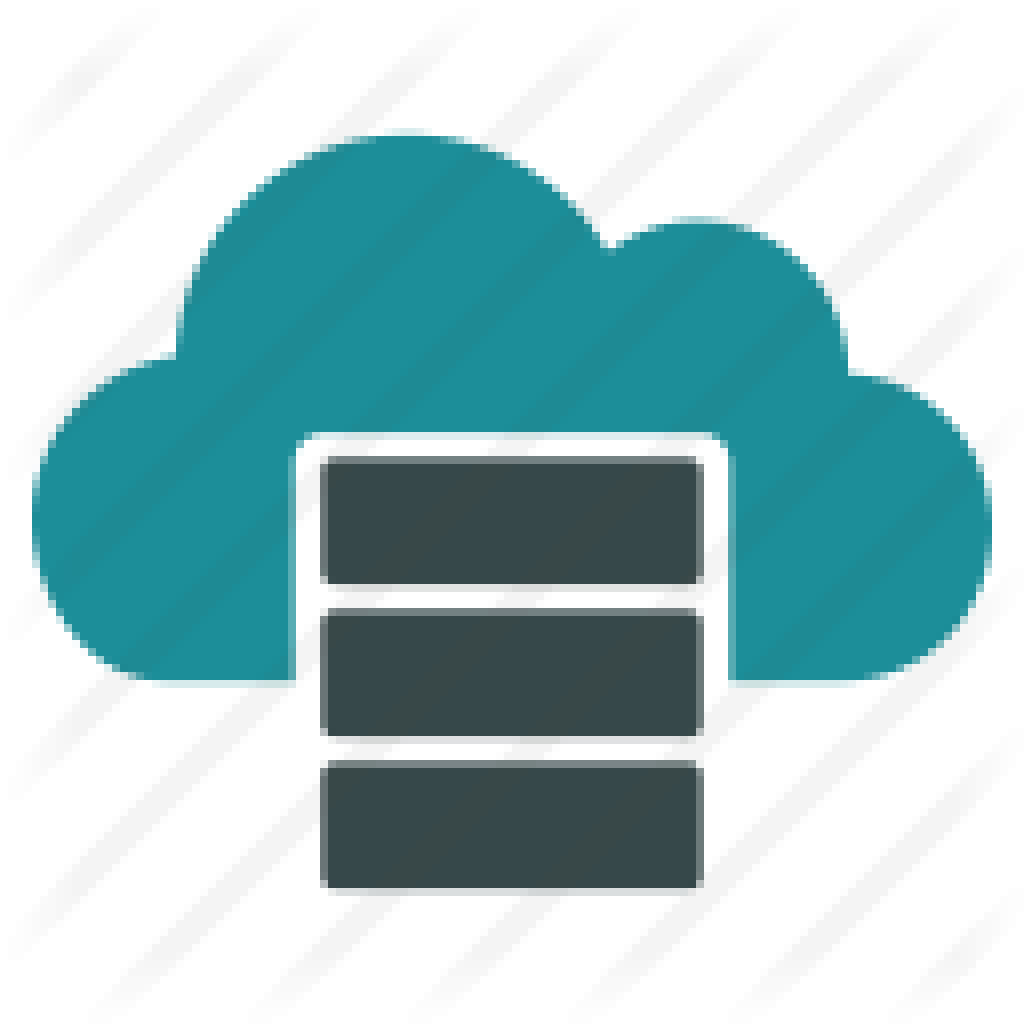}};
  \node[below] at (t2.south) {\huge $b_2$};
   \node at (t2.north west) {\textcolor{red}{2}};
  \node (t3) at (12,0) {\includegraphics[width = 1cm]{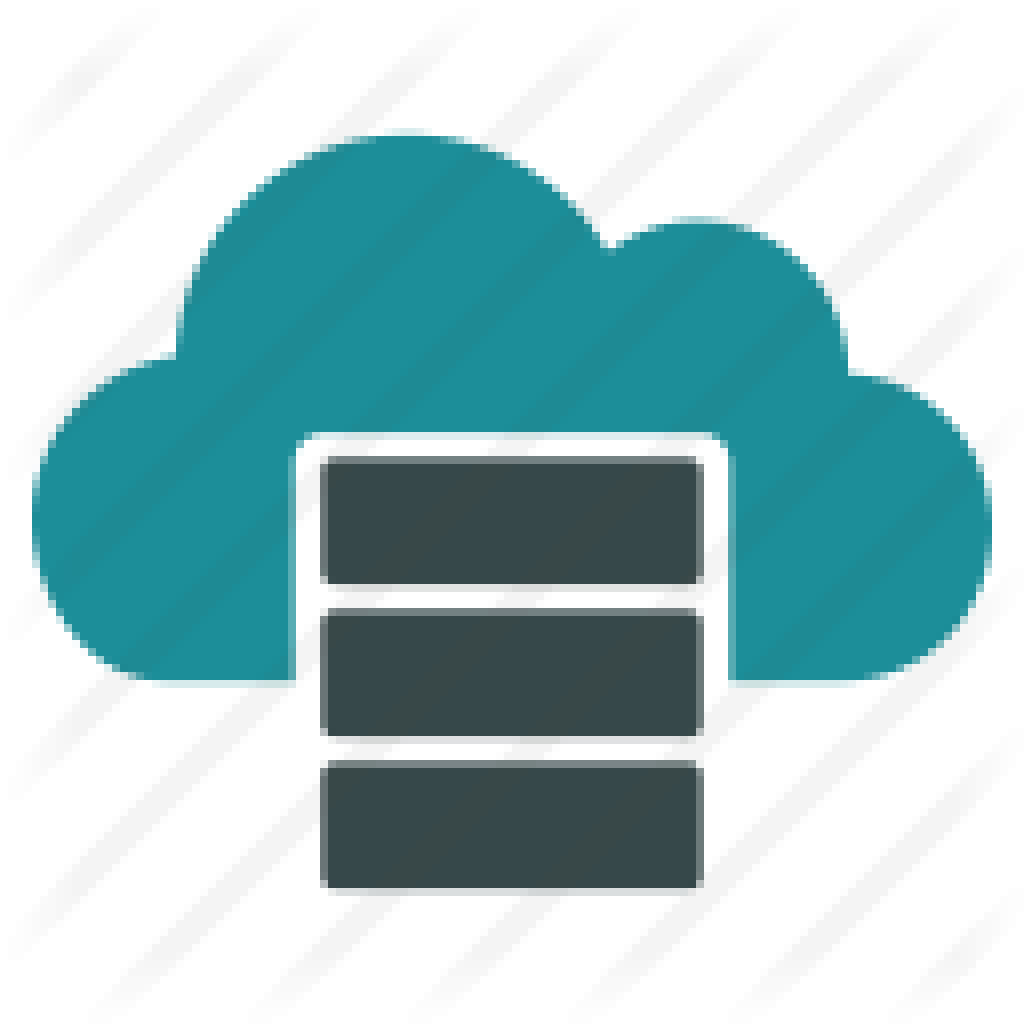}};
  \node[below] at (t3.south) {\huge $b_3$};
   \node at (t3.north west) {\textcolor{red}{3}};
    \node (c1) at (4,2) {\includegraphics[width = 1cm]{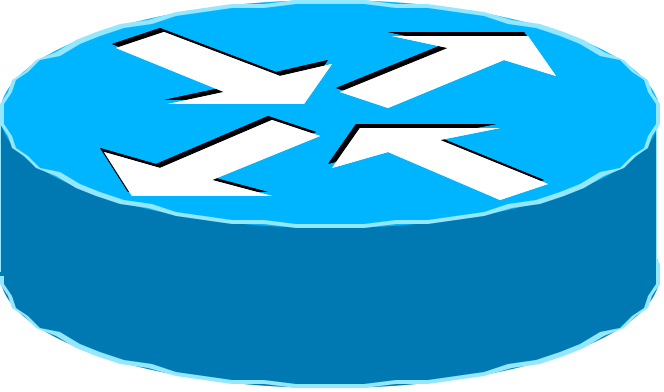}};
  
    \node (c2) at (8,2) {\includegraphics[width = 1cm]{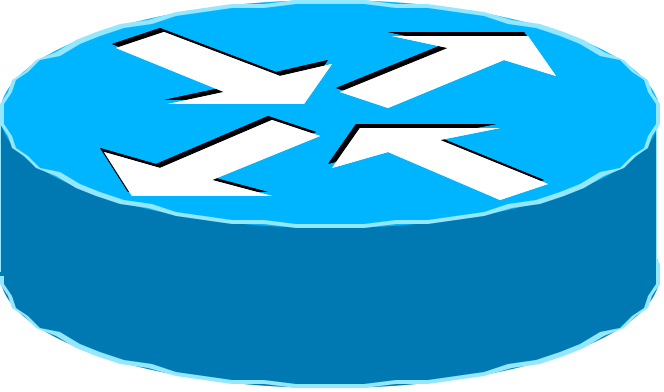}};
  
\path (s1) edge [<->] node[anchor=south,inner sep = 0.2cm]{$5$} (c1);

\path (s2) edge [<->] node[anchor=south,inner sep = 0.2cm]{$3$} (c1);
\path (s3) edge [<->] node[anchor=south,inner sep = 0.2cm]{$7$} (c1);
\path (c2) edge [<->] node[anchor=south,inner sep = 0.2cm]{$4$} (t2);
\path (c2) edge [<->] node[anchor=south,inner sep = 0.2cm]{$1$} (t3);

\path (c2) edge [<->] node[anchor=south,inner sep = 0.2cm]{$2$} (t1);

\node[above] at (6,1.75) {7};

 \path (c1) edge [<-,bend left=15] node[anchor=south,inner sep = 0.2cm]{\huge $c_2$} (c2);
\path (c1) edge [->,bend right=15] node[anchor=north,inner sep = 0.2cm]{\huge $c_1$} (c2);

  \Vertex[x=14,y=4, L = {\huge $s_1$}]{s1};
  \Vertex[x=14,y=2, L = {\huge $s_2$}]{s2};
\Vertex[x=14,y=0, L = {\huge $s_3$}]{s3};
\Vertex[x=26,y=4, L = {\huge $t_1$}]{s1p};
\Vertex[x=26,y=2, L = {\huge $t_2$}]{s2p};
\Vertex[x=26,y=0, L = {\huge $t_3$}]{s3p};
\Vertex[x=22,y=2, L = {\huge $c_2$}]{c2};

  \Vertex[x=18,y=2, L = {\huge $c_1$}]{c1}

\path (s1) edge [->] node[anchor=south,inner sep = 0.2cm]{$5$} (c1);

\path (s2) edge [->] node[anchor=south,inner sep = 0.2cm]{$3$} (c1);
\path (s3) edge [->] node[anchor=south,inner sep = 0.2cm]{$7$} (c1);
\path (c2) edge [->] node[anchor=south,inner sep = 0.2cm]{$7+3$} (s2p);
\path (c2) edge [->] node[anchor=south,inner sep = 0.2cm]{$7+7$} (s3p);

\path (c2) edge [->] node[anchor=south,inner sep = 0.2cm]{$7+5$} (s1p);

\path (c1) edge [->] node[anchor=south,inner sep = 0.2cm]{$7+4+\textcolor{red}{2}+4$} (c2);

\path (c1) edge [->,bend left=30] node[anchor=south,inner sep = 0.2cm]{$7+2+\textcolor{red}{3}+2$} (c2);
\path (c1) edge [->,bend right=30] node[anchor=north,inner sep = 0.2cm]{$7+1+\textcolor{red}{3}+1$} (c2);
   \node[circle] (buf) at (22, 0.5) {$\cal{B}$};
   \draw[dashed] (22,2) ellipse  (0.6cm and 1cm);

\end{tikzpicture}
}

\caption{Left, a physical fronthaul network, and right, the star routed network modeling a round trip in the fronthaul network. The computation time in the BBU is given in red.}

\label{fig:star}
\end{center}
\end{figure}

When solving \pall or \pazl on a star-routed network, a period, a datagram size and a deadline function are also given. When the period is fixed, we modify the deadline function to do several simplifying assumptions on the parameters of the star-routed network without loss of generality. We say that a star-routed network is \textbf{canonical}, for a period $P$, if the weights of the arcs between $c_1$ and $c_2$ are in $[P]$ and the others are equal to zero. Hence, $\lambda(r_i)$, the length of a route is equal to the length of its arc $(c_1,c_2)$. Moreover, $\lambda(r_0) = 0$. See Figure~\ref{fig:canonical} for an example of the canonical star-routed network of Figure~\ref{fig:star}.
  
\begin{figure}
\begin{center}

 \scalebox{0.5}{

\begin{tikzpicture}
  \SetGraphUnit{5}
    \tikzset{
  EdgeStyle/.append style = {->} }
   \tikzstyle{VertexStyle}=[shape = circle, draw, minimum size = 30pt]
   \renewcommand{\VertexLightFillColor}{orange}
  \Vertex[x=0,y=4, L = {\huge $s_1$}]{s1};
  \Vertex[x=0,y=2, L = {\huge $s_2$}]{s2};
\Vertex[x=0,y=0, L = {\huge $s_3$}]{s3};
\Vertex[x=15,y=4, L = {\huge $t_1$}]{s1p};
\Vertex[x=15,y=2, L = {\huge $t_2$}]{s2p};
\Vertex[x=15,y=0, L = {\huge $t_3$}]{s3p};
\Vertex[x=10,y=2, L = {\huge $c_2$}]{c2};

  \Vertex[x=5,y=2, L = {\huge $c_1$}]{c1}
  
\path (s1) edge [->] node[anchor=south,inner sep = 0.2cm]{$0$} (c1);

\path (s2) edge [->] node[anchor=south,inner sep = 0.2cm]{$0$} (c1);
\path (s3) edge [->] node[anchor=south,inner sep = 0.2cm]{$0$} (c1);
\path (c2) edge [->] node[anchor=south,inner sep = 0.2cm]{$10$} (s2p);
\path (c2) edge [->] node[anchor=south,inner sep = 0.2cm]{$14$} (s3p);

\path (c2) edge [->] node[anchor=south,inner sep = 0.2cm]{$12$} (s1p);

\path (c1) edge [->] node[anchor=south,inner sep = 0.2cm]{$15$} (c2);

\path (c1) edge [->,bend left=30] node[anchor=south,inner sep = 0.2cm]{$11$} (c2);
\path (c1) edge [->,bend right=30] node[anchor=north,inner sep = 0.2cm]{$9$} (c2);
   \node[circle] (buf) at (10, 0.5) {$\cal{B}$};
   \draw[dashed] (10,2) ellipse  (0.6cm and 1cm);

\end{tikzpicture}
}

 $d_1 = 25$, $d_2 = 31$, $d_3 = 25$

$\downarrow$

 \scalebox{0.5}{

\begin{tikzpicture}
  \SetGraphUnit{5}
    \tikzset{
  EdgeStyle/.append style = {->} }
   \tikzstyle{VertexStyle}=[shape = circle, draw, minimum size = 30pt]
   \renewcommand{\VertexLightFillColor}{orange}
  \Vertex[x=0,y=4, L = {\huge $s_1$}]{s1};
  \Vertex[x=0,y=2, L = {\huge $s_2$}]{s2};
\Vertex[x=0,y=0, L = {\huge $s_3$}]{s3};
\Vertex[x=15,y=4, L = {\huge $t_1$}]{s1p};
\Vertex[x=15,y=2, L = {\huge $t_2$}]{s2p};
\Vertex[x=15,y=0, L = {\huge $t_3$}]{s3p};
\Vertex[x=10,y=2, L = {\huge $c_2$}]{c2};

  \Vertex[x=5,y=2, L = {\huge $c_1$}]{c1}

\path (s1) edge [->] node[anchor=south]{$0$} (c1);

\path (s2) edge [->] node[anchor=south,inner sep = 0.2cm]{$0$} (c1);
\path (s3) edge [->] node[anchor=south,inner sep = 0.2cm]{$0$} (c1);
\path (c2) edge [->] node[anchor=south,inner sep = 0.2cm]{$0$} (s2p);
\path (c2) edge [->] node[anchor=south,inner sep = 0.2cm]{$0$} (s3p);

\path (c2) edge [->] node[anchor=south,inner sep = 0.2cm]{$0$} (s1p);

\path (c1) edge [->] node[anchor=south,inner sep = 0.2cm]{$15$} (c2);

\path (c1) edge [->,bend left=30] node[anchor=south,inner sep = 0.2cm]{$11$} (c2);
\path (c1) edge [->,bend right=30] node[anchor=north,inner sep = 0.2cm]{$9$} (c2);
   \node[circle] (buf) at (10, 0.5) {$\cal{B}$};
   \draw[dashed] (10,2) ellipse  (0.6cm and 1cm);

\end{tikzpicture}
}

 $d_1 = 14$, $d_2 = 21$, $d_3 = 11$

$\downarrow$

 \scalebox{0.5}{

\begin{tikzpicture}
  \SetGraphUnit{5}
    \tikzset{
  EdgeStyle/.append style = {->} }
   \tikzstyle{VertexStyle}=[shape = circle, draw, minimum size = 30pt]
   \renewcommand{\VertexLightFillColor}{orange}
  \Vertex[x=0,y=4, L = {\huge $s_1$}]{s1};
  \Vertex[x=0,y=2, L = {\huge $s_2$}]{s2};
\Vertex[x=0,y=0, L = {\huge $s_3$}]{s3};
\Vertex[x=15,y=4, L = {\huge $t_1$}]{s1p};
\Vertex[x=15,y=2, L = {\huge $t_2$}]{s2p};
\Vertex[x=15,y=0, L = {\huge $t_3$}]{s3p};
\Vertex[x=10,y=2, L = {\huge $c_2$}]{c2};

  \Vertex[x=5,y=2, L = {\huge $c_1$}]{c1}

\path (s1) edge [->] node[anchor=south]{$0$} (c1);

\path (s2) edge [->] node[anchor=south,inner sep = 0.2cm]{$0$} (c1);
\path (s3) edge [->] node[anchor=south,inner sep = 0.2cm]{$0$} (c1);
\path (c2) edge [->] node[anchor=south,inner sep = 0.2cm]{$0$} (s2p);
\path (c2) edge [->] node[anchor=south,inner sep = 0.2cm]{$0$} (s3p);

\path (c2) edge [->] node[anchor=south,inner sep = 0.2cm]{$0$} (s1p);

\path (c1) edge [->] node[anchor=south,inner sep = 0.2cm]{$0$} (c2);

\path (c1) edge [->,bend left=30] node[anchor=south,inner sep = 0.2cm]{$1$} (c2);
\path (c1) edge [->,bend right=30] node[anchor=north,inner sep = 0.2cm]{$4$} (c2);
   \node[circle] (buf) at (10, 0.5) {$\cal{B}$};
   \draw[dashed] (10,2) ellipse  (0.6cm and 1cm);

\end{tikzpicture}
}

 $d_1 = 4$, $d_2 = 6$, $d_3 = 6$
 
\end{center}

\caption{Transformation of the star routed network of Figure~\ref{fig:star} into its canonical form, using the method of Proposition~\ref{prop:canonical}. Initially $\tau = 1$, $P=5$, $d_1 = 30$, $d_2 = 34$, $d_3 = 32$.}
\label{fig:canonical}
\end{figure}

  \begin{proposition}\label{prop:canonical}
   Let $I = (N, P, \tau , d)$, with $N = (\cal{R},\,\cal{B},\,\omega)$ a star routed network. Then, there is 
   $I' = (N', P, \tau , d')$, with  $N' = (\cal{R},\,\cal{B},\,\omega')$ a canonical star routed network, such that:
     $$I \in \pall \Leftrightarrow I' \in \pall \text{ and } I \in \pazl \Leftrightarrow I' \in \pazl$$
  \end{proposition}

  \begin{proof}
  We define $\omega'$ and $d'$ from $\omega$ and $d$ in such a way that there is a bijection 
  between valid assignments of $I$ and $I'$, which proves the proposition. In this bijection,
  the offsets $o_i$ for an assignment of $I$ will be mapped to $o'_i$, while the waiting times remain the same.
  
  The routed network $N'$ is equal to $N$ except for the weight function $\omega'$.
  We set the weights of the arcs $(s_i,c_1)$ to zero in $N'$. We obtain the bijection between valid assignments of $I$ and $I'$ by setting $o_i' + \omega(r_i,s_i) = o_i $ and $d'(r_i) = d(i) - \omega(r_i,s_i)$. The weights $\omega'(r_i,c_2)$ are also set to $0$. This does not change the possible collisions
  for an assignment, but it changes the transmission time. Hence we set $d'(r_i) = d'(r_i) - \omega(r_i,c_2)$
  to preserve the bijection between valid assignments of $I$ and $I'$. 

  We let $\omega'(r_i,c_1) = \omega(r_i,c_1) \mod P$. Again, it does not change collisions since computing a possible collision is done modulo $P$. However, we must change $d'$ to be $d'(r_i) = d'(r_i) - \omega(r_i,c_1) + \omega'(r_i,c_1)$ to keep the same constraints on the transmission times of valid assignments.

  Finally, we assume w.l.o.g. that $\omega'(r_0,c_1)$ is the smallest weight among the weights of the arcs
  $(c_1,c_2)$. We let $\omega'(r_i,c_1) = \omega'(r_i,c_1) - \omega'(r_0,c_1)$, which implies that $\omega'(r_0,c_1) = 0$.  All weights of arcs $(c_1,c_2)$ are changed by the same value, hence collisions are not modified. We change $d'(r_i)$ to  $d'(r_i) - \omega'(r_0,c_1)$ for all $i$ so that the constraints on the deadlines stay the same.
  \end{proof}

   From now on, we may assume that a star-routed network is canonical, using Proposition~\ref{prop:canonical}. To give an instance of \pall where the routed network is a canonical star routed network, it is enough to give the weights of the arcs $(c_1,c_2)$ for all routes, the period, the datagram size, and $d$ the deadline function. For an instance of \pazl, we can also omit $d$.

\section{Hardness of \texttt{PALL} and \texttt{PAZL}}
  \label{sec:complexity}

	We show in this section that \pall is $\NP$-hard by proving $\NP$-hardness for a restricted version: \pazl with $\tau = 1$. We give two proofs that \pazl is $\NP$-complete.
	The first proof works even for contention depth two, but not for star-routed networks.
	 For contention depth one, the problem is trivial: either the load is less than one and there is a valid bufferless assignment, or there is no valid assignment. 
	 The second proof works for graphs with contention width $2$: the conflicts are locally very simple, but the problem is complex globally nonetheless. Solving \pall is trivial on trees because they can be reduced to one vertex of contention depth one. Thus, it may be interesting to study its complexity on bounded treewidth (or dagwidth) networks, a common property of real networks~\cite{de2011treewidth}

 \begin{theorem}
\pazl is $\NP$-complete on the class of routed networks with contention depth 2.
\end{theorem}
 \begin{proof}
 \pazl is in $\NP$ since, given an offset for each route in an assignment, it is easy to check whether there are collisions, in linear time in the size of the routed network.
 
  Let $H=(V,E)$ be an undirected graph, and let $P$ be its maximum degree. We consider the problem to determine whether $H$ is arc-colorable with $P$ or $P+1$ colors. The arc coloring problem is $\NP$-hard~\cite{holyer1981np}, and we reduce it to \pazl to prove its $\NP$-hardness. To do that, we define from $H$ a routed network $N = ({\cal R},\, \omega)$ as follows. 

  Let us choose an arbitrary total order $<$ on $V$.
  For each edge $(u,v) \in E$, if $u<v$, there is a route $s_{u,v},u,v,t_{u,v}$ in ${\cal R}$. 
  All these arcs are of weight $0$. The routed network $N$ is of contention depth $2$, as required by the theorem statement. 

  The existence of a $P$-coloring of $H$ is equivalent to the existence of a $(P,1)$-periodic bufferless assignment of $N$. Indeed, a $P$-coloring of $H$ can be seen as a labeling of its edges by the integers in $[P]$. It induces a bijection between $P$-colorings of $H$ and offsets of the routes of ${\cal R}$, which represent the edges of $H$. Having no collision on some vertex $v$ implies that all offsets of routes going through $v$ are different since all arcs are of weight $0$. Hence, edges of $H$ incident to $v$, colored by the offsets of a valid assignment, are all of distinct colors. Therefore, we have reduced arc coloring to \pazl by a polynomial-time transformation, which concludes the proof. 
 \end{proof}

 We have used weights of zero for all arcs in the proof. It is a further restriction to the 
 class of graphs for which \pazl is $\NP$-hard. We could ask the weights to be strictly positive, another possible restriction that makes more sense in our model, since weights represent the delay of physical links. Then, we can prove $\NP$-completeness using the same proof, by setting all weights to the period $P$.

We now give a hardness proof for routed networks with contention width two but large contention depth. Note that a vertex of contention depth one does not induce a collision and can be removed from the routed network without loss of generality. The presented reduction can be used to prove an inapproximability result. Let \minpazl be the following problem: given a routed network and $\tau$, find the minimal period $P$ such that there is a $(P,\tau)$-periodic bufferless assignment (a positive instance of \pazl).

\begin{theorem}\label{th:inapprox}
If $P$ $\neq$ $\NP$, the problem \minpazl on the class of routed networks of contention width two cannot be approximated in polynomial time within a factor $n^{1-o(1)}$ where $n$ is the number of routes.
\end{theorem}

\begin{proof}
 We reduce the problem of finding the minimal vertex coloring of a graph to \minpazl. Let $H = (U,E)$ be a graph, an instance of the problem of finding a minimal vertex coloring.  Let us now define the routed network $N$ from $H$.
 
 Let $<$ be an arbitrary total order on $U$. 
 The vertices of $N$ are in the set $\{v_{u,w} \mid (u,w) \in E\} \cup \{u^1, u^2 \mid u \in U\}$. 
 For each vertex $u$ in $H$, there is a route $r_u$ in ${\cal R}$, whose first and last vertices
 are $u^1$ and $u^2$. In between, the route contains all vertices $v_{u,w}$, following the order $<$ on the vertices $w$. The weights of all arcs are zero. By construction, a contention vertex corresponds to an edge and belongs to exactly two routes representing the vertices of the edge, thus $N$ is of contention width $2$. This reduction is illustrated in Figure~\ref{fig:reductionminpazl}.

  The existence of a $P$-coloring of $H$ is equivalent to the existence of a $(P,1)$ assignment of $N$ without waiting time: the offset of a route can be identified with the color of the corresponding vertex. Indeed, since all weights are zero, the absence of collision at contention point $v_{u,w}$ is equivalent to the fact that the offsets of $r_u$ and $r_w$ are different and reciprocally.

   Therefore, if we can approximate the minimum value of $P$ within some factor such that there is a $(P,1)$ assignment, we could approximate the minimal number of colors needed to color a graph within the same factor. The proof follows from the hardness of approximability of finding a minimal vertex coloring~\cite{zuckerman2006linear}.
\end{proof}
    \begin{figure}[ht]
    \centering
    \scalebox{0.6}{
    \begin{tikzpicture}
    \tikzset{
      LabelStyle/.style = { rectangle, rounded corners, draw,
        font = \bfseries },
    EdgeStyle/.append style = {->} }
      \SetGraphUnit{5}

      \node[draw,circle] (s3) at (4, 2) {\Large $w^1$}; 
      \node[draw,circle] (s2) at (0, 4) {\Large $v^1$}; 
      \node[draw,circle] (s1) at (0, 6) {\Large $u^1$}; 

      \node[draw,circle] (t3) at (12, 3) {\Large $w^2$}; 
      \node[draw,circle] (t2) at (10, 5) {\Large $v^2$}; 
      \node[draw,circle] (t1) at (10, 2) {\Large $u^2$};

      \tikzstyle{VertexStyle}=[shape = circle, draw, minimum size = 20pt]
  \tikzset{
   VertexStyle/.append style = {blue} }
  \Vertex[x=-8,y=3, L = {\huge $u$}]{u};
        \tikzset{
      VertexStyle/.append style = {green} }
    \Vertex[x=-7,y=5, L = {\huge $v$}]{v}

      \tikzset{
      VertexStyle/.append style = {red} }
    \Vertex[x=-6,y=4, L = {\huge $w$}]{w}
    \tikzset{
      VertexStyle/.append style = {black} }

       \SetVertexNoLabel
       \Vertex[x=2,y=5]{A}

       \Vertex[x=6,y=3]{E}

       \tikzset{
       EdgeStyle/.append style = {green} }
       \Edge(s2)(A)

      \Edge(A)(t2)

       \tikzset{
      EdgeStyle/.append style = {red} }
       \Edge(s3)(E)
       \Edge(E)(t3) 
  \tikzset{
       EdgeStyle/.append style = {blue} }
       \Edge(s1)(A)

       \Edge(A)(E)

       \Edge(E)(t1)

  \tikzset{
       EdgeStyle/.append style = {black,-} }

       \Edge(u)(v)
       \Edge(u)(w)
     \node (1) at (-3,4){\Huge $\rightarrow$};

     \node (2) at (-7,0){\Huge $H$};
      \node (3) at (10,0){\Huge $N$};
     \end{tikzpicture}
     }
     \caption{Reduction from  vertex coloring to \minpazl}
     \label{fig:reductionminpazl}
     \end{figure}

The previous theorem implies that \pazl is $\NP$-complete on the class of routed networks with contention width two. This also underlines the fact that, for general graphs, the best $P$ such that there is a $(P,\tau)$ assignment may correspond to a very small load. We can build on the reduction of the previous theorem to prove that \mintra, the problem of minimizing $TR(A)$, is hard to approximate too.

\begin{theorem}
If $\P{ }\neq\NP$, the problem \mintra, on graphs of contention width two, cannot be approximated in polynomial time within a factor $n^{1-o(1)}$ where $n$ is the number of routes.
\end{theorem}

\begin{proof}
We reduce the problem of finding the minimal vertex coloring of a graph to \mintra.
 Let $H = (U,E)$ be a graph, instance of the problem of finding a minimal vertex coloring. 
 We define the routed network $N$ in two steps. 

 Let the elements of $U$ be $u_0,\dots, u_{n-1}$. There are $n$ routes in $N$, denoted by $r_i$ for $i \in [n]$. In their first part, they go from $u_i^0$ to $u_i^1$, through some vertices in $\{v_{i,j,k}\}_{i,j,k \in [n]}$ that we later define. Moreover, $u_i^1 \in \cal{B}$, hence the waiting time is added at $u_i^1$. Assume that $r_i$ has offset $o_i$ and $r_j$ has offset $o_j$ and let us fix the datagram size to $1$ and the period to $n$. If $r_i$ and $r_i$ go through some vertex $v_{i,j,k}$, and  $\lambda(r_i,v_{i,j,k}) = \lambda(r_j,v_{i,j,k}) + k$, then to avoid a collision, the equation $o_i \neq o_j + k \mod n$ must be satisfied. If $r_i$ and $r_j$ go through $v_{i,j,k}$ satisfying the previous constraints for all $k \neq l$, it implies $o_i = o_j + l \mod n$. 
 It is easy to choose the weights of the two arcs going to $v_{i,j,k}$ to realize the previous condition, whatever the choice of weights of the previous arcs of the routes $r_i$ and $r_j$.

We ensure, using the vertices $v_{i,j,k}$ for $k \neq i-j$,
that $o_{i} = o_{j} + i - j \mod n$. It implies that there is some $o$, such that 
$ o = o_{i} - i \mod n$ for all $i \in [n]$. Now, for each route $r_i$, we set the weight of the
arc going to $v_i^1$, from the last vertex of the form $v_{i,j,k}$ in $r_i$, to be $n-i$.
With this construction, we have ensured, that the datagram of $r_i$ arrives at 
$v_i^1$ at time $o$ modulo $n$, for all $i \in [n]$. 

The second part of the routes, from $v_i^1$ to $v_i^2$ is built exactly as in the proof of Theorem~\ref{th:inapprox}. Hence, the waiting time in the vertices $v_i^1$ plays the exact same role as the offset in the graph of Theorem~\ref{th:inapprox}: the valid $(n,1)$-assignments are in bijection with colorings of $H$, the waiting times corresponding to the colors.

Finally, set the weights of the last arc going to $v_i^2$, for all $i \in [n]$, such that, for all $i,j \in [n]^2$, $\lambda(r_i) = \lambda(r_j)$.  Since all routes are of the same size, $TR(A)$ is equal to the maximal waiting time of $A$. Hence, the maximum waiting time is equal to the number of different waiting times required to have a valid assignment. A valid $(n,1)$-assignment which minimizes $TR(A)$ is in bijection with a minimal proper coloring of $H$, which proves the theorem.
\end{proof}

	We would like to prove hardness for even more restricted networks, in particular star routed networks.
   The problem \pazl on star routed networks is similar to the minimization of makespan in a two flow-shop with delays (see Section~\ref{sec:wtaheuristic}), a problem known to be $\NP$-complete~\cite{yu2004minimizing}. It suggests that \pazl is $\NP$-complete on star routed network, however we have not been able to prove it yet,  because the makespan cannot easily be encoded in \pazl. If we relax the definition of routed network by allowing loops,  we can model a network with a single half-duplex shared link, that is collisions can happen between datagrams going in both directions. This variant can be shown to be $\NP$-complete by a reduction from the subset sum problem, as it is done for a similar problem of scheduling pairs of tasks~\cite{orman1997complexity}.

We show in Section~\ref{sec:exp_PAZL}, that \pazl can be very often solved positively, in particular for short routes and when the load is moderate. The dependency of \pazl on the load has been studied in detail in a follow-up work~\cite{guiraud2024scheduling}, in which \pazl is solved for higher loads using more involved polynomial-time algorithms.

\section{Solving \texttt{PAZL} problem}
\label{sec:pazl}
\subsection{Shortest-Longest Policy}

We first present a simple policy, which works when the period is large, relative to the lengths of the routes. More generally, it works as soon as the length of the routes modulo the period are close. The algorithm is called \shortestlongest: it sends datagrams on the shared link from the route with the shortest arc $(c_1,c_2)$ to the longest. There is no idle time in the contention point $c_1$, i.e. a datagram goes through $c_1$ right after the previous one has left $c_1$.

\begin{proposition} Let $N$ be a canonical star routed network, with $r$ the longest route. If $n\tau + \lambda(r) \leq P$ then \shortestlongest produces a $(P,\tau)$-periodic bufferless assignment of $N$ in time $O(n\log(n))$.\label{prop:SL}
\end{proposition}
\begin{proof}
By hypothesis, $N$ is in canonical form, hence $\lambda(r,s_i) = 0$ for all $i \in [n]$. Moreover, $\lambda(r_0) = 0$ and we assume the routes are sorted so that, for all $i$, $\lambda(r_i) \leq \lambda(r_{i+1})$ (equivalently $\omega(r_i,c_1) \leq \omega(r_{i+1},c_1))$. We fix $P$ and $\tau$. The algorithm \shortestlongest sets $o_{r_i} = i\tau$ for all $i \in [n]$. Then, $[r_{i},c_1] = \{i\tau,\dots, (i+1)\tau -1\}$ and since $n\tau < P$, there is no collision on $c_1$.

By definition, we have $[r_{i},c_2] = \{\lambda(r_{i}) + i\tau \mod P, \dots, \lambda(r_{i}) + (i+1)\tau -1 \mod P\}$. By hypothesis, $n\tau + \lambda(r_{n-1}) \leq P$, hence $[c_2,r_{i}] = \{\lambda(r_{i}) + i\tau, \dots, \lambda(r_{i}) + (i+1)\tau -1\}$. Since $\lambda(r_i) \leq \lambda(r_{i-1})$, we have proven that $[c_2,r_{i}] \cap [c_2,r_{j}]$ for $i \neq j$. Hence, there is no collision on $c_2$ and the $(P,\tau)$ assignment built by \shortestlongest is valid.

The algorithm's complexity is dominated by the sorting of the routes in $O(n\log(n))$.
\end{proof}

If the period is slightly smaller than the bound of Proposition~\ref{prop:SL}, there is a collision of the last route with $r_0$ on $c_1$. Hence, this policy is not useful as a heuristic for longer routes, as confirmed by the experimental results of Section~\ref{sec:exp_PAZL}.

\subsection{Greedy Algorithm}

We propose a greedy algorithm that tries to build a valid assignment and always succeeds when the load is less than $1/3$. Therefore, in the rest of the article, we are only concerned with loads larger than $1/3$. In fact, in a follow-up work~\cite{guiraud2024scheduling}, we prove that there is always an assignment for loads smaller than $0.4$ and with a high probability for loads less than $0.5$. In this article, we present only the simplest greedy method to solve \pazl and focus on its comparison to the previous method and to an exact algorithm we present in the next section.

The idea is to restrict the possible offsets which can be chosen for the routes. It seems counter-intuitive since it decreases artificially the number of available offsets to schedule new datagrams. However, it does reduce the number of forbidden offsets for unscheduled datagrams. A \textbf{meta-offset} is an offset of value $i\tau$, with $i$ an integer from $0$ to $P / \tau$. We call \metaoffset the greedy algorithm which works as follows: for each datagram, in the order they are given, it tries all meta-offsets from $0$ to $P/\tau$ as an offset for the assignment until one does not create a collision with the current partial assignment.

    \begin{theorem}
    \metaoffset solves \pazl positively on star routed network and load less than $1/3$. 
    The assignment is found in time $O(n^2)$.
    \end{theorem}
    \begin{proof}
    Let us prove that \metaoffset always schedules the $n$ routes when the load is less than $1/3$. Let us assume it has built an assignment for the routes $r_0,r_1,\ldots,r_{k-1}$, using only meta-offsets. The number of meta-offsets is $P/\tau$ and already $k$ of them are used, hence to avoid collision in $c_1$, we have $P/\tau - k$ choices. We choose an offset among those for the route $r_k$ so that there is no collision in $c_2$. Exactly two consecutive meta-offsets can create a collision between $r_k$ and some route $r_i$ with $i < k$ in $c_2$, since the datagrams are all of size $\tau$; see Figure~\ref{fig:metaoffset}. Hence, there are at most $2k$ meta-offsets forbidden by collisions in $c_2$. In conclusion, there are at least $P/\tau - k - 2k$ possible meta-offsets so that its choice for $r_k$ does not create a collision in $c_1$ or $c_2$.  \metaoffset terminates and provides a valid bufferless assignment as soon as $P/\tau - 3(n-1) > 0$, which is true when $n\tau /P \leq 1/3$: the load is less or equal to $1/3$.

     This algorithm works in time $O(n^2)$, since for the $k$th route we have to try at most $3k$ meta-offsets before finding a correct one. We can test whether these $3k$ offsets cause a collision in $c_2$ in time $O(k)$ by maintaining an ordered list of the intervals of tics in the period used by already scheduled routes in $c_2$.
     \end{proof}
         
     \begin{figure}
      \begin{center}
      \includegraphics[width=0.9\textwidth]{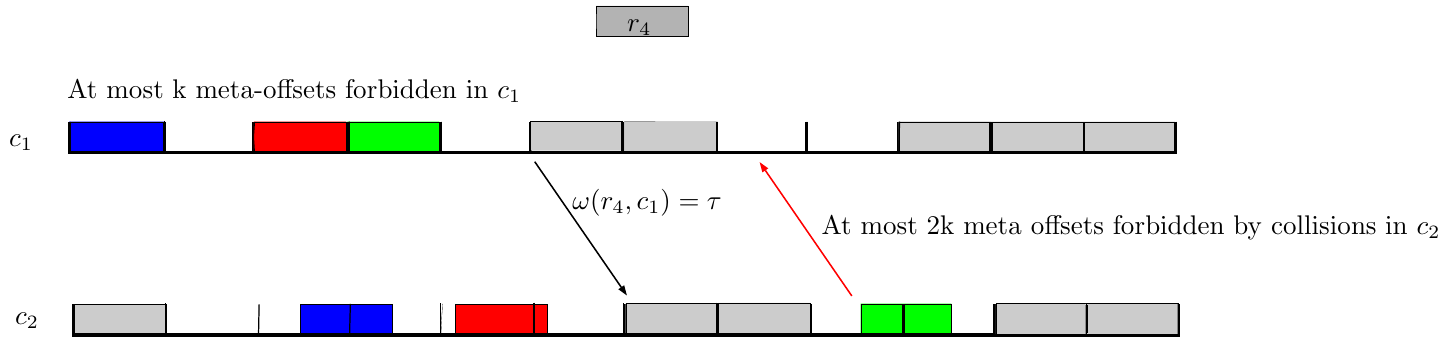}
      \end{center}
      \caption{Times used in the period in $c_1$ and $c_2$, when scheduling the $4$th route in \metaoffset, with a period $P = 3n\tau$}
      \label{fig:metaoffset}
      \end{figure}					
	
This algorithm, in contrast to the previous one, may work well, even for loads higher than $1/3$.
Experimental data in Section~\ref{sec:exp_PAZL} suggest that the algorithm finds a solution when the load is less than $1/2$.

\subsection{Compact Assignment}

In this section, we show how every bufferless assignment can be put into a canonical form.
We use that canonical form to design an algorithm to solve \pazl in fixed-parameter tractable time ($\FPT$), with parameter $n$ the number of routes (for more on parameterized complexity see~\cite{downey2012parameterized}). This is justified since $n$ is small in practice, from $10$ to $20$ in our settings, and the other parameters such as $P$, $\tau$, or the weights are large.

Let $({\cal R},\omega)$ be a star-routed network and let $A$ be a bufferless $(P,\tau)$ assignment.
We say that $A$ is \textbf{compact} if there is a route $r_0 \in \cal{R}$ such that the following holds: for all subsets $S\subset \cal{R}$ with $r_0 \notin S$, the bufferless assignment $A'$, defined by $A'(r) = A(r) - 1 \mod P$ if $r \in S$ and $A(r)$ otherwise, is not valid. In other words, an assignment is compact if, for all routes $r$ but one, $A(r)$ cannot be reduced by one; that is, either in $c_1$ or in $c_2$, there is a route $r'$ using the tic just before $r$. See Figure~\ref{fig:compact} for an example of a compact assignment obtained by the procedure of the next proposition. 
  \begin{figure}
      \begin{center} 
      \includegraphics[width=\textwidth]{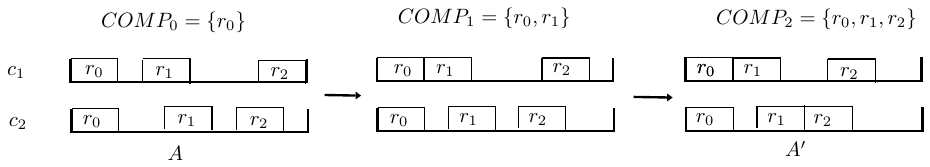}
      \end{center}
      \caption{Transformation of a bufferless assignment $A$ into a compact assignment $A'$, following the process of Proposition~\ref{prop:compactification}}
      \label{fig:compact}
      \end{figure}
\begin{proposition}\label{prop:compactification}
Let $N = ({\cal R}, \omega)$ be a star-routed network. If there is a $(P,\tau)$-periodic bufferless assignment of $N$, then there is a compact $(P,\tau)$ assignment of $N$.
\end{proposition}
\begin{proof}
Consider $A$, a $(P,\tau)$-periodic bufferless assignment of $N$.
We describe an algorithm that builds a sequence $COMP_i$ of sets of routes and a sequence  
$A_i$ of valid bufferless assignments. For all $i \leq n$, the set $COMP_i$ has cardinal $i$ and satisfies  $COMP_{i-1} \subset COMP_i$

Let $r$ be an arbitrary route of ${\cal R}$ and $A_0 = A$. We set $COMP_0 = \emptyset$.
 For $i = 1$ to $n$, we choose a route $r$, denoted by $r_i$, as follows.
Let $A_{i} = A_{i-1}$. While there is no collision, for all routes $r \in {\cal R} \setminus COMP_{i-1}$, let $A_i(r) = A_i(r) - 1$. Then choose any route $r$ in ${\cal R} \setminus COMP_{i-1}$ such that setting $A_i(r) = A_i(r)-1$ creates a collision and let $r_i = r$. By construction, $A_i$ is a valid bufferless assignment, since it is modified only when no collision is created. We let $COMP_i = COMP_{i-1} \cup \{r_i\}$.

We prove by induction on $i$, that $A_i$ is compact when restricted to $COMP_{i}$.
We have $|COMP_1| = 1$, hence $A_1$ is compact over $COMP_1$. Let us consider $A_i$. By the induction hypothesis, since the offsets of routes in $COMP_{i-1}$ are not modified at step $i$ of the algorithm, $A_i$ is compact when restricted to $COMP_{i-1}$. 

 Consider $S \subseteq COMP_i$ which does not contain $r_0$. If $S$ contains
an element of $COMP_{i-1}$, then $S \setminus \{r_i\}$ is not empty and by compactness we cannot decrement all offsets of $S\setminus \{r_i\}$ without creating a collision. The same property is true for $S$. If $S = \{r_i\}$, then by construction of $r_i$ by the algorithm, removing one from $A_i(r_i)$ creates a collision. Hence, $A_i$ is compact restricted to $COMP_{i}$, which proves the induction and the proposition.
\end{proof}

We now present an algorithm to find a $(P,\tau)$ assignment by trying all compact assignments.

\begin{theorem}\label{th:FPT}
$\pazl \in \FPT$ over star-routed networks when parametrized by the number of routes.
\end{theorem}
\begin{proof}
Let $N = ({\cal R},\omega)$ be a canonical star-routed network and let $P$ be the period and $\tau$ the size of a datagram. For a given assignment and a route $r$ with offset $o_r$, by removing $o_r$ to all offsets, we can always assume that $o_r = 0$. By this remark and Proposition~\ref{prop:compactification}, we need only to consider all \emph{compact assignments} with an \emph{offset $0$} for the route $r_0$. We now evaluate the number of compact assignments and prove that it only depends on $n$ the number of routes to prove the theorem.

 We describe a way to build any compact assignment $A$ by determining its offsets one after the other, which gives a bound on their number and an algorithm to generate them all. We fix an arbitrary total order on ${\cal R}$. Let $r_0$ be the first route in this order; its offset is set to zero, and we let $S = \{r_0\}$,
 $S_1 = \{r_0\}$ and $S_2 = \{r_0\}$. $S$ represents the routes whose offsets are fixed. The offsets of unscheduled routes are chosen so that they follow a route of $S_1$ in $c_1$ or a route of $S_2$ in $c_2$.

 At each step, we add an element to $S$: let $r$ be the smallest element of $S_1$, if it is non-empty. Then, select any route $r' \in {\cal R} \setminus S$ 
 such that $o_{r'} = o_{r} + \tau$ does not create a collision (by construction $o_{r'} = o_{r} + \tau - 1$ does create a collision in $c_1$). Then, we update the sets as follows:
 $S = S \cup \{r'\}$, $S_1 = S_1 \setminus \{r\} \cup \{r'\}$ and $S_2 = S_2 \cup \{r'\}$. If 
 $S_1$ is empty, $r$ is the smallest element of $S_2$, and we set $o_{r'} = o_{r} + \tau + \omega(r,c_2) - \omega(r',c_2)$.
 We can also remove $r$ from $S_1$ (or from $S_2$ if $S_1$ is empty) without adding any element to $S$. The value of the offset of the route added to $S$ is entirely determined by the values of the offsets of the routes in $S$.

 Any compact assignment can be built by the previous procedure if the proper choice of an element to add is made at each step. Hence, this process generates all compact assignments. We now bound the number of compact assignments it can produce. When $|S| = i$, we can add any of the $n-i$ routes in ${\cal R} \setminus S$ to $S$. Hence, the number of sequences of choices of routes to add is $n!$ (but some of these sequences can fail to produce a valid assignment). We have not yet taken into account the steps at which an element is removed from either $S_1$ or $S_2$, without adding something to $S$. At each step of the algorithm, we can remove an element or not. There are at most $2n$ steps in the algorithm. Hence, there are at most $4^n$ sequences of such choices during the algorithm. To conclude, there are at most $4^nn!$ compact assignments.

The algorithm to solve \pazl builds every possible compact assignment in the incremental manner described here and tests at each step whether, in the built partial assignment, there is a collision, which can be done in time linear in the size of $N$. Therefore, $\pazl \in \FPT$.
\end{proof}

We call the algorithm described in Theorem~\ref{th:FPT} \textbf{Exhaustive Search of Compact Assignments}
or \ESCA. The complexity of \ESCA is in $O(4^n n!)$. While a better analysis
of the number of compact assignments could improve this bound, the simple star routed networks with all arcs of weights $0$ has $(n-1)!$ compact assignments. Hence, to improve significantly on \ESCA, one should find an even more restricted notion of bufferless assignment than compact assignment.

To make \ESCA more efficient in practice, we make cuts in the search tree used to explore all compact assignments. Consider a set $S$ of $k$ routes whose offsets have been fixed at some point in the search tree. We consider the times used by these routes in $c_1$. It divides the period into $[(a_0,b_0), \dots, (a_{k-1},b_{k-1})]$ where the intervals $(a_i,b_i)$ are the times not used yet in $c_1$. Therefore, at most $\displaystyle{ \sum_{i=0}^{k-1} \lfloor(b_{i} -a_i)/\tau\rfloor}$ routes can still send a datagram through $c_1$. If this value is less than $n - k$, it is not possible to create a compact assignment by extending the current one on $S$, and we backtrack in the search tree. The same cut is also used for the contention point $c_2$. These cuts rely on the fact that the partial assignment is wasting bandwidth by creating intervals that are not multiples of $\tau$. These cuts significantly speed up \ESCA on instances of large loads, which are also the most difficult to solve.

   \subsection{Experimental Evaluation}\label{sec:exp_PAZL}

  In this section, the experimental results of the three presented algorithms are compared.
   Notice that both \metaoffset and \shortestlongest are polynomial-time algorithms but are not always able to find a solution, depending on the load or the size of the routes. On the other hand, \ESCA finds a solution if it exists, but works in exponential time in $n$. We compare the performance of the algorithms in two different regimes: routes are either short (depending on $\tau$), or unrestricted.

   \paragraph{Experimental Settings}

     The default parameters of all experiments in this article are derived from the C-RAN context~\cite{wang2017cloud} and summarized in the table of Figure~\ref{tab:params}: a tic corresponds to the sending time of $64$ Bytes of data on links of bandwidth $10$~Gbps. The datagrams have an approximate size of $1.18$~Mbit, which corresponds to $2,500$ tics. 

     All experiments are done on synthetic networks generated randomly. We generate the physical fronthaul
     network, as represented in Figure~\ref{fig:star}, by drawing the size of each link according to some distribution that depends on the experiment. Then, the corresponding canonical star routed network is built from the generated fronthaul and the algorithms tested on it. 

     In the following experiments, we illustrate how well the algorithms work for different values of the load. To change the load, we choose to fix both parameters $\tau$ and $n$, and to modify the period $P$, which allows for smooth control of the load and does not impact the execution time of the algorithms. In most experiments, we fix the number of routes to $n = 8$. In the case of fronthaul networks, the period is one ms, which corresponds to $20,000$ tics and thus a load of $0.95$ with eight routes.

     For all experiments of this paper, the code in C is available on the web page of one author~\cite{webpage} under a copyleft license. The code has been run on a standard laptop with a $2.2$~GHz Intel Core i5 and the sources are compiled with gcc version 11.2. All experiments end in at most a few dozen seconds.
\begin{figure}
\begin{center}
\begin{tabular}{|c|c|c|}
\hline
Parameter& Value & Time in tics \\
\hline
Duration of a tic& $\simeq51$~ns&1 tic\\
\hline
Datagram size ($\tau$)&  $1.18$~Mbit & $2,500$ tics\\
\hline
Period ($P$)& $1$ ms&$\simeq20,000$ tics\\
\hline
Bandwidth of links &  $10$~Gbps & -\\
\hline
Number of routes ($n$) & $8$ & -\\
\hline
\end{tabular}

\end{center}

\caption{Parameters of experiments on realistic network topologies}
\label{tab:params}
\end{figure}
    \paragraph{Short Routes}

We first consider routes that are shorter than $\tau$: a datagram cannot be contained completely in a single arc, which is common in our applications. We generate random star-routed networks by drawing uniformly at random the weights of the arcs of the fronthaul network in $[700]$, which corresponds to links of length less than $5$~km between a BBU and an RRH.

In the following experiment, we generate $10,000$ random instances of \pazl for a load of $1$ down to $0.4$. We represent, in Figure~\ref{fig:short}, the percentage of success of each algorithm as a function of the load. We make three experiments with $8$, $12$, and $16$ routes to understand the effect of the number of routes on the quality of our algorithms. A bound on the maximal success rate is given by the exhaustive search, which always finds a solution if there is one.

\begin{figure}[h]
\begin{center}
\includegraphics[width=0.9\textwidth]{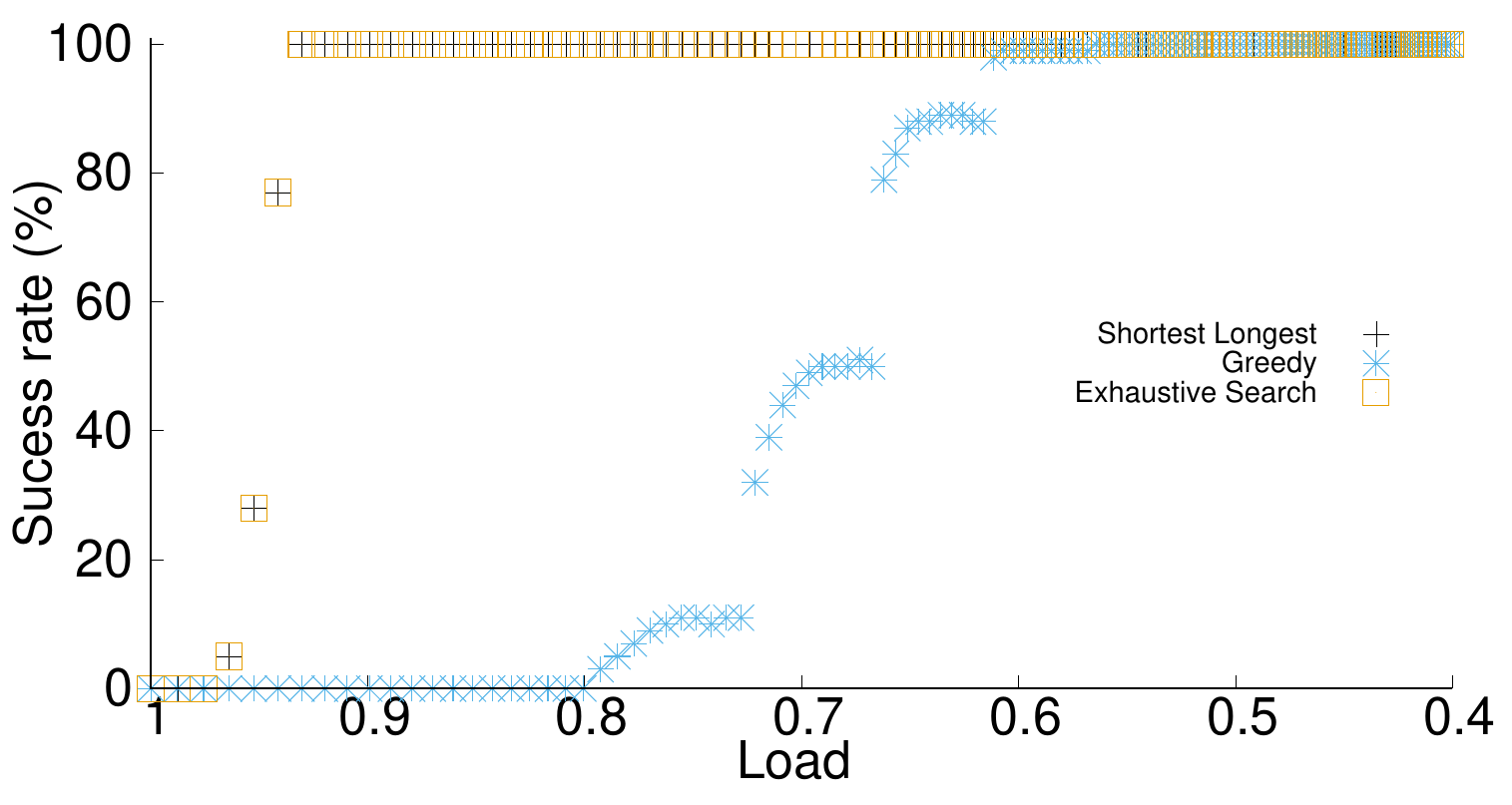}
\vspace{1cm}

\includegraphics[width=0.45\textwidth]{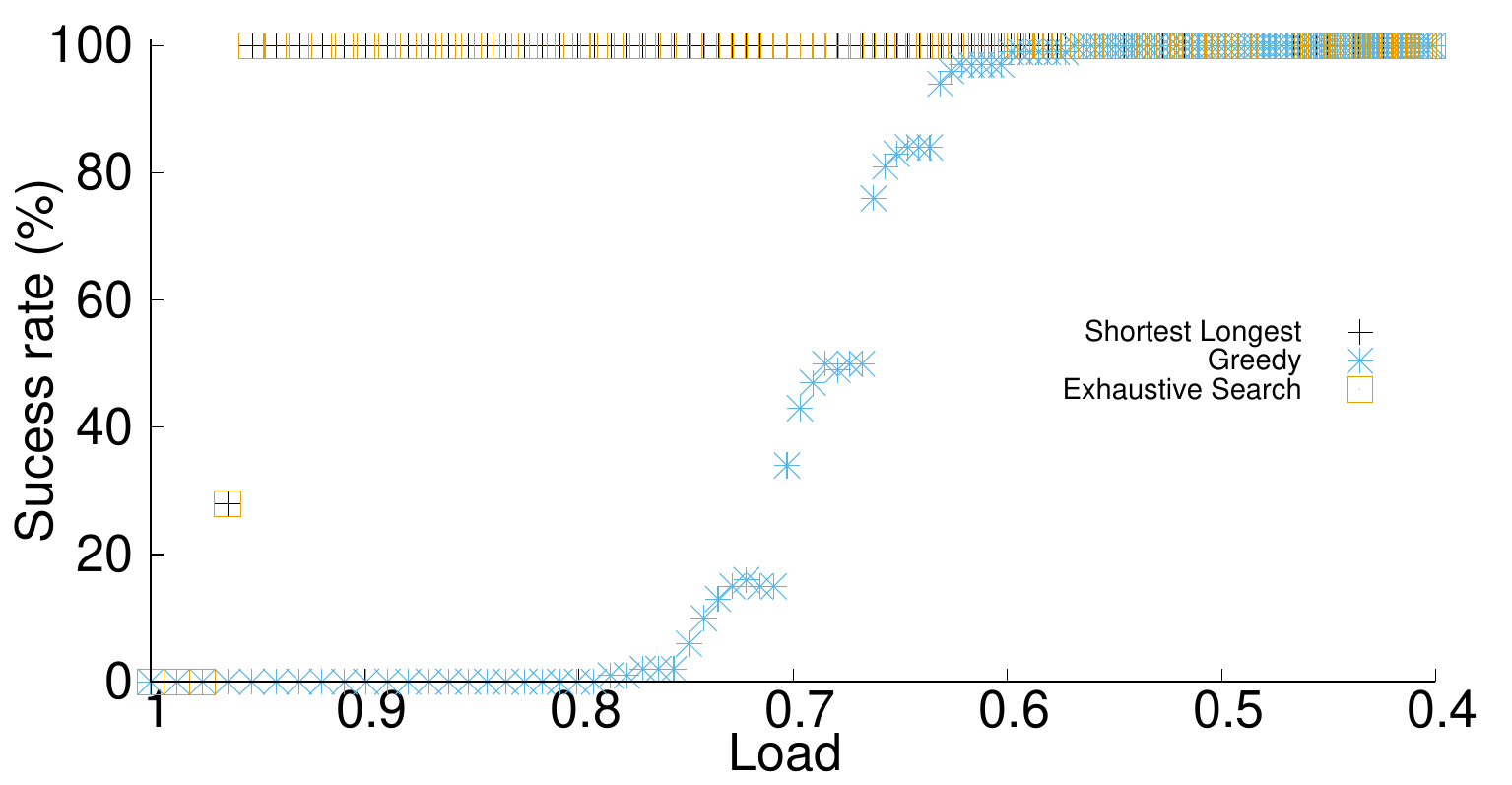}
\includegraphics[width=0.45\textwidth]{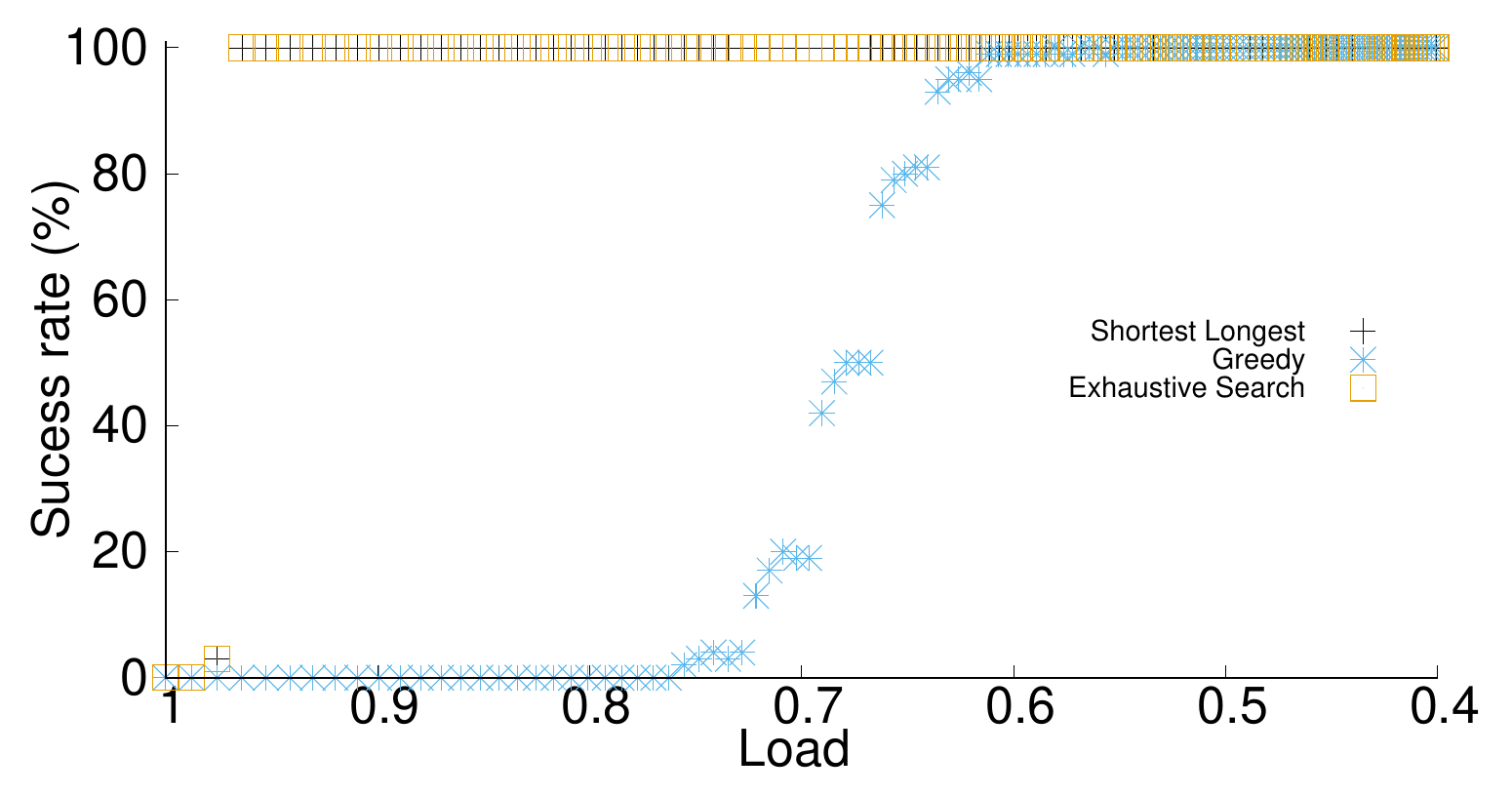}
\end{center}
\caption{Success rate of the three algorithms solving \pazl, for short routes and $8$ routes (top), $12$ routes (bottom left) and $16$ routes (bottom right)}\label{fig:short}
\end{figure}

First, \ESCA finds a solution even when the load is high. It justifies the idea to look for a bufferless assignment in this short routes regime.
It seems that increasing the number of routes makes the exhaustive search even more efficient, meaning that the more the routes, the more instances have a bufferless assignment.
Second, \shortestlongest is as good as the exhaustive search. While it was expected to be good with short routes (see Proposition~\ref{prop:SL}), it turns out to be optimal for all the random star routed networks we have tried. Therefore, we should use it in practical applications with short routes, instead of the exhaustive search, which is much more computationally expensive.

Finally, the greedy algorithm seems to always work when the load is less than $1/2$ and has a good probability to work up to a load of $2/3$, which is twice as good as the theoretical bound. The performance of \metaoffset seems to depend on the load only and not on the number of routes. There are discontinuities in the probability of success at several loads, which seem to smooth out when the number of routes increases. It can be explained by the fact that \metaoffset becomes better when decreasing the load makes the number of available meta-offsets larger. The number of meta-offsets increases when $\tau$ is added to the period, which is more frequent when there are more routes.
      
        \paragraph{Long routes}
      
      We now want to understand the performance of these algorithms when the length of the routes is unbounded. In this experiment, we fix the number of routes to eight. The weights of the arcs of the fronthaul network are drawn following a uniform distribution in $[P]$. We represent in Figure~\ref{fig:long} the percentage of success of each algorithm computed from $10,000$ random instances, for load from $1$ down to $0.4$.
\begin{figure}[h]

       \begin{center}
      \includegraphics[width=0.9\textwidth]{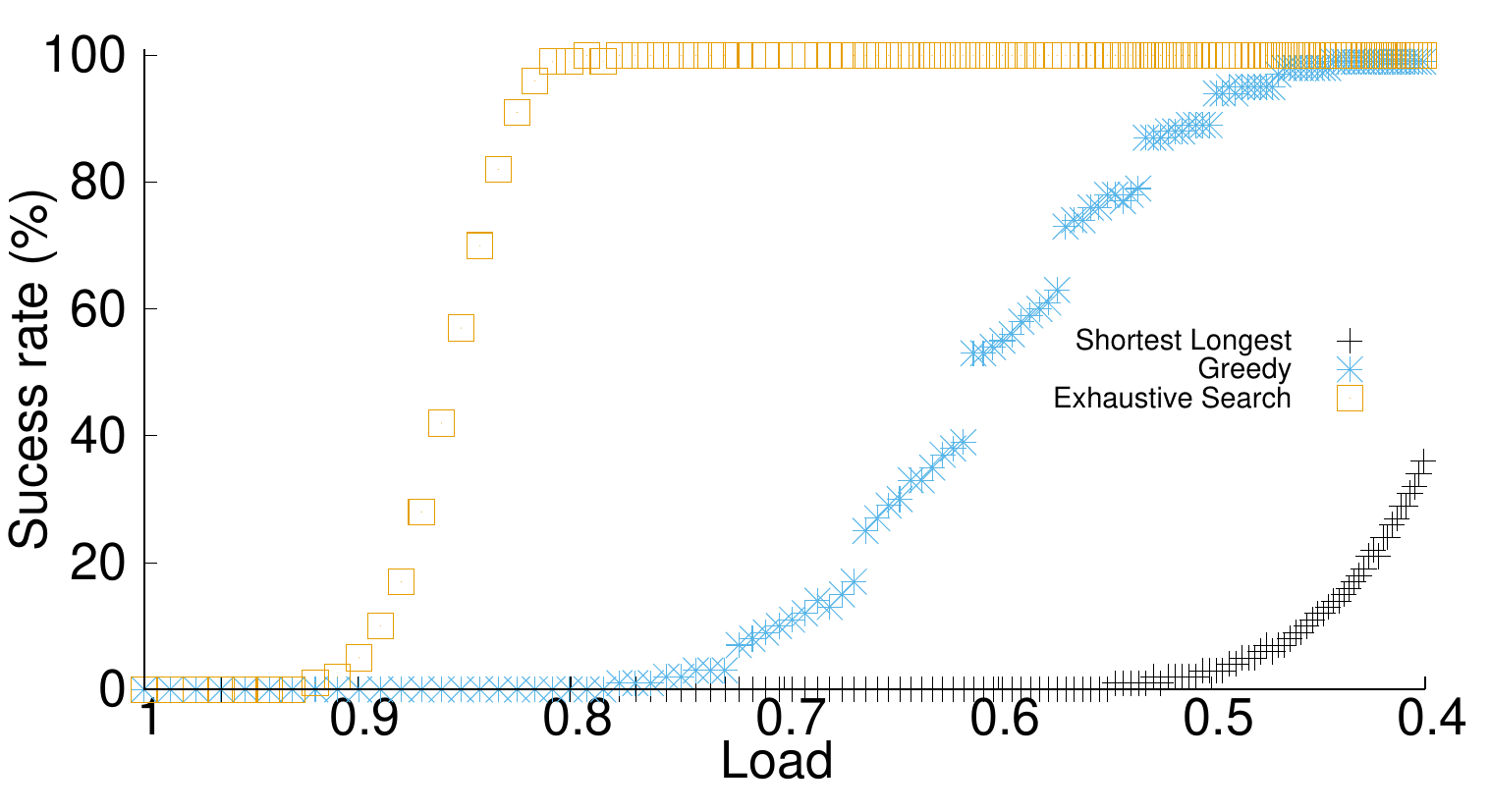}
      \end{center}
        
      \caption{Success rate of the three algorithms solving \pazl, for $8$ routes of arbitrary length}\label{fig:long}
     \end{figure}
      
       In this regime, the performance of \shortestlongest is abysmal because it depends on the difference of size between the longest and the smallest route, which is large here.  Algorithm \metaoffset has a performance not far from the short routes regime, which is expected since it does not directly depend on the size of the route. 
      
       When the load is larger than $0.5$, the \ESCA finds more solutions than \metaoffset which justifies its use. However, for loads larger than $0.8$ there are instances for which there is no solution to \pazl. It means that with long routes and high load, looking for a bufferless assignment is far too restrictive. This justifies the design of algorithms for the general \pall problem, which we present in the next section. We will test them on $8$ long routes and a load between $1$ and $0.8$, parameters for which, as shown here, there is not always a bufferless assignment.
      
       The computation time of \ESCA is bounded by $O(4^nn!)$ as shown in Theorem~\ref{th:FPT}, but it can be much better in practice, either because it finds a solution quickly or because a large part of the tree of compact assignments is pruned during the algorithm. We study the evolution of the running time  of the algorithm when $n$ grows in the following experiment. The weights of the arcs are drawn following a uniform distribution from $[P]$ and the load is set to $0.95$.  The table of Figure~\ref{fig:table} shows the time before \ESCA ends, for $8$ to $16$ routes, averaged over $100$ random star routed networks. This shows that for less than $20$ routes, which corresponds to all current topologies, the algorithm is efficient enough, but we should improve it further to work on more routes.
       
             \begin{figure}[h]
         \begin{center}
         \begin{tabularx}{\textwidth}{|l|X|X|X|X|X|}
    \hline
   $n$ & $8$ & $10$& $12$&$14$& $16$\\
    \hline
   Time (s) & $6.10^{-5}$&$8.10^{-4}$&$2.10^{-2}$& $0.4$& $11$\\
    \hline
      \end{tabularx}
      \end{center}
      \caption{Running time of \ESCA, averaged over 100 random instances}
      \label{fig:table}
      \end{figure}
      
         \section{Solving \texttt{PALL} on Star Routed Networks}\label{sec:PALL}
    
    In this section, we consider the more general \pall problem on star routed networks. The datagrams are allowed to wait in the BBUs to yield more possible assignments. Hence, we allow the transmission time of a route to be greater than the length of the route, but it must be bounded by its deadline.

	\subsection{Simple Star Routed Networks}

	Often in real networks, the lengths of the routes are not arbitrary and we may exploit that to solve \pall easily. For instance, all the weights on the arcs $(c_1,c_2)$ are the same if all the BBUs are in the same data center and all datagrams require the same time to be processed in the BBUs.
    Finding an assignment in that case is trivial: send all datagrams so that they follow each other without gaps in $c_1$. In the corresponding canonical routed network, one can set $o_i = i\tau$.  Since all arcs $(c_1,c_2)$ are of weight zero in this case, the intervals of time used in $c_2$ are the same as for $c_1$ and there is no collision in $c_2$.

	Another possible assumption would be that all deadlines are larger than the longest route, which happens when all RRHs are at almost the same distance to the shared link.

	 \begin{proposition}\label{prop:asym}
	Let $N = ({\cal R}, \omega)$ be a canonical star-routed network with $n$ routes, let $P \geq n\tau$ and let $d$ be a deadline function. Let $r_{n-1}$ be the longest route, and assume that for all $r\in {\cal R}$, $d(r) \geq \lambda(r_{n-1})$. Then, there is a $(P,\tau)$ assignment for $N$ and $d$ and it can be built in time $O(n)$.
	 \end{proposition}
      \begin{proof}
       The idea is to set the waiting times of all routes so their datagrams behave exactly as the datagram of $r_{n-1}$. The offset of the route $r_i$ is set to $i\tau$, which ensures that there is no collision in $c_1$ as soon as $P \geq n\tau$. The waiting time of the route $r_i$ is $w_i = \lambda(r_{n-1}) - \lambda(r_{i})$.
        
    The time at which the datagrams of $r_i$ arrives in $c_2$ is $t(r_i, c_2) = w_i + i\tau + \lambda(r_{i})$. Substituting $w_i$ by its value, we obtain $t(r_i, c_2) =  i\tau + \lambda(r_{n-1})$.
    Hence, there is no collision in $c_2$. We denote by $A$ the defined assignment. By definition of the transmission time, we have $TR(r_i,A) = w_i + \lambda(r_i) = \lambda(r_{n-1})$. By hypothesis, $d(r_i) \geq \lambda(r_{n-1})$, which proves that the assignment respect the deadlines.

	Finally, the complexity is in $O(n)$ since we have to find the maximum length of the $n$ routes, and the computation of each $w_i$ is done by a constant number of arithmetic operations.
     \end{proof}

     \subsection{Two Stage Approach}
     
      We may decompose an algorithm solving \pall on a star-routed network into two parts: first, set all the offsets of routes so that there is no collision in $c_1$ and then, knowing this information, find waiting times so that there is no collision in $c_2$ while respecting the deadlines.

First, we give several heuristics to choose the offsets, which are experimentally evaluated in Section~\ref{sec:resultsPALL}.
For the first two heuristics, we need to define the margin. The \textbf{margin} of a route $r$ in a routed network $N$, with a deadline function $d$, is $ d(r) - \lambda(r)$. The margin is a bound on the waiting time of a route in a valid assignment.

For all presented algorithms, we assume that the star-routed network is given in its canonical form.
We send the datagrams through $c_1$ in a compact way (no gap between datagrams). It means that for $n$ routes, denoted by $r_0, \dots, r_{n-1}$, the offsets are $o_i = \sigma(i) \times \tau$, for some permutation $\sigma \in \Sigma_n$. We consider the following orders $\sigma$:

\begin{itemize}
\item Decreasing Margin (DM): Decreasing order on the margin of the routes.
\item Increasing Margin (IM): Increasing order on the margin of the routes.
\item Decreasing Arc Weight (DA): Decreasing order on the weight of the arcs $(c_1,c_2)$.
\item Increasing Arc Weight (IA): Increasing order on the weight of the arcs $(c_1,c_2)$. This sending order yields a $(P,\tau)$ assignment in which the waiting times are zero if the period is large enough (see Proposition \ref{prop:SL}).
\end{itemize}

Alternatively, we propose to fix the offsets of the routes according to some random order.
If we pack the datagrams as previously, we call Random Order (RO), the heuristic of choosing an order
uniformly at random. We may also allow some time between two consecutive datagrams in $c_1$. The order of the routes in $c_1$ is still random, and we consider two variations. Either the time between two datagrams in $c_1$ is random, and we call this heuristic Random Order and Random Spacing (RORS) or the time between two consecutive datagrams is always the same, and we call this heuristic Random Order and Balanced Spacing (ROBS).

We call \textbf{W}aiting \textbf{T}ime \textbf{A}ssignment or \wta the problem \pall where the offsets of the routes are also given as input. A solution to \wta is a valid assignment such that the offsets coincide with those given in the instance.

 \bigskip

      \noindent {\bf Waiting Time Assignment }(\wta)

      \noindent {\bf Input:}  A routed network $N=({\cal R},{\cal B},\omega)$, integers $P$ and $\tau$, a deadline function $d$ and an offset $o_r$ for each route $r \in {\cal R}$.
      
      \noindent {\bf Question:} Does there exist a valid $(P,\tau)$ assignment $A$ of $N$ such that for all $r \in {\cal R}$, $TR(r,A) \leq d(r)$ and $A(r) = (o_r,w_r)$?
      \bigskip

   In the rest of the section, we study different methods to solve \wta either by polynomial-time heuristics or by an FPT algorithm. The methods to solve \wta are then combined with the heuristics we have described to fix the offsets of the routes, which yields an algorithm for solving \pall.  
   
   \subsection{Greedy Scheduling of Waiting Times}

   We consider an instance of the problem \wta, consisting of a canonical routed network, a deadline function and an offset for each route. The \textbf{release time} of a route is defined as the first time its datagram can go through $c_2$: for a route $r$ with offset $o_r$, it is $\lambda(r,c_2) + o_r$, it is the same as the arrival time in $c_1$, $t(r,c_1)$, and it is fixed in an instance of \wta.

    The first algorithm we propose to solve \wta is a greedy algorithm that sets the waiting times by prioritizing the routes with the earliest deadline to best satisfy the constraints on the transmission time. Since the network is in canonical form, $\omega(r,t_r) = 0$ for all routes $r$, thus choosing the earliest deadline is equivalent to choosing the route with the smallest margin.
    
    We call the algorithm \greedydeadline, and it works as follows. Set $t=0$ and $U = \cal{R}$. While there is a route in $U$, find $s \geq t$ the smallest time for which there is $r \in U$ with a release time less than or equal to $s$. If there are several routes in $U$ with a release time less than or equal to $s$, then $r$ with the smallest deadline is selected, and we set $w_r = s - \lambda(r,c_2)$, $t = s + \tau$ and $ U = U \setminus \{r\}$.

    This algorithm does not take into account the periodicity, which may create an assignment that is not valid. Let $r_0$ be the first route selected by the algorithm; then $t_0 = t(r_0,c_2)$ is the first time at which a datagram go through $c_2$.
	Then, if all routes $r$ are such that $t(r, c_2) \leq t_0 + P - \tau$, 
	then by construction, there is no collision on the central arc.
      However, if a route $r$ has $t(r, c_2)$ larger than $t_0 + P - \tau$, since we consider everything modulo $P$ to determine collision, it may collide with another route. Therefore, we correct \greedydeadline by this simple modification: $s \geq t$ is the smallest time for which there is $r \in U$ with a release time less than or equal to $s$ \emph{such that there is no collision if a datagram goes through $c_2$ at time $s$}. This rule guarantees that if \greedydeadline succeeds to set all waiting times, it finds a solution to \wta, as illustrated in Figure~\ref{fig:greedydeadline}. However, it can fail to find the value $s$ at some point because the constraint on collisions cannot be satisfied. In that case, \greedydeadline stops without finding a solution.
    
    \begin{figure}
          \begin{center}
   \begin{tabularx}{0.9\textwidth}{|c|X|X|X|X|X|X|}
    \hline
     Route& $0$ & $1$ & $2$& $3$ & $4$\\
    \hline
    Deadline & $10$ &$15$&$5$&$7$&$32$\\
    \hline
     Release time & $0$ &$2$&$3$&$16$&$17$\\
    \hline
    Waiting time & $0$ &$5$&$1$&$0$&$15$\\
    \hline
      \end{tabularx}
      
      \vspace{1cm}
      \includegraphics[width=0.9\textwidth]{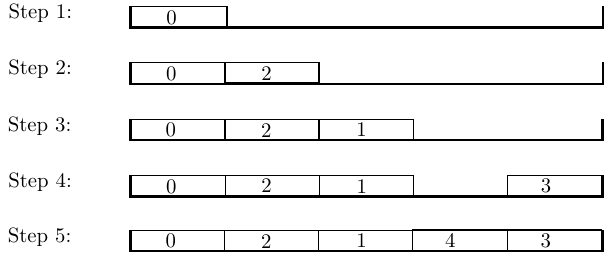}
      \caption{A run of \greedydeadline with $P = 20, \tau = 4$.}
           \label{fig:greedydeadline}
      \end{center}
      
    \end{figure}

    The complexity of \greedydeadline is in $O(n\log(n))$, using the proper data structures. The set of routes $\cal{R}$ must be maintained in a binary heap to be able to find the one with the smallest deadline in time $O(\log(n))$. To deal with the possible collisions, one maintains a list of the intervals
    of time during which a datagram can go through $c_2$. When the waiting time of a route is fixed, an interval is split into at most two intervals in constant time. During the whole algorithm, each element of this list is used at most twice, either when doing an insertion or when looking for the next free interval. Hence, the time needed to maintain the list is in $O(n)$. 
  
     \subsection{Earliest Deadline Scheduling}\label{sec:wtaheuristic}

     If we forget periodicity, the problem \wta is similar to the classical \emph{single processor scheduling} problem: Given a set of tasks with \emph{release times} and \emph{deadlines}, schedule all tasks on a single processor; that is, choose the time at which they are executed, so that no two tasks are scheduled at the same time. A task is always scheduled after its release time and it must be dealt with before its deadline. This scheduling problem can be solved in polynomial time when \emph{jobs are unit-time}~\cite{simons1978fast}, that is, when they are run for the same time.
     The problem is $\NP$-complete~\cite{lenstra1977complexity} when the running times of the jobs are different. Several algorithms solve this problem
     for all tasks with the same running time~\cite{simons1978fast,carlier1979probleme,garey1981scheduling} and the fastest one is in time $O(n\log(n))$, where $n$ is the number of jobs.
     Note that all algorithms also minimize the makespan, that is the time at which the last job is scheduled, a property useful for our work. 

     The problem \wta is the same as the single processor scheduling problem with unit-task but adding constraints arising from
     the periodicity: The tasks are the routes, the size of a datagram is the running time of a task, 
     the release time and the deadline are the same in both models, when we assume the star routed network to be canonical.
	 Let us call \textbf{M}inimal \textbf{L}atency \textbf{S}cheduling, denoted by \MLS, the algorithm which transforms an instance of \wta into an instance of the described scheduling problem to solve it in time $O(n\log(n))$ using the algorithm of~\cite{garey1981scheduling}.
     
     Recall that $t(r,c_2)$ is the time at which the datagram of $r$ goes through $c_2$. Let us denote by $t_{min}$ and $t_{max}$ the smallest and largest value of $t(r_i,c_2)$ for all $i \in[n]$. When \MLS finds an assignment $A$, it always satisfies $PT(r) \leq d(r)$ for all $r$. Moreover, by construction \MLS schedules the datagrams without collision if we forget about the periodicity (each route sends only one datagram). Let us assume that $t_{max}- t_{min} \leq P -\tau $; then all datagrams go through $c_2$ during an interval of time less than $P$. Hence, when we compute potential collisions modulo $P$, all the relative positions of the datagrams stay the same, which implies there is no collision. However, if $t_{max}- t_{min} > P -\tau $, then computing $t(r_i,c_2)$ modulo $P$ for all $i$ may reveal some collisions. Since the scheduling algorithm minimizes $t_{max}$, it tends to find small values for $t_{max} - t_{min}$ and \PMLS may succeed in finding a valid assignment (as shown in Section~\ref{sec:resultsPALL}), but not for all instances. 
     
     We now present a variant of the previous algorithm that we call
     \textbf{P}eriodic \textbf{M}inimal \textbf{L}atency \textbf{S}cheduling, denoted by \PMLS. The aim is to deal with the periodicity, by modifying the instance without changing the assignments, so that there is a better chance of finding a solution with $t_{max}- t_{min} \leq P -\tau $.  If an instance has a valid assignment, we can guarantee that one route has a waiting time of zero in some valid assignment. 
     
      Recall that $t(r,c_1)$ is the release time of $r$. The \PMLS algorithm runs, for each route $r \in \cal{R}$, the algorithm \MLS on an instance defined as follows. Subtract $t(r,c_1)$ to all the release times and deadlines of the routes to obtain an equivalent problem. Therefore, $t(r,c_1)$ is zero in the instance we build, and the waiting time $w_r$ is set to zero. Hence, the datagram of $r$ goes through $c_2$ at time $0$ and $t_{min} = 0$.
     Then, as in Proposition~\ref{prop:canonical}, the instance is modified so that all release times are in $[P-\tau]$. Each release time $t(r_i,c_1)$ is replaced by $t(r_i,c_1) \mod P$ and $d(r_i) = d(r_i) - (t(r_i,c_1) - (t(r_i,c_1) \mod P))$. Furthermore, if the release time of a route $r$ is between $P-\tau$ and $P$, we set it to $0$ and $d(r) = d(r) - P$.  The deadline of each route is set to the minimum of its deadline and $P - \tau$. Hence, if \MLS finds a solution for such a modified instance, by construction of the instance, we have $t_{max} \leq P -\tau $. Since $t_{min} = 0$, the assignment is valid. \PMLS returns the first assignment it finds when running \MLS for some $r \in \cal{R}$.

     The instance of \wta we have defined in this transformation is equivalent 
     to the original instance, except we have fixed the waiting time of 
     $r$ to be zero. If there is some valid assignment, then at least one route has a waiting time of zero. Hence, when \MLS finds an assignment, \PMLS also finds one. \MLS is used at most $n$ times, thus the complexity of \PMLS is in $O(n^2\log(n))$. Note that \PMLS is a heuristic and may fail to find a solution even if it exists. It is the case when, for the $n$ modified instances, there is no solution with times $t(r_i,c_2)$ using an interval of time less than $P$ in $c_2$.

\subsection{FPT algorithms for \texttt{WTA} and \texttt{PALL}}

As a warm-up, we give a simple FPT algorithm for \wta, which is practical,
and then we build on it to give a more complicated FPT algorithm for \pall. Unfortunately, the dependency on $n$, the number of routes in the second algorithm, is too large to be useful in practice. 

\begin{theorem}\label{th:braFPT}
$\wta \in \FPT$ over star-routed networks when parametrized by the number of routes.
\end{theorem}
\begin{proof}
 Consider an instance of \wta, which is given as a release time and a deadline for each route.
 We show that we can build a set of instances such that one of these instances has a valid assignment if and only if the original instance has a valid assignment.

  As for \PMLS, for each route $r$, we consider the instance where $r$ has release time and waiting time zero ($t(r,c_1) = w_r = 0$). The release times and deadlines of all routes are modified so that all release times are less than $P$ as in the transformation described for \PMLS. If there is an assignment such that $t_{max} < P-\tau$, then the periodicity does not come into play for this assignment and the algorithm \MLS will find the assignment as explained in Section~\ref{sec:wtaheuristic}.

 If there is a valid assignment for an instance with the previously stated properties,
 then there is a valid assignment satisfying for all $i$, $t(r_i,c_2) \leq 2P - \tau$.  
 Indeed, if there is a $i$ such that $t(r_i,c_2) \geq 2P$ in an assignment, then we have 
 $w_i = t(r_i,c_2) - \lambda(r_i,c_2) \geq P$. Hence, we can set $w_i = w_i -P \geq 0$ and we still have 
 a valid assignment. Moreover, for all $r_i \neq r$, it is not possible that $2P-\tau < \lambda(r_i,c_2) \leq 2P$, since it implies a collision between $r$ and $r_i$.

From an instance $I$, with the properties of the first paragraph, we define a new instance $I'$ whose valid assignments are a subset of the ones of $I$. Moreover, one of the valid assignments of $I'$ satisfies that, for all $i \in [n]$, $t(r_i,c_2) \leq P - \tau$ and is thus found by \MLS. 
Let us now consider $A$ a valid assignment of $I$; we can assume that, for all $i \in [n]$, $t(r_i,c_2) \leq 2P - \tau$. Let $S$ be the set of routes $r_i$ such that  $P - \tau < t(r_i,c_2) \leq 2P - \tau$. The instance $I'$ is defined by changing, for all routes $r \in S$, $t(r,c_1)$ and $d(r)$ to $t(r,c_1) - P$ and $d(r) - P$. Then, by construction, $A$ is also a valid assignment of $I'$. Assignment $A$ as a solution of $I'$, satisfies $t(r_i,c_2) \leq P - \tau$ for all $i\in [n]$. 

The FTP algorithm is the following: for each route, $r$ build a modified instance as in $\PMLS$.
Then, for each subset $S$ of routes, remove $P$ to the release time and the deadline of each route in $S$ and run \MLS on the instance so modified. If there is a valid assignment, then we have proved that there is some $S$, such that the instance built from $S$ has a valid assignment with $t(r_i,c_2) \leq P - \tau$ for all $i\in [n]$. Hence, \MLS finds a valid assignment for this instance.
\end{proof}

The algorithm of Theorem~\ref{th:braFPT} has a complexity of $O(2^nn^2\log(n))$. If we consider some valid assignment, the routes $r$ with $t(r,c_2) > P$, must satisfy $t(r,c_2) > P + \tau$ to avoid collision with the first route. Hence, the deadline of these routes must be larger than $P + \tau$. These routes are exactly those that must be put in $S$. Hence, we can enumerate only the subsets of routes with a deadline larger than $P + \tau$. In practice, only $k$ routes have a deadline larger than $P + \tau$ with $k << n$, and we need only to consider $2^k$ subsets. Let us call this algorithm \textbf{A}ll \textbf{S}ubsets \PMLS, and let us denote it by \ASPMLS.

\begin{theorem}\label{th:pallFPT}
$\pall \in \FPT$ over star-routed networks when parameterized by the number of routes.
\end{theorem}
\begin{proof}
 Consider an instance of \pall with a valid assignment. We characterize such a valid assignment by a set of necessary and sufficient linear systems it must satisfy.  These conditions are expressed in terms of the values $t(r,c_1)$ and $t(r,c_2)$, and choosing these values is equivalent to choosing the offsets and the waiting times, that is, choosing an assignment.

First, we assume the star-routed network is canonical. Hence, there is a valid assignment $A$, such that for all routes $r \in \cal{R}$, $0 \leq t(r,c_1) < P -\tau$ and $0 \leq t(r,c_2) < 2P-\tau$. 
By definition $t(r,c_2) = t(r,c_1) + \omega(r,c_2) + w_r$. Since a waiting time is non-negative, we have $t(r,c_2) \leq t(r,c_1) + \omega(r,c_2)$. 
Now, let $S$ be the set, defined as in Theorem~\ref{th:braFPT}, of the routes $r$ such that  $P - \tau < t(r,c_2) \leq 2P - \tau$. We want to guarantee that for $r \in \cal{R}$, $t(r,c_2) \in [P-\tau]$.
To do that, we replace the inequality $t(r,c_2) \leq t(r,c_1) + \omega(r,c_2)$ by $t(r,c_2) \leq t(r,c_1) + \omega(r,c_2) - P$ and $d(r)$ by $d(r) - P$ for all $r \in S$. The presented linear constraints now depend on $S$, which itself depends on $A$.

 Let $\sigma$ and $\sigma'$ be two permutations of $\Sigma_n$ such that $\sigma$ is the order 
 of the routes $r_0,\dots, r_{n-1}$ according to the value $t(r,c_1)$ and $\sigma'$ according to the value $t(r,c_2)$.  Since all $t(r,c_1)$ and $t(r,c_2)$ are in $[P-\tau]$, we have $t(r,c_1) = t(r,c_1) \mod P $ and $t(r,c_2) = t(r,c_2) \mod P $. Hence, we can express the constraints on the absence of collision between routes by adding the following inequalities to the ones of the previous paragraph:
 
 \begin{itemize}
 	\item for all $i < n-1$, $t(r_{\sigma_{i}},c_1) \leq r_{\sigma_{i+1}},c_1 + \tau)$ (no collision in $c_1$)
 	\item for all $i < n-1$, $t(r_{\sigma'_{i}},c_2) \leq r_{\sigma'_{i+1}},c_2 + \tau)$ (no collision in $c_2$)
 	\item for all $i < n$,  $t(r_{i},c_2) < d(r_i)$ (deadline respected)
 \end{itemize}

Consider now the system of inequalities $E_{S,\sigma,\sigma'}$ we have built from $A$.
The values $t(r,c_1)$ and $t(r,c_2)$ given by $A$ satisfy the system by construction. 
Moreover, any solution to these inequalities yields a valid assignment because the inequalities guarantee 
that there is no collision, that the offsets and the waiting times are non-negative, and that all routes meet their deadlines. However, a solution of $E_{S,\sigma,\sigma'}$ may be rational, while offsets and waiting times must be integers. We use the following simple fact: $x + e_1 \leq y + e_2$ implies $\lceil x \rceil + e_1 < \lceil y \rceil + e_2$ when $e_1$ and $e_2$ are integers. Since all inequalities of $E_{S,\sigma,\sigma'}$ have this form, if we take the upper floor of the components of a solution, it is still a solution of $E_{S,\sigma,\sigma'}$ with \emph{integer} values. As a consequence, any solution to $E_{S,\sigma,\sigma'}$ yields a valid assignment of the original instance of \pall.

The algorithm to solve $\pall$ is the following. Build $E_{S,\sigma,\sigma'}$ for all triples $(S,\sigma,\sigma')$. Then, solve each linear system, and if it admits a solution, convert it back into a
valid assignment of the instance of \pall by rounding. There are $2^n$ sets $S$ and $n!$ orders $\sigma$. Thus, $2^n(n!)^2$ systems with $2n$ variables and a bit size of the same order as the original instance are solved at most. Since solving each system can be done in time polynomial in the size of the instance, it proves that the algorithm is $\FPT$ in $n$. Moreover, it always finds a valid assignment if there is one, since we have shown that from a valid assignment, we can find $(S,\sigma,\sigma')$ for which the values associated with $A$ satisfy $E_{S,\sigma,\sigma'}$.
\end{proof}

    \subsection{Experimental Evaluation}
    \label{sec:resultsPALL}

    In this section, we set the number of routes to $8$ to make comparisons with the results of Section~\ref{sec:exp_PAZL} easier (the datagram size and the bandwidth of the links also remain the same). To not overload the section, we choose to draw the weights of the arcs of the fronthaul network uniformly in $[P]$ because it is harder than considering small routes and this highlights the performance differences of the algorithms we propose. At the end of this section, we explore even harder distributions of weights, which correspond to the realistic cases of RRHs connected to one or two data centers. 

    We use \emph{the same deadline} for all routes, which is the most common constraint when modeling a C-RAN problem: all RRHs have the same latency constraint. We define the {\bf margin} of an instance as the margin of the longest route of the routed network. Since all routes have the same deadline, it is the difference between the length of the longest route and the deadline. Note that the margin is defined before making the network canonical since this operation makes the deadlines all different and thus breaks the semantics of the margin.

	The margin represents the \emph{logical latency} that can be used by the communication process we try to find, abstracting away the physical length of the network. For a given star-routed network, it is equivalent to set a margin or to set all the deadlines to the length of the longest route plus the margin. However, to compare different star routed networks with different lengths of routes, the margin is more relevant than the deadline. Hence, in all the following experiments, we give the success rates of different algorithms, for margins from $0$ to $3,000$ tics, to characterize how much logical latency is needed in the assignments we find. We look at two different regimes, a medium load of $0.8$ and a high load of $0.95$. Considering smaller loads is not relevant since we can solve the problem using bufferless assignments, as shown in Section~\ref{sec:exp_PAZL}.

    \paragraph{Finding the best first stage heuristic}

   	We first try to understand what is the best choice of heuristics for the first stage of the algorithm. The first stage is followed in this experiment by \greedydeadline, the simplest algorithm to solve \wta. In Figure~\ref{fig:success1random}, the success rate of all possible first-stage heuristics to solve \pall is given as a function of the margin of the instances. The success rate is an average computed over $10,000$ random star-routed networks.

\begin{figure}[h] 
\begin{center} 
  \begin{minipage}[b]{0.45\linewidth}

  \includegraphics[height=5cm]{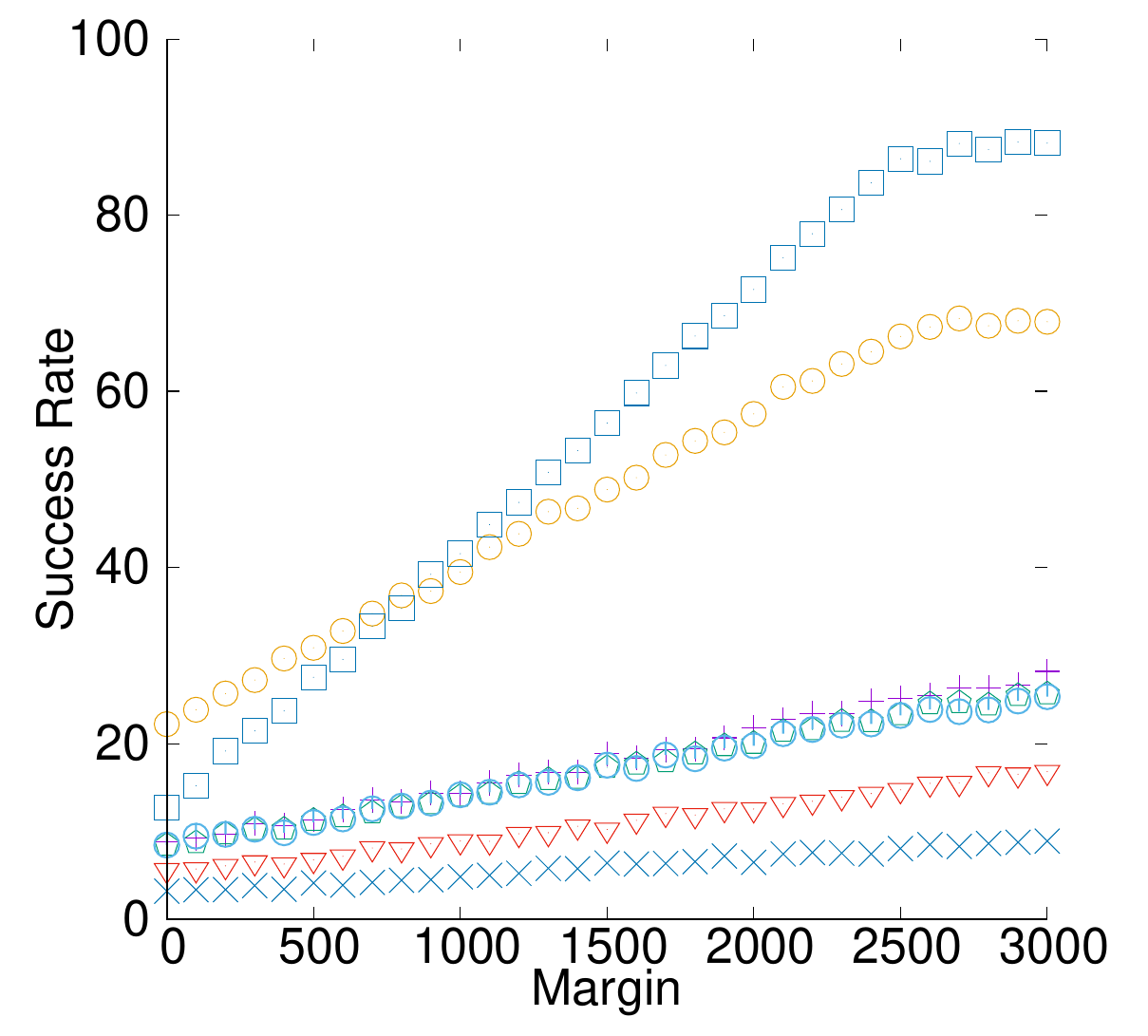}
  \end{minipage}
  \begin{minipage}[b]{0.54\linewidth}
  \includegraphics[height=5cm]{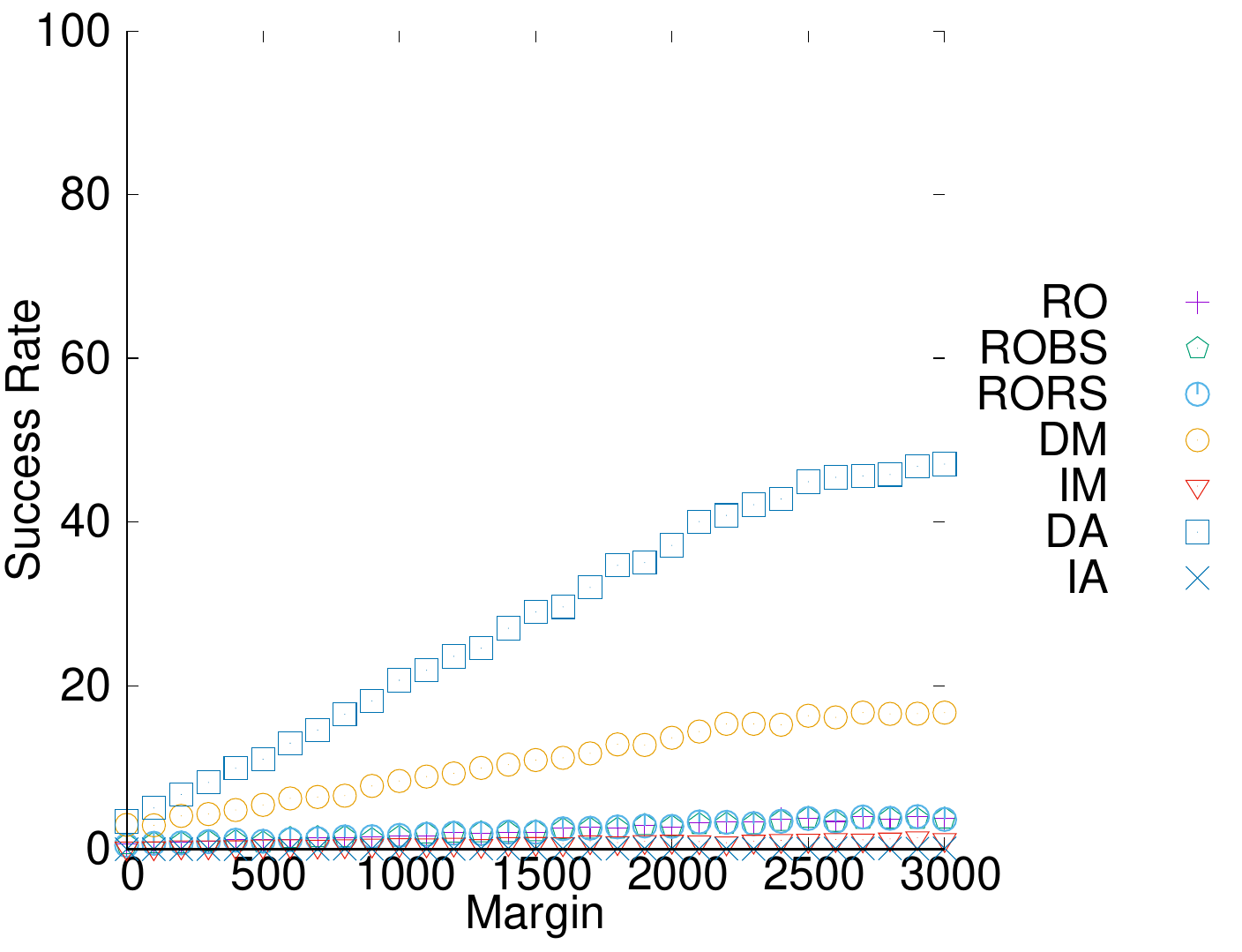}

\end{minipage}
      \caption{Success rate of different sending orders; load of $0.8$ at left and load of $0.95$ at right.}
           \label{fig:success1random}
           \end{center}

     \end{figure}

  According to our experiments, policy IA, which consists of sending the datagrams by increasing order on the length of the arcs $(c_1,c_2)$, does not work well. It corresponds to the policy of Proposition~\ref{prop:SL}, which we already know to be bad for \pazl when the routes are long, as in this experiment. Sending in decreasing order by the margin of the routes (DM) or by the length of the arcs $(c_1,c_2)$ (DA) works better, and it seems that DA is better than DM, especially in a loaded network.

Sending the datagrams using a random order does not perform well,
but is still better than IM and IA, which shows that the latter are poor choices for the first stage of our algorithm. The interest in using a random order is that we can draw many of them. In Figure~\ref{fig:success1000random}, the same experiment is made for the three random heuristics, but we now draw $1,000$ different random orders and solve each induced instance of \wta using \greedydeadline. The algorithm is considered to succeed as soon as a valid assignment is found for one order. Each random order drawn is used for RO, RORS, and ROBS to make the comparison fairer.

\begin{figure}[h] 
  \centering
  \begin{minipage}[b]{0.45\linewidth}

  \includegraphics[height=5cm]{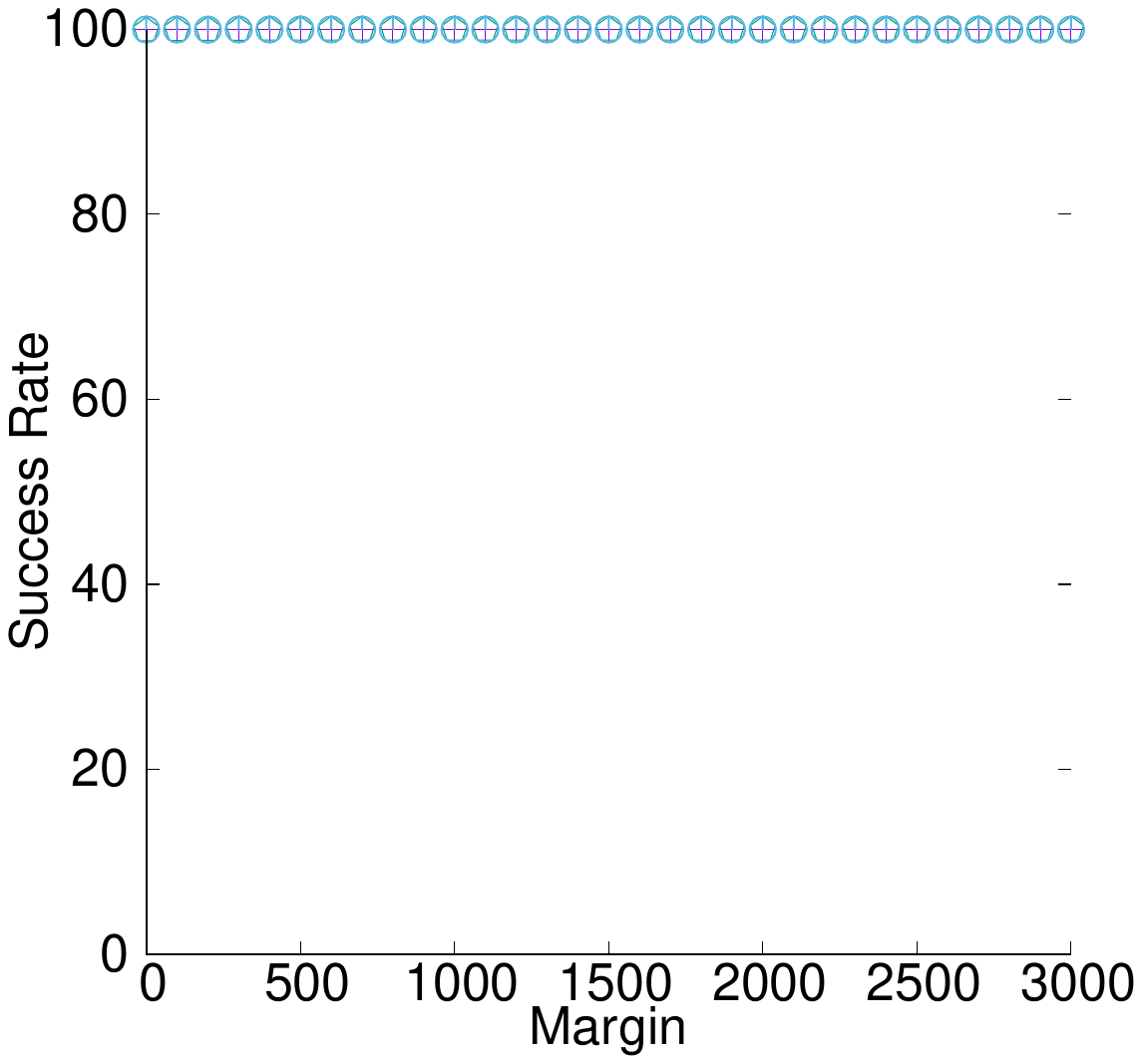}
  \end{minipage}
  \begin{minipage}[b]{0.54\linewidth}
  \includegraphics[height=5cm]{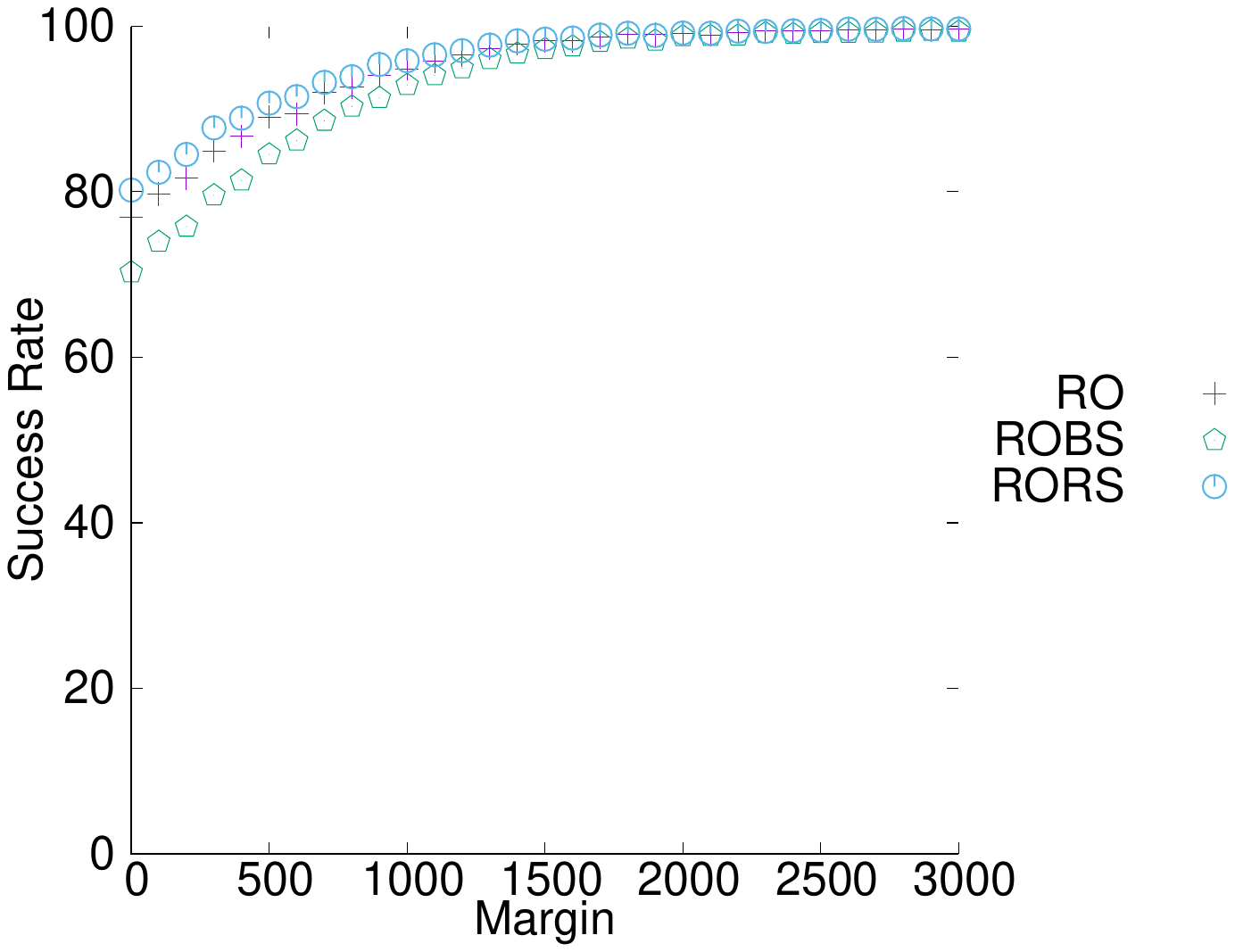}

\end{minipage}
       \caption{Success rate of different sending orders with the random orders generated $1000$ times; load of $0.8$ at left and load of $0.95$ at right.}
      \label{fig:success1000random}
          \end{figure}

     Our algorithms find assignments with a margin $0$ for many instances with a load of $0.95$ and long routes. This is not possible when only looking for bufferless assignments (see Section~\ref{sec:exp_PAZL}). It justifies the interest of studying \pall and not only \pazl.
  
     Using many random orders is much better than DA, the best policy using one specific order. 
     With a load of $0.95$, a solution is found with a margin of $0$ most of the time. The three random order policies have similar performances, but RORS has a slightly better success rate than the other two, under high load and small margin. Hence, in the following experiments, we always draw $1,000$ random orders using the policy RORS to set the offsets of the assignments.
    
 \paragraph{Comparison of the algorithms solving \texttt{WTA}}

We now compare the performances of the four different algorithms used in the second stage to set the waiting times. Since \greedydeadline already finds assignments with a margin $0$ under a mild load of $0.8$, it is more interesting to focus on the behavior of the algorithms under a high load of $0.95$. In Figure~\ref{fig:success21000}, we represent the success rate of the four algorithms with regard to the margin, computed over $10,000$ random star routed networks generated with the same parameters as previously.

\begin{figure} [h]
\begin{center}
\includegraphics[width=0.8\textwidth]{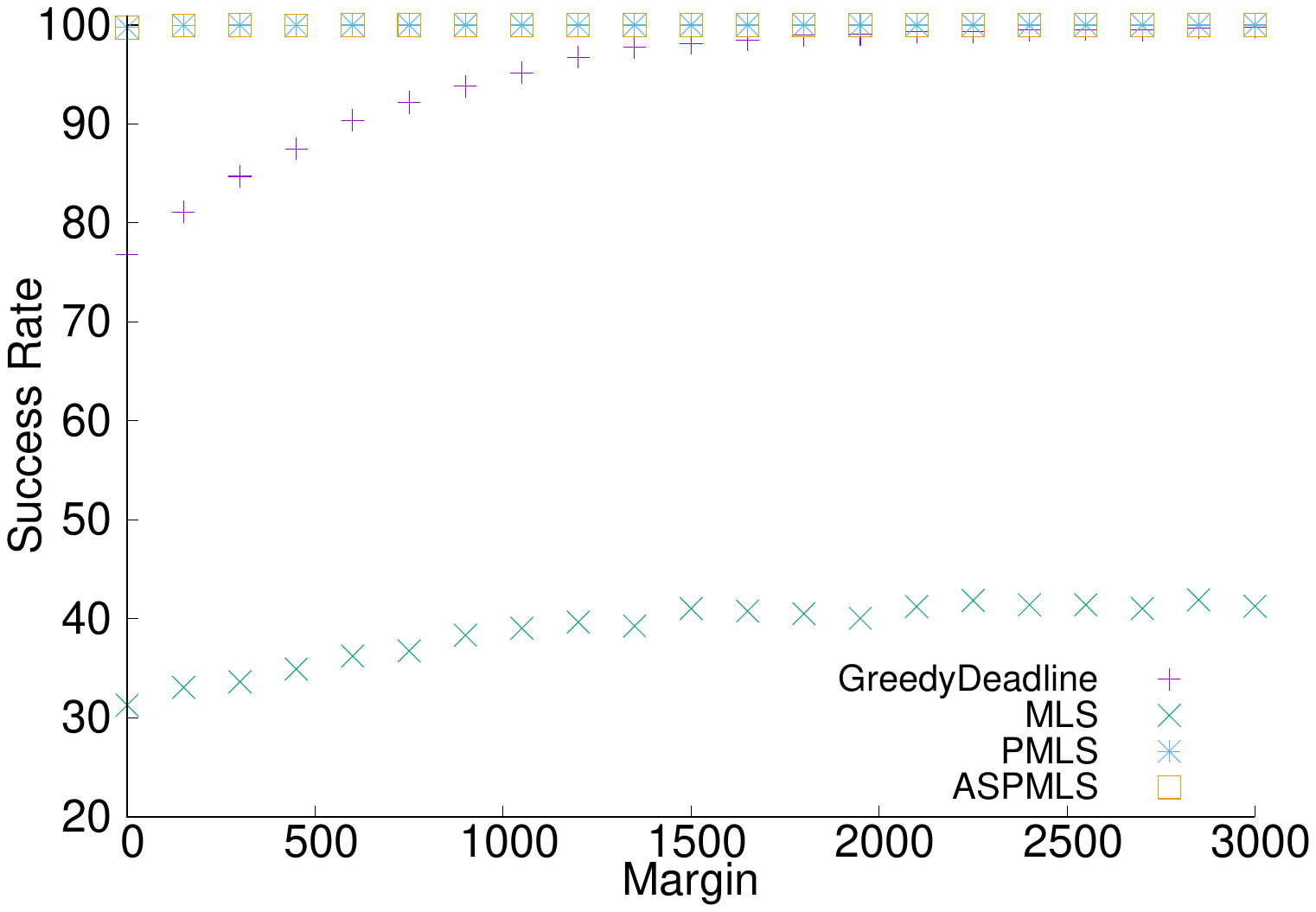}
\end{center}
\caption{Success rate of four algorithms solving \pall, load of $0.95$.}
\label{fig:success21000}
\end{figure}

As explained in Section~\ref{sec:wtaheuristic}, the \MLS algorithm does not consider periodicity while computing the waiting times, unlike the other algorithms for \wta. This explains why \MLS performs poorly, worse than \greedydeadline, \PMLS, and \ASPMLS, and shows that \emph{taking into account the periodicity} is fundamental.
\greedydeadline has close to a $100\%$ success rate for margins larger than $1,500$ while \PMLS and \ASPMLS algorithms find a solution for more than $99\%$ of the random instances, even \emph{with a margin $0$}. In other words, for very high load and no margin, there are very few instances for which we do not find an assignment. With a margin of $300$, which corresponds to about $15\mu$s of additional delay with the chosen parameters, we always find a solution.

It turns out that the performances of \PMLS and \ASPMLS are almost identical. Even with a load of $1$ and a margin of $0$, we have to draw $100,000$ random instances before finding one that can be solved by \ASPMLS and not by \PMLS. Since \ASPMLS is of exponential complexity in $n$, it may not be relevant to use it within the parameters of this experiment. To verify that, we present the computing time of \PMLS and \ASPMLS for different instance sizes. To stress the algorithms, we set the margin to $0$ and the load to $0.95$. The table of Figure~\ref{fig:tps_fpt} shows the computation times of \PMLS and \ASPMLS, averaged over $1,000$ instances. Recall that both \PMLS and \ASPMLS use the same first stage, which produces $1,000$ instances of \wta, using the policy RORS.
     
          \begin{figure}[h] 
       \begin{center}
   \begin{tabularx}{0.8\textwidth}{|c|X|X|X|X|X|X|}
    \hline
    \# routes& $8$ & $12$ & $16$& $20$ & $24$\\
    \hline
    \ASPMLS (ms) & $1.88$ &$5.98$&$47.75$&$209.2$&$1815$\\
    \hline
     \PMLS (ms) & $0.07$ &$0.08$&$0.09$&$0.10$&$0.12$\\
    \hline
    Ratio & $27$ &$78$&$523$&$2122$&$14882$\\
    \hline
      \end{tabularx}
      \end{center}
   \caption{Computation time for \PMLS and \ASPMLS, as a function of the number of routes}
        \label{fig:tps_fpt}
     \end{figure}

  The complexity of both algorithms depends on the number of routes. As shown in Figure~\ref{fig:tps_fpt}, the time complexity of \PMLS seems linear on \emph{average}, while its theoretical worst-case complexity is roughly quadratic. As expected, \ASPMLS scales exponentially with the number of routes. Both algorithms are fast enough for instances of at most $20$ routes, but for $40$ routes or more \ASPMLS becomes too slow. Since \ASPMLS rarely finds a solution when \PMLS does not and is much slower, one should prefer to use \PMLS. 

    When evaluating the computing time of our method, we should take into account how many random orders are drawn. In previous experiments, we have drawn $1,000$ random orders, which may be $1,000$ time slower than using a single fixed order. There is a trade-off between the number of random orders and the success rate. We investigate the success rate of our algorithms with regard to the number of random orders drawn, a load of $0.95$, and a margin of $0$. The table of Figure~\ref{fig:randomdrawing} presents the success rate for different numbers of sending orders, averaged over $10,000$ instances, for \greedydeadline, \PMLS, and \ASPMLS.

         \begin{figure}[h] 
       \begin{center}
   \begin{tabularx}{0.8\textwidth}{|c|X|X|X|X|X|X|}
    \hline
    \# orders& $1$ & $10$ & $100$& $1,000$& $10^{4}$&$10^{5}$\\
    \hline
    \greedydeadline & $0.55$ &$6.05$&$35.44$&$77.43$&$90.1$&$92.4$\\
    \hline
    \PMLS & $82.04$ &$98.84$&$99.71$&$99.80$&$99.83$&$99.83$\\
    \hline
    \ASPMLS & $91.33$&$99.17$&$99.72$&$99.80$ &$99.83$&$99.83$\\
    \hline
      \end{tabularx}
      \end{center}
   \caption{Success rates function of the number of random orders drawn in the first stage of the three algorithms}
        \label{fig:randomdrawing}
     \end{figure}

	First, observe that the better the algorithm to solve $\wta$ is, the fewer random orders it needs in stage one to achieve its best success rate. In particular, \ASPMLS has better results than \PMLS for less than $1,000$ random orders, but not beyond. This further justifies our choice to draw $1,000$ random orders to obtain the best success rate within the shortest time.

	The number of different orders is $7!= 5,040$ since we have $8$ routes and the solutions are invariant up to a circular permutation of the order. Hence, for $8$ routes, it is possible to test every possible order. However, the computation time of this exhaustive method scales badly with $n$. The fact that \PMLS and \ASPMLS already have high success rates for $10$ random orders hints that even for a larger number of routes, drawing $1000$ random orders is sufficient to obtain good assignments.

     \paragraph{Harder Topologies}
     
    Previous experiments use instances with weights of arcs uniformly drawn from a large interval. However, it is quite natural to consider that most routes are of roughly the same length or can be arranged in two groups of similar length when the fronthaul network involves one or two data centers.
    
		  By Proposition~\ref{prop:asym}, there is an assignment with a margin equal to the maximum difference
    between the sizes of the routes. Hence, if all routes have almost the same size, the needed margin is small. If the routes are drawn uniformly in a large interval, then the expected difference between the longest route and the second longest route is large. This difference can be seen as a free waiting time for most routes, hence we expect to need little margin in this regime too. As a consequence, the harder instances should be for routes with length drawn from an interval of moderate size compared to the period.

  	Figure \ref{fig:2grp} shows the probability of success of \PMLS  over $10,000$ instances as a function of the margin. In the top experiment, the weights of the arcs are drawn from $[0,I]$, where $I$ goes from $0$ to $6400$. As expected, the success rate decreases when the size of the interval increases until $I = 1600$ and then increases again. In the most difficult settings, only $78\%$ of the instances can be solved with a margin $0$, and we need a margin of $1,900$ to ensure that \PMLS always finds a solution. Results for \ASPMLS are not shown since they are the same as for \PMLS, even on these hard instances.

 	 We do the same experiment at the bottom of the figure, except that the weights of arcs of half of the routes are drawn from $[I]$, and the weights of the other half are drawn from $[P/2,P/2 + I[$. Remember that $P = 20.000$. The situation is the same as for the previous experiment but with better success rates, hence the case of two data centers seems simpler to deal with.
  
           \begin{figure}

       \begin{center}
      \includegraphics[width = 0.9\linewidth]{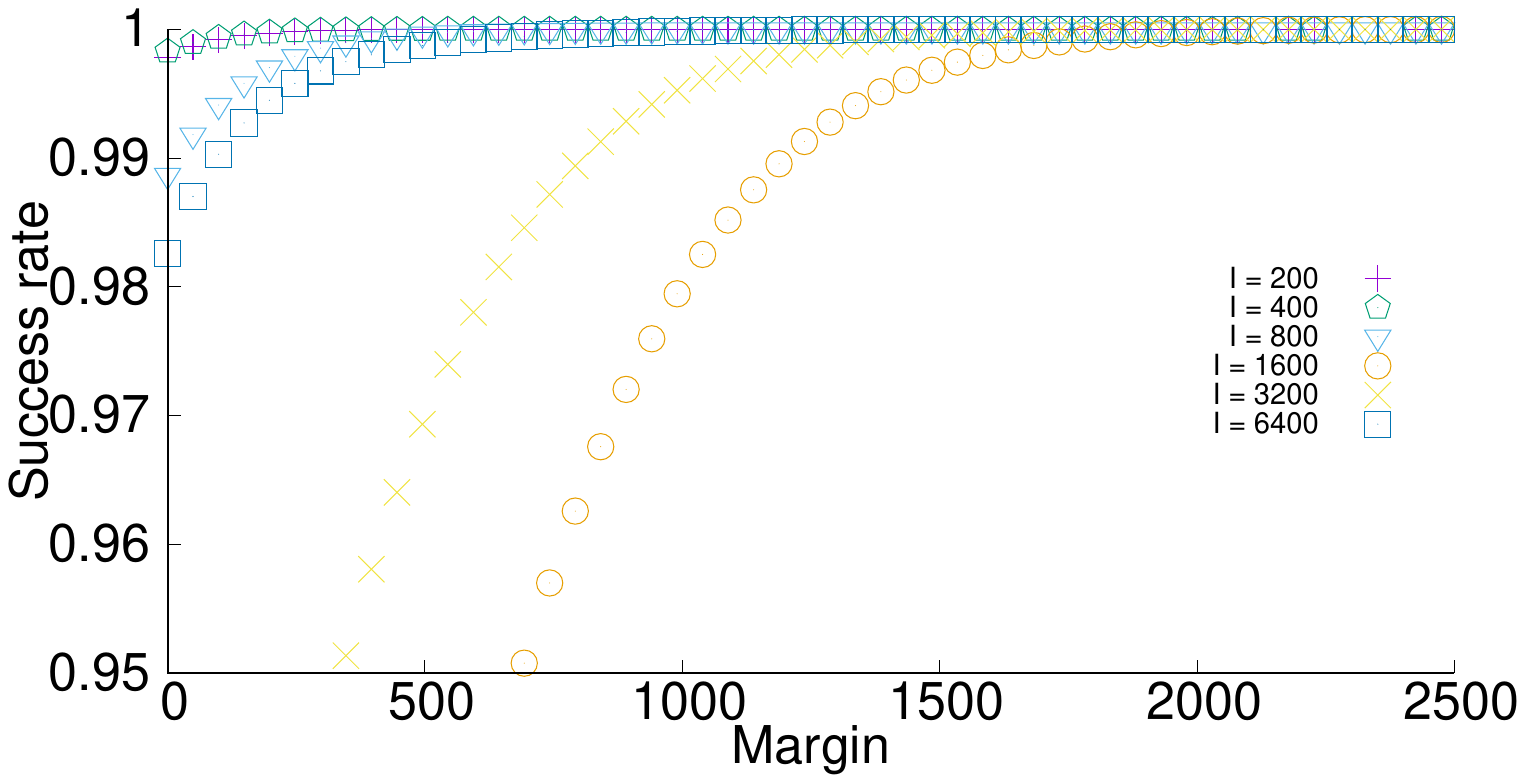}

      \includegraphics[width = 0.9\linewidth]{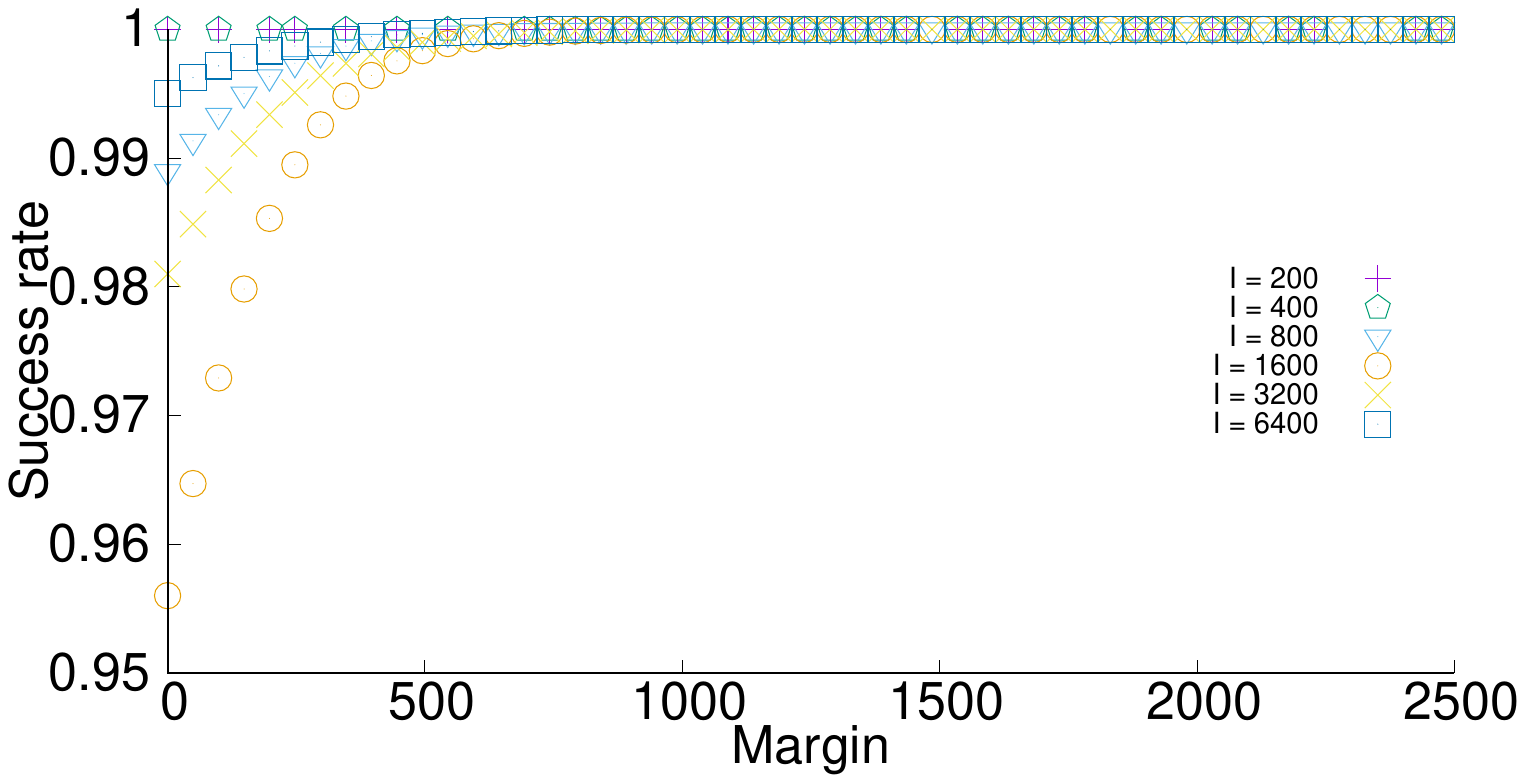}

         \end{center}
         \caption{Success rate of \PMLS, with length of arcs drawn either from $[I]$ (top) or from $[P/2,P/2 + I[$ (bottom).}
      \label{fig:2grp} 
  \end{figure}
\section{Deterministic Assignments vs Statistical Multiplexing}\label{sec:comparison}

\subsection{Performance of Statistical Multiplexing}

Now that we have designed and tuned \PMLS to solve \pall efficiently, we compare its performances against the actual way to manage the messages in a network: \emph{statistical multiplexing}, with a \FIFO buffer in each node of the network to resolve collisions. For statistical multiplexing, the time at which the datagrams are sent in the network is not managed by the user as in our approach. Thus, we assume the offset of each route is fixed to some random value and stays the same over time.
We consider a second policy to manage buffers called \critdead. In a buffer with several datagrams, this policy sends the one with the smallest remaining margin, which is the time it can wait before missing its deadline.

We have implemented a statistical multiplexing algorithmic simulator to evaluate the performance of these two policies. We compare them to our solution, finding an assignment with the smallest possible margin using \PMLS.
For statistical multiplexing, both contention points have a buffer. The process is not periodic:
even if the offset of a route is the same each period, it is possible that some datagram does not arrive at the same time in a contention point in two consecutive periods because of buffering. Therefore, we must measure the transmission time of each route over several periods if we want to compute the maximum latency of the network. We choose to simulate it for $1,000$ periods, but we have observed that the transmission time usually stabilizes in less than $10$ periods. The \textbf{margin}, for statistical multiplexing, is defined as the maximum transmission time, computed as explained, minus the size of the longest route of the star-routed network.

In Figure~\ref{fig:sto}, we represent the probability of success of statistical multiplexing and \PMLS for different margins. The success rates are computed from $10,000$ star-routed networks for each margin. On the top of Figure~\ref{fig:sto}, the network arc lengths are uniformly drawn from $[P]$, while on the bottom of Figure~\ref{fig:sto}, the network arc lengths are uniformly drawn from $[1600]$ (the hardest settings of the previous section). The other parameters of the experiences are the same as previously and the load is $0.95$.

    \begin{figure}

       \begin{center}
      \includegraphics[width = 0.9\linewidth]{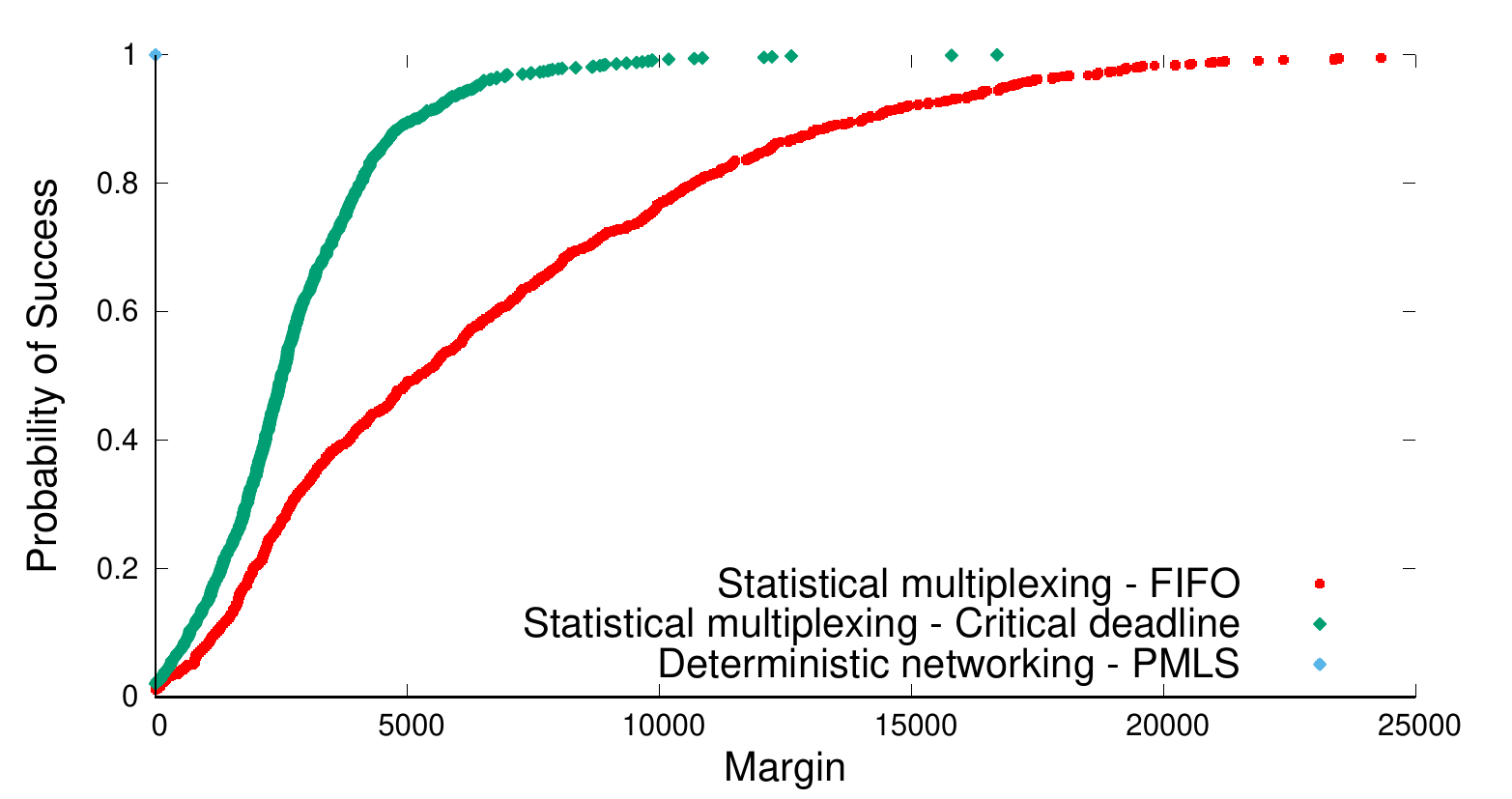}

     \includegraphics[width = 0.9\linewidth]{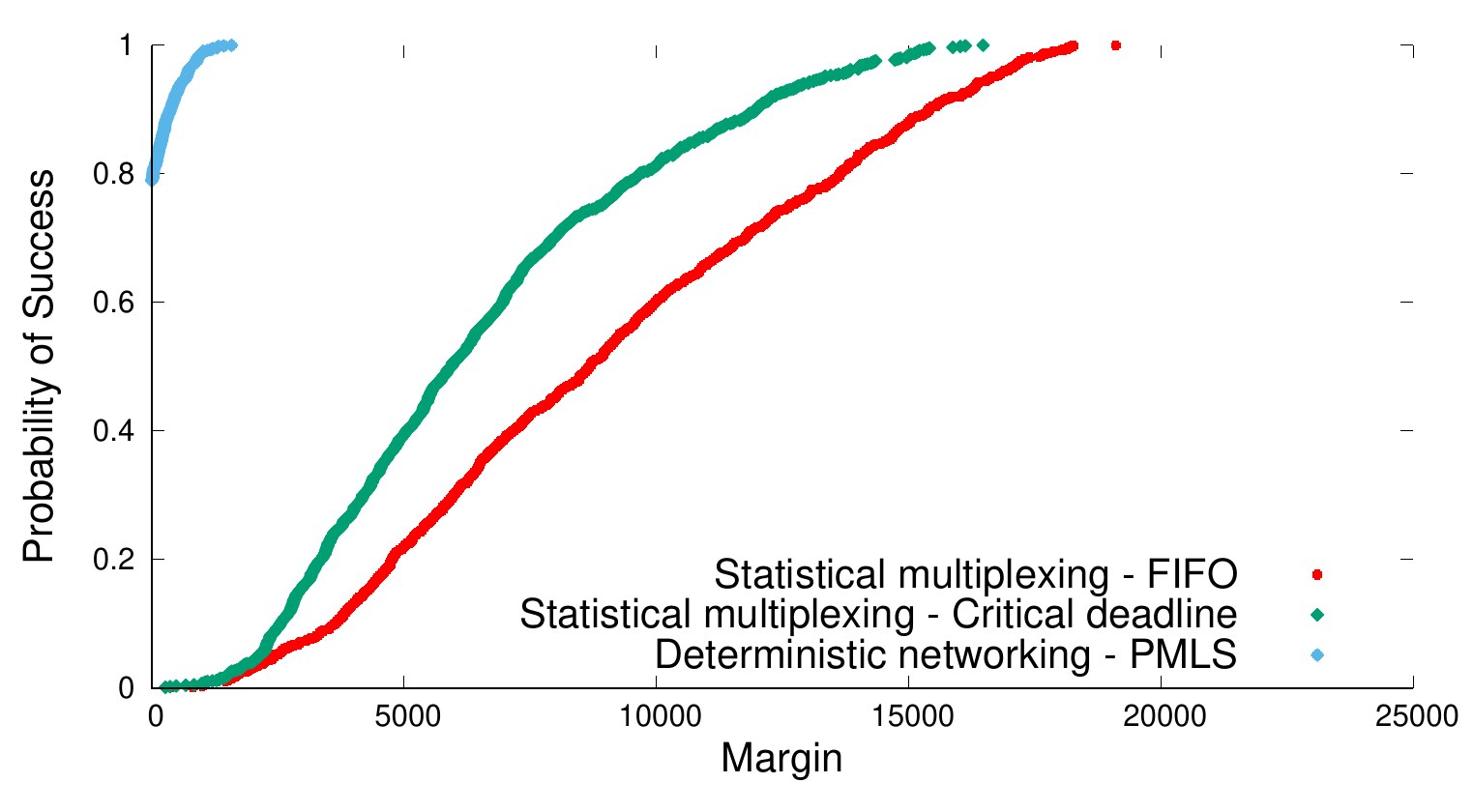}

       \end{center}

  \caption{Probability of success of statistical multiplexing and \PMLS for several margins on random topologies when network arc lengths are uniformly drawn either from $[P]$ (top) or from  $[1600]$ (bottom).}
      \label{fig:sto} 
      \end{figure}

   The experiment shows that statistical multiplexing cannot ensure a small enough latency. 
    For random topologies, the latency is extremely high when using FIFO ($6,538$ tics on average), with a margin of about $10,000$ for the worst $30\%$ of instances, which corresponds to half the period ($0.5$~ms). Even when the messages are managed with \critdead, $20\%$ of the instances have a margin of more than $4,000$ ($2,838$ tics on average) while PMLS finds an assignment with $0$ margin $99\%$ of the time! 
    
    For hard topologies (bottom figure), the average margin of statistical multiplexing ($9,052$ tics for FIFO, $6,574$ tics for \critdead) is worse than for random topologies. The worst case of \critdead remains the same ($\simeq 16,500$ tics) while the worst case of FIFO decreases from $30,828$ tics for random topologies to $19,105$ tics on hard topologies. The settings are stressful for \PMLS, and we find an assignment with a margin of $0$ in only $78\%$ of the instances, and it needs a margin of $2,000$ tics to be sure to find an assignment. However, \PMLS still vastly outperforms statistical multiplexing both for the average margin and for the worst margin. 
    
    Even under a light load of $0.4$, for which we can always find a bufferless assignment, statistical multiplexing has a very high average margin ($1,290$ tics for FIFO and $1,052$ tics for \critdead) and worst-case margin ($10,963$ tics for FIFO and $6,938$ tics for \critdead).

    For each $1,000$ tics of latency, we save from the periodic process, we can lengthen the routes by $10$km, which has a huge economic impact. We feel that it strongly justifies the use of a deterministic sending scheme for latency-critical applications such as our C-RAN motivating problem.    
     
    \subsection{Periodic Assignment and Random Traffic}
    
    The algorithms proposed in this paper are designed to manage deterministic periodic flows in dedicated networks. In this section, the objective is to determine the effect of adding in the network non-deterministic flows (internet traffic, best-effort) managed by statistical multiplexing.

    The algorithms solving \pall are not designed to take into account additional best-effort traffic. In particular, they often build very compact assignments, with all datagrams following one another in a contention vertex, which is bad for the latency of best-efforts datagrams trying to go through the same contention point. Thus, we propose an adaptation of any algorithm for solving \pall, to find assignments where the unused tics are as evenly spaced as possible in the period. Such assignments minimize the maximal latency of any random datagram trying to go through the contention points. A similar approach to decrease the latency of best-effort datagrams while scheduling C-RAN datagrams on an optical ring can be found in~\cite{DBLP:conf/ondm/BarthGS19}.

    \subsubsection{Spaced Assignments}

    Most algorithms for \pall, when determining the waiting times, send datagrams as early as possible
    and thus create long sequences of datagrams in $c_2$, without free tics between them. We propose to modify any algorithm solving \pall on an instance with datagram size $\tau$ as follows: compute a $(P,\tau')$ assignment using the algorithm for the largest possible $\tau' \geq \tau$. 

    \begin{lemma}\label{lemma:smaller_tau}
    Let $I' = (N,P,\tau',d)$ be an instance of \pall, for which there is an assignment, and let 
    $\tau \leq \tau'$. Then, there is also an assignment for $I = (N,P,\tau,d)$.
    \end{lemma}  
    \begin{proof}
    Let $A$ be the assignment of $I'$, the absence of collision is the absence of 
    intersection between intervals $[r_i,c_1]_{P,\tau'}$ (and $[r_i,c_2]_{P,\tau'}$). 
    If we consider $A$ as an assignment of $I$, then the intervals of used tics are $[r_i,c_1]_{P,\tau}$ (resp. $[r_i,c_2]_{P,\tau}$). 
    These intervals are strictly included in $[r_i,c_1]_{P,\tau'}$ (resp. $[r_i,c_2]_{P,\tau'}$), hence they do not have intersection either. 
    \end{proof}

     Lemma~\ref{lemma:smaller_tau} gives a way to obtain a solution of an instance from the same instance with a larger message size, as illustrated in Figure~\ref{fig:space}. 
     This transformation guarantees that all datagrams are separated by at least $\tau' - \tau$ free tics in each contention point. We are interested in finding the maximal $\tau'$ for which there is an assignment. Since the property of having an assignment is monotone with regard to $\tau'$, we can do so by a dichotomic search on $\tau$.

           \begin{figure}
       \begin{center}
      \includegraphics[width = 0.8\textwidth]{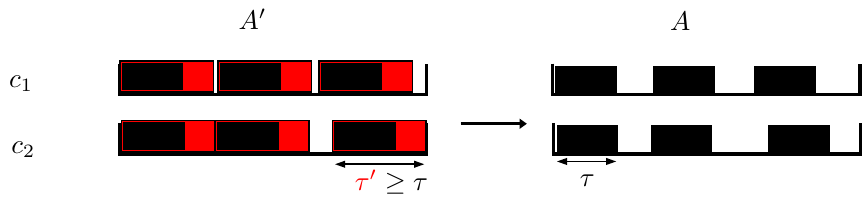}
      \end{center} 
      \caption{A $(P,\tau')$-assignment interpreted as a $(P,\tau)$ assignment}
      \label{fig:space}   
     \end{figure}

    We call \SPMLS, for \textbf{S}paced \PMLS, the adaptation of \PMLS which finds an assignment for the largest possible $\tau$ by dichotomic search on $\tau$. We experimentally investigate how large $\tau'$ can be so that \SPMLS finds a $(P,\tau')$ assignment. In Figure~\ref{fig:spacetau}, we represent the probability of finding a $(P,\tau')$ assignment function of $\tau'$. The star-routed networks are generated as in Section~\ref{sec:resultsPALL}, with $8$ routes having arc lengths drawn from $[P]$. The network has a load of $0.6$ of C-RAN traffic, hence the period is set to $33,333$ for $\tau = 2500$. The network is less loaded with C-RAN traffic than in the previous sections because it will also support stochastic traffic, incurring an additional load.

    \begin{figure}
       \begin{center}
      \includegraphics[width = 0.8\textwidth]{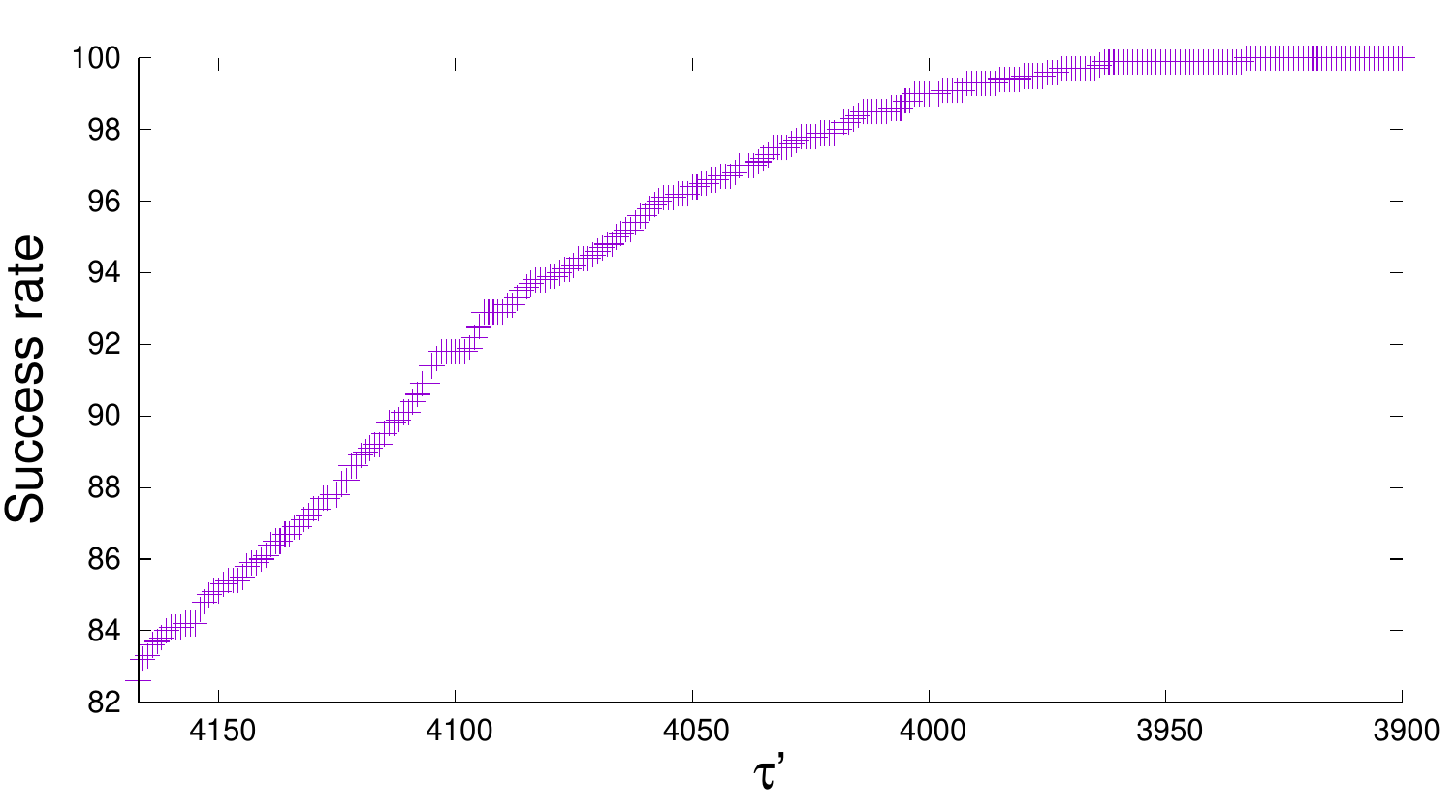}
      \end{center}
      \caption{Probability of finding a $(P,\tau')$-assignment over $10,000$ instances}
      \label{fig:spacetau}   
     \end{figure}   

	For more than $80\%$ of the instances, there is an assignment for the maximal size of a datagram $\tau' = \frac{P}{n} = 4166$. This means that \SPMLS perfectly balances the free tics in the period. In the worst case, a solution with $\tau' = 3925$ is found, which still yields at least $3925 - 2500 = 1425$ unused tics between datagrams. Hence, we expect \SPMLS to work well in conjunction with random traffic. The excellent performance of \PMLS, when the load is high, explains this result and further justifies the work we have done to solve \pall efficiently under high loads rather than just requiring mild loads in applications.

    \subsubsection{Performance Evaluation}
    
    We evaluate in this section different ways to manage statistical and deterministic traffic together in the same network.
 
    \paragraph{Best-effort datagrams generation}
  Let us denote best-effort by BE. The BE traffic is generated as follows. The size of a BE datagram is small in practice and is set to $50$ tics in our experiments. We generate random BE traffic such that it adds on average $0.2$ to the load to achieve a total load of $0.8$. BE datagrams do not go back and forth in the network like C-RAN datagrams; they pass through a single contention point.
We thus independently generate BE datagrams for each of the two contention points $c_1$ and $c_2$. The \textbf{latency} of a BE datagram is defined as the time it is buffered in its contention point.

At each contention point, the generation is split into two exponential distributions, which give the time before the next arrival of datagrams~\cite{el2018performance}. The first one models background traffic, corresponding to an average load of $0.15$. It generates one BE datagram every $333$ tics on average. The second models a burst of BE datagrams and corresponds to an average load of $0.05$; it generates \emph{ten} BE datagrams at the same time every $10,000$ tics on average.

   	\paragraph{Statistical multiplexing policy}

    We try several policies to deal with all traffic using statistical multiplexing.
   	The BE traffic is managed using \FIFO, and we propose two policies to deal with C-RAN. First, all datagrams, BE or C-RAN, are stored in the same buffer and dealt with the \FIFO policy regardless of their type. We call this policy \FIFO.

    To minimize the latency of C-RAN traffic, we can store the two types of datagrams in two different buffers, each managed with \FIFO, but we prioritize the C-RAN datagrams, which are always sent first. It can be technically implemented using TSN 802.1Qbu~\cite{ieee802}, which allows defining priority classes in the traffic to schedule first the traffic with the highest priority, here the C-RAN traffic. We call this policy \framepre.
   
   We also consider C-RAN traffic scheduled by \PMLS or \SPMLS. In that case, we need to forbid the transit of a BE datagram that collides with a C-RAN datagram. Thus, in each contention point, we reserve $50$ tics (the size of a BE datagram) before the arrival of a C-RAN message. Observe that it wastes some resources and thus slightly decreases the maximal throughput and may worsen the latency of BE datagrams.
    
    Figure~\ref{fig:belatency} shows the cumulative distribution of the logical latency of BE datagrams, that is the probability that a BE datagram has a latency less than some value.
    The distribution is computed over $1000$ random instances, and for each instance, the traffic is simulated for ten periods.

     \begin{figure}
       \begin{center}
      \includegraphics[width = 1\textwidth]{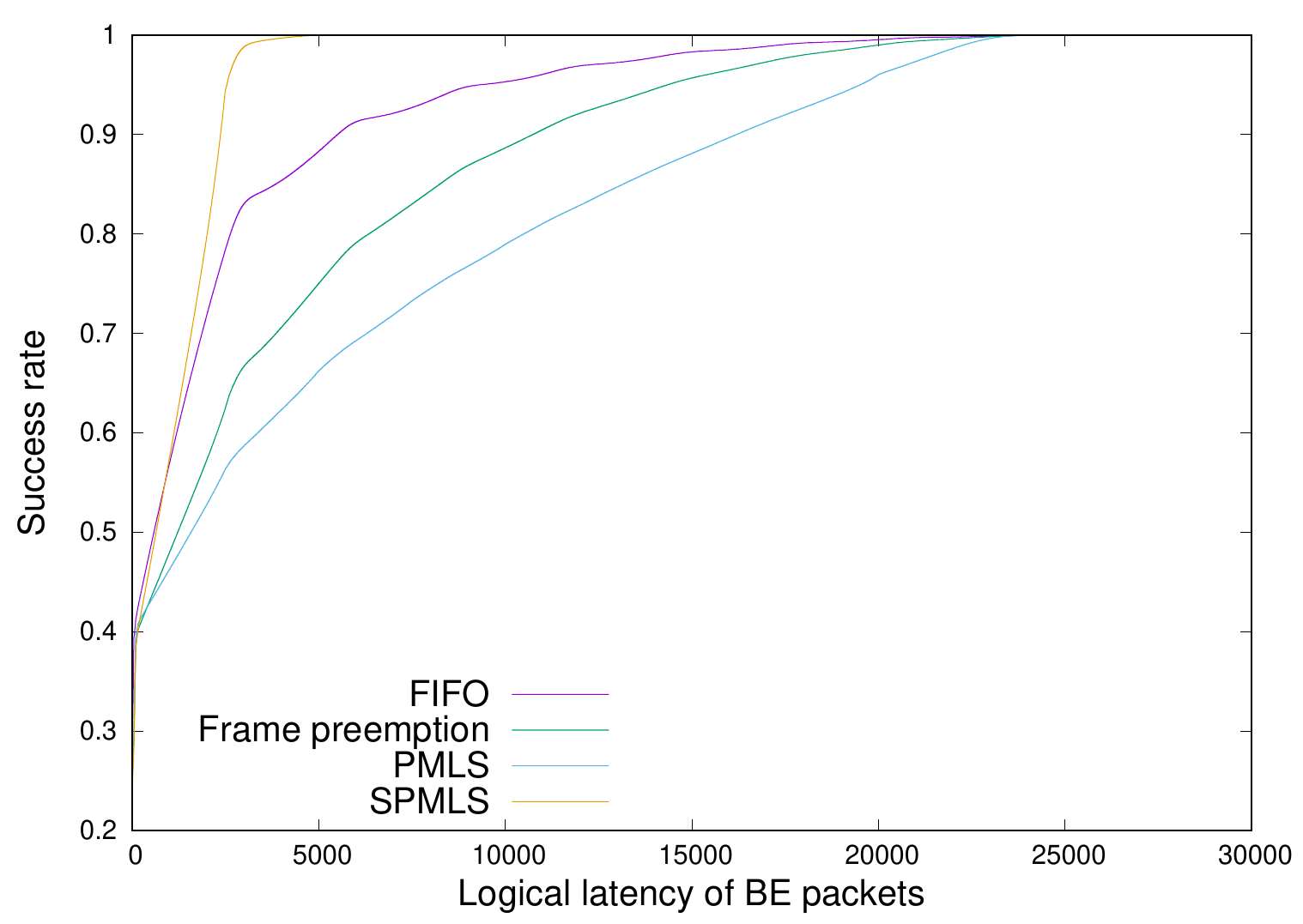}
      \end{center}
      \caption{Cumulative distribution of the latency of BE datagrams for several network management schemes}
      \label{fig:belatency}   
     \end{figure}

     If we compare \FIFO and \framepre, we see that the latency of BE datagrams is better ($1977$ tics on average) with \FIFO. It is expected since in \framepre the C-RAN datagrams are prioritized and thus 
     the latency of the BE datagrams is strictly worse, $3256$ tics on average. However, this is a trade-off with the margin of the C-RAN datagrams, which is strictly better for \framepre: $1919$ tics on average versus $5265$ tics for \FIFO. 

     Using a deterministic approach for C-RAN with \PMLS, the trade-off is even stronger:
      the C-RAN margin is down to $0$, but the BE traffic is more impacted, with a latency of $4909$ tics on average. This can be explained by both the reservation of tics to deal with the periodic sending scheme and the long sequences of C-RAN datagrams without free time in contention points.
     
     Using \SPMLS, the C-RAN traffic is smoothed over the period, to regularly leave some free tics for BE traffic. By construction, we still have a C-RAN margin of $0$ but it improves the latency of BE datagrams to $949$ tics on average, which is even better than with \FIFO. 
     
      This result shows that managing deterministic traffic deterministically is also good for the stochastic sources of traffic in the network. We have already observed such a phenomenon in~\cite{DBLP:conf/ondm/BarthGS19}, a similar problem on an optical ring.

 \section{Conclusion}
 This article proposed two types of deterministic scheduling schemes to achieve low-latency periodic communication between BBUs and RRHs in a fronthaul network. We demonstrated that finding such schemes is $\NP$-complete, even for simple networks with width or depth of 2. Consequently, we focused on solving them in star-routed networks, which, despite their simplicity, effectively model practical fronthaul architectures.

The first scheme eliminates buffering and incurs no additional latency. Such solutions exist for short routes (using \shortestlongest) or when the load is below 0.8 with at most 20 routes (using \ESCA). For higher loads, we introduced \PMLS, which permits buffering at BBUs while maintaining minimal logical latency. Comparative analysis showed that \PMLS performs on par with \ASPMLS, which relies on an optimal FPT subroutine rather than a nearly linear-time heuristic.

Our deterministic approach significantly outperforms statistical multiplexing across all network loads, even in the presence of random traffic. Notably, statistical multiplexing induces excessive latency, even in low-load scenarios, whereas our method achieves zero logical latency in most cases, even under stringent conditions. This highlights the principle that \emph{deterministic traffic is best managed deterministically}.

Several challenges remain before practical deployment in C-RAN fronthaul networks. Theoretical confirmation of the $\NP$-hardness of \pazl and \pall on star-routed networks is needed. Additionally, designing a more efficient FPT algorithm for \pall, comparable to the one for \pazl, would clarify any performance trade-offs of using \PMLS heuristically. Further, extending our study to other fronthaul topologies—such as caterpillars, trees, cycles, and bounded-treewidth graphs—remains an open direction~\cite{DBLP:conf/ondm/BarthGS19, guiraud2021deterministic}.

Beyond C-RAN, variations of our model could capture broader use cases. Allowing datagrams of varying sizes would better represent diverse payloads but would render \PMLS inapplicable due to the increased complexity. Supporting links with different speeds requires precise modeling of inter-link interfaces, potentially addressable through multiprocessor scheduling techniques~\cite{simons1989fast}. Alternative optimization objectives, such as minimizing average rather than worst-case latency, might simplify the problem, making it solvable via linear programming. Introducing preemption—splitting datagrams into smaller packets—could further reduce latency. Additionally, exploring pseudo-periodic scheduling (periodicity over multiple cycles) or dynamically computing routes to minimize $TR(A)$ presents further research opportunities.

\section*{Declarations} 
\subsection*{Ethical Approval} Not applicable
\subsection*{Competing interests} All authors declare that they have no competing interest.
\subsection*{Author's Contributions} The manuscript has been written by all authors. Maël Guiraud prepared figures. All authors reviewed the manuscript.
\subsection*{Funding} This work has been partially supported by the French ANR project N-GREEN.
\subsection*{Availability of data and materials }  For all experiments of this paper, the code in C is available on the web page of one author~\cite{webpage} under a copyleft license. All the data are available on~\cite{datas}.

 	\paragraph*{Acknowledgments} 
 	We thank Olivier Marcé and Brice Leclerc for introducing us to the problem from a practical perspective. We also thank Christian Cadéré and David Auger for their friendly discussions on the subject and insightful remarks. 

\section*{Glossary of Acronyms}
Common acronyms:
\begin{longtable}{|l|p{10cm}|}
    \hline
    \textbf{Acronym} & \textbf{Definition} \\
    \hline
    BBU & Base Band Unit \\
    BE & Best Effort \\
    FIFO & First In, First Out \\
    FPT & Fixed-Parameter Tractable \\
    NP & Non-deterministic Polynomial-time \\
    RAN & Radio Access Network \\
    RRH & Remote Radio Head \\
    TSN & Time-Sensitive Networking \\
    \hline
\end{longtable}
Problems and algorithms acronyms:
\begin{longtable}{|l|p{10cm}|}
    \hline
    \textbf{Acronym} & \textbf{Definition} \\
    \hline
    \wta & Waiting Time Assignment \\
    \pall & Periodic Assignment for Low Latency \\
    \pazl & Periodic Assignment for Zero Latency \\
     \MLS & Minimal Latency Scheduling \\
     \PMLS & Periodic Minimal Latency Scheduling \\
    \ASPMLS & All Subset Periodic Minimal Latency Scheduling \\
    \SPMLS & Spaced Periodic Minimal Latency Scheduling \\
    \hline
\end{longtable}
\bibliography{Sources}

\end{document}